\documentclass[acmsmall,authorversion,nonacm]{acmart}

\AtBeginDocument{%
  \providecommand\BibTeX{{%
    \normalfont B\kern-0.5em{\scshape i\kern-0.25em b}\kern-0.8em\TeX}}}

\usepackage{multirow}
\usepackage{booktabs}
\usepackage[subrefformat=parens,labelformat=parens]{subfig}
\usepackage{bm}
\usepackage{balance}
\usepackage{tikz}
\usepackage{colortbl}
\usepackage{subfig}
\usepackage{graphicx}
\usetikzlibrary{positioning}
\usepackage{enumitem}
\usepackage{amsthm}
\usepackage{amsmath}
\usepackage{amsfonts}
\usepackage{bbm}
\usepackage{apptools}
\usepackage{xcolor}
\usepackage{placeins}

\usepackage{xcolor}
\definecolor{red}  {rgb}{0.9,0.0,0.0}
\definecolor{green}{rgb}{0.0,0.7,0.0}
\definecolor{blue} {rgb}{0.0,0.0,0.9}

\usepackage{acronym}
\acrodef{ltr}[\textsc{LtR}]{Learning to Rank}
\acrodef{nir}[NIR]{Neural Information Retrieval}
\acrodef{ir}[IR]{Information Retrieval}
\acrodef{bow}[\textsc{BoW}]{Bag of Words}
\acrodef{rrf}[RRF]{Reciprocal Rank Fusion}

\begin{document}

\title{An Analysis of Fusion Functions for Hybrid Retrieval}

\author{Sebastian Bruch}
\email{sbruch@acm.org}
\orcid{0000-0002-2469-8242}
\affiliation{%
  \institution{Pinecone}
  \city{New York}
  \state{NY}
  \country{USA}
}

\author{Siyu Gai}
\authornote{Contributed to this work during a research internship at Pinecone.}
\email{catherine_gai@berkeley.edu}
\orcid{0000-0002-5912-5196}
\affiliation{%
  \institution{University of California, Berkeley}
  \city{Berkeley}
  \state{CA}
  \country{USA}
}

\author{Amir Ingber}
\affiliation{%
  \institution{Pinecone}
  \city{Tel Aviv}
  \country{Israel}}
\email{ingber@pinecone.io}
\orcid{0000-0001-6639-8240}

\ccsdesc[500]{Information systems~Retrieval models and ranking}
\ccsdesc[500]{Information systems~Combination, fusion and federated search}

\keywords{Hybrid Retrieval, Lexical and Semantic Search, Fusion Functions}

\begin{abstract}
We study hybrid search in text retrieval where lexical and semantic search are \emph{fused} together with the intuition that the two are complementary in how they model relevance.
In particular, we examine fusion by a convex combination (CC) of lexical and semantic scores, as well as the Reciprocal Rank Fusion (RRF) method, and identify their advantages and potential pitfalls.
Contrary to existing studies, we find RRF to be sensitive to its parameters; that the learning of a CC fusion is generally agnostic to the choice of score normalization; that CC outperforms RRF in in-domain and out-of-domain settings; and finally, that CC is sample efficient, requiring only a small set of training examples to tune its only parameter to a target domain.
\end{abstract}

\maketitle              

\section{Introduction}
\label{section:introduction}

Retrieval is the first stage in a multi-stage ranking system~\cite{asadi2013phd,asadi2013efficiency,yin2016ranking}, where the objective is to find the top-$k$ set of documents, that are the most relevant to a given query $q$, from a large collection of documents $\mathcal{D}$. Implicit in this task are two major research questions: a) how do we measure the relevance between a query $q$ and a document $d \in \mathcal{D}$?; and, b) how do we find the top-$k$ documents according to a given similarity metric \emph{efficiently}? In this work, we are primarily concerned with the former question in the context of text retrieval.

As a fundamental problem in \ac{ir}, the question of the similarity between queries and documents has been explored extensively. Early methods model text as a \ac{bow} and compute the similarity of two pieces of text using a statistical measure such as the term frequency-inverse document frequency (TF-IDF) family, with BM25~\cite{bm25,bm25original} being its most prominent member. We refer to retrieval with a \ac{bow} model as \emph{lexical search} and the similarity scores computed by such a system as \emph{lexical scores}.

Lexical search is simple, efficient, (naturally) ``zero-shot,'' and generally effective, but has important limitations: it is susceptible to the vocabulary mismatch problem and, moreover, does not take into account the semantic similarity of queries and documents~\cite{chen2022ecir}. That, it turns out, is what deep learning models are excellent at. With the rise of pre-trained language models such as BERT~\cite{devlin2019bert}, it is now standard practice to learn a vector representation of queries and documents that does capture their semantics, and thereby, reduce top-$k$ retrieval to the problem of finding $k$ nearest neighbors in the resulting vector space~\cite{formal2022splade,karpukhin-etal-2020-dense,lin2021pretrained,reimers-2019-sentence-bert,xiong2021approximate,wang2021bert}---where closeness is measured using vector similarity or distance. We refer to this method as \emph{semantic search} and the similarity scores computed by such a system as \emph{semantic scores}.

Hypothesizing that lexical and semantic search are complementary in how they model relevance, recent works~\cite{chen2022ecir,karpukhin-etal-2020-dense,Kuzi2020LeveragingSA,Ma2020HybridFR,Ma2021ARS,Wu2019EfficientIP} began exploring methods to \emph{fuse} together lexical and semantic retrieval: for a query $q$ and ranked lists of documents $R_\textsc{Lex}$ and $R_\textsc{Sem}$ retrieved separately by lexical and semantic search systems respectively, the task is to construct a final ranked list $R_\textsc{Fusion}$ so as to improve retrieval quality. This is often referred to as \emph{hybrid search}.

It is becoming increasingly clear that hybrid search does indeed lead to meaningful gains in retrieval quality, especially when applied to out-of-domain datasets~\cite{chen2022ecir,wang2021bert}---settings in which the semantic retrieval component uses a model that was not trained or fine-tuned on the target dataset. What is less clear and is worthy of further investigation, however, is \emph{how} this fusion is done.

One intuitive and common approach is to linearly combine lexical and semantic scores~\cite{karpukhin-etal-2020-dense,Ma2021ARS,wang2021bert}. If $f_\textsc{Lex}(q, d)$ and $f_\textsc{Sem}(q, d)$ represent the lexical and semantic scores of document $d$ with respect to query $q$, then a linear (or more accurately, convex) combination is expressed as $f_\textsc{Convex} = \alpha f_\textsc{Sem} + (1 - \alpha) f_\textsc{Lex}$ where $0 \leq \alpha \leq 1$. Because lexical scores (such as BM25) and semantic scores (such as dot product) may be unbounded, often they are normalized with min-max scaling~\cite{lin2021pretrained,wang2021bert} prior to fusion.

A recent study~\cite{chen2022ecir} argues that convex combination is sensitive to its parameter $\alpha$ and the choice of score normalization.\footnote{c.f. Section 3.1 in~\cite{chen2022ecir}: ``This fusion method is sensitive to the score scales \ldots which needs \emph{careful} score normalization'' (emphasis ours). } They claim and show empirically, instead, that \ac{rrf}~\cite{cormack2009rrf} may be a more suitable fusion as it is non-parametric and may be utilized in a zero-shot manner. They demonstrate its impressive performance even in zero-shot settings on a number of benchmark datasets.

This work was inspired by the claims made in~\cite{chen2022ecir}; whereas~\cite{chen2022ecir} addresses \emph{how} various hybrid methods perform relative to one another in an empirical study, we re-examine their findings and analyze \emph{why} these methods work and what contributes to their relative performance. 
Our contributions thus can best be summarized as an in-depth examination of fusion functions and their behavior.

As our first research question (RQ1), we investigate whether the convex combination fusion is a reasonable choice and study its sensitivity to the normalization protocol. We show that, while normalization is essential to create a bounded function and thereby bestow consistency to the fusion across domains, the specific choice of normalization is a rather small detail: there always exist convex combinations of scores normalized by min-max, standard score, or any other linear transformation that are rank-equivalent. In fact, when formulated as a per-query learning problem, the solution found for a dataset that is normalized with one scheme can be transformed to a solution for a different choice.

We next investigate the properties of \ac{rrf}. We first unpack \ac{rrf} and examine its sensitivity to its parameters as our second research question (RQ2)---contrary to~\cite{chen2022ecir}, we adopt a parametric view of \ac{rrf} where we have as many parameters as there are retrieval functions to fuse, a quantity that is always one more than that in a convex combination. We find that, in contrast to a convex combination, a tuned \ac{rrf} generalizes poorly to out-of-domain datasets.
We then intuit that, because \ac{rrf} is a function of \emph{ranks}, it disregards the distribution of scores and, as such, discards useful information. Observe that the distance between raw scores plays no role in determining their hybrid score---a behavior we find counter-intuitive in a metric space where distance \emph{does} matter. Examining this property constitutes our third and final research question (RQ3).

Finally, we empirically demonstrate an unsurprising yet important fact: tuning $\alpha$ in a convex combination fusion function is extremely sample-efficient, requiring just a handful of labeled queries to arrive at a value suitable for a target domain, regardless of the magnitude of shift in the data distribution. \ac{rrf}, on the other hand, is relatively less sample-efficient and converges to a relatively less effective retrieval system.

We believe our findings, both theoretical and empirical, are important and pertinent to the research in this field. Our analysis leads us to believe that the convex combination formulation is theoretically sound, empirically effective, sample-efficient, and robust to domain shift. Moreover, unlike the parameters in \ac{rrf}, the parameter(s) of a convex function are highly interpretable and, if no training samples are available, can be adjusted to incorporate domain knowledge.

We organized the remainder of this article as follows. In Section~\ref{section:background}, we review the relevant literature on hybrid search. Section~\ref{section:setup} then introduces our adopted notation and provides details of our empirical setup, thereby providing context for the theoretical and empirical analysis of fusion functions. In Section~\ref{section:analysis:convex}, we begin our analysis by a detailed look at the convex combination of retrieval scores. We then examine \ac{rrf} in Section~\ref{section:analysis:rrf}. In Section~\ref{section:discussion}, we summarize our observations and identify the properties a fusion function should have to behave well in hybrid retrieval. We then conclude this work and state future research directions in Section~\ref{section:conclusion}.
\section{Related Work}
\label{section:background}

A multi-stage ranking system is typically comprised of a \emph{retrieval} stage and several subsequent \emph{re-ranking} stages, where the retrieved candidates are ordered using more complex ranking functions~\cite{asadi2013efficiency,wang2011cascade}. Conventional wisdom has that retrieval must be recall-oriented while improving ranking quality may be left to the re-ranking stages, which are typically \ac{ltr} models~\cite{LTR-Liu09,matsubara2020multistage,nogueira2019multi,wang2011cascade,lambdamart}. There is indeed much research on the trade-offs between recall and precision in such multi-stage cascades~\cite{macdonald2013whens,ecir13}, but a recent study~\cite{Zamani2022cascade} challenges that established convention and presents theoretical analysis that suggests retrieval must instead optimize \emph{precision}. We therefore report both recall \emph{and} NDCG~\cite{jarvelin2000ir}, but focus on NDCG where space constraints prevent us from presenting both or when similar conclusions can be reached regardless of the metric used.

One choice for retrieval that remains popular to date is BM25~\cite{bm25,bm25original}. This additive statistic computes a weighted lexical match between query and document terms: it computes, for each query term, the product of its ``importance'' (i.e., frequency of a term in a document, normalized by document and global statistics such as average length) and its propensity---a quantity that is inversely proportionate to the fraction of documents that contain the term---and adds the scores of query terms to arrive at the final similarity or relevance score. Because BM25, like other lexical scoring functions, insists on an exact match of terms, even a slight typo can throw the function off. This \emph{vocabulary mismatch problem} has been subject to much research in the past, with remedies ranging from pseudo-relevance feedback to document and query expansion techniques~\cite{li2022prf,nogueira2019document,tao2006expansion}.

Trying to address the limitations of lexical search can only go so far, however. After all, they additionally do not capture the semantic similarity between queries and documents, which may be an important signal indicative of relevance. It has been shown that both of these issues can be remedied by Transformer-based~\cite{vaswani2017attention} pre-trained language models such as BERT~\cite{devlin2019bert}. Applied to the ranking task, such models~\cite{nogueira2020passage,nogueira2020monot5,nogueira2019multi,mitra2016dual} have advanced the state-of-the-art dramatically on benchmark datasets~\cite{nguyen2016msmarco}.

The computationally intensive inference of these deep models often renders them too inefficient for first-stage retrieval, however, making them more suitable for re-ranking stages. But by cleverly disentangling the query and document transformations into the so-called dual-encoder architecture, where, in the resulting design, the ``embedding'' of a document can be computed independently of queries, we can pre-compute document vectors and store them offline. In this way, we substantially reduce the computational cost during inference as it is only necessary to obtain the vector representation of the query during inference. At a high level, these models project queries and documents onto a low-dimensional vector space where semantically-similar points stay closer to each other. By doing so we transform the retrieval problem to one of similarity search or Approximate Nearest Neighbor (ANN) search---the $k$ nearest neighbors to a query vector are the desired top-$k$ documents. This ANN problem can be solved efficiently using a number of algorithms such as FAISS~\cite{Johnson2021faiss} or Hierarchical Navigable Small World Graphs~\cite{malkov2016hnsw} available as open source packages or through a managed service such as Pinecone\footnote{\url{http://pinecone.io}}, creating an opportunity to use deep models and vector representations for first-stage retrieval~\cite{karpukhin-etal-2020-dense,xiong2021approximate}---a setup that we refer to as semantic search.

Semantic search, however, has its own limitations. Previous studies~\cite{chen2022ecir,thakur2021beir} have shown, for example, that when applied to out-of-domain datasets, their performance is often worse than BM25. Observing that lexical and semantic retrievers can be complementary in the way they model relevance~\cite{chen2022ecir}, it is only natural to consider a hybrid approach where lexical and semantic similarities both contribute to the makeup of final retrieved list. To date there have been many studies~\cite{karpukhin-etal-2020-dense,Kuzi2020LeveragingSA,luan-etal-2021-sparse,Ma2020HybridFR,Ma2021ARS,wang2021bert,Wu2019EfficientIP,zhan2020repbert} that do just that, where most focus on in-domain tasks with one exception~\cite{chen2022ecir} that considers a zero-shot application too. Most of these works only use one of the many existing fusion functions in experiments, but none compares the main ideas comprehensively. We review the popular fusion functions from these works in the subsequent sections and, through a comparative study, elaborate what about their behavior may or may not be problematic.

\section{Setup}
\label{section:setup}

In the sections that follow, we study fusion functions with a mix of theoretical and empirical analysis. For that reason, we present our notation as well as empirical setup and evaluation measures here to provide sufficient context for our arguments.

\subsection{Notation}
We adopt the following notation in this work. We use $f_\textsc{o}(q, d): \mathcal{Q}\times\mathcal{D} \rightarrow \mathbb{R}$ to denote the score of document $d \in \mathcal{D}$ to query $q \in \mathcal{Q}$ according to the retrieval system $\textsc{o} \in \mathcal{O}$. If $\textsc{o}$ is a semantic retriever, $\textsc{Sem}$, then $\mathcal{Q}$ and $\mathcal{D}$ are the space of (dense) vectors in $\mathbb{R}^d$ and $f_\textsc{Sem}$ is typically cosine similarity or inner product. Similarly, when $\textsc{o}$ is a lexical retriever, $\textsc{Lex}$, $\mathcal{Q}$ and $\mathcal{D}$ are high-dimensional sparse vectors in $\mathbb{R}^{|V|}$, with $|V|$ denoting the size of the vocabulary, and $f_\textsc{Lex}$ is typically BM25. A retrieval system $\textsc{o}$ is the space $\mathcal{Q} \times \mathcal{D}$ equipped with a metric $f_\textsc{o}(\cdot, \cdot)$---which need not be a proper metric.

We denote the set of top-$k$ documents retrieved for query $q$ by retrieval system $\textsc{o}$ by $R^{k}_\textsc{o}(q)$. We write $\pi_\textsc{o}(q, d)$ to denote the rank of document $d$ with respect to query $q$ according to retrieval system $\textsc{o}$. Note that, $\pi_\textsc{o}(q, d_i)$ can be expressed as the sum of indicator functions:
\begin{equation}
    \pi_\textsc{o}(q, d_i) = 1 + \sum_{d_j \in R^{k}_\textsc{o}(q)} \mathbbm{1}_{f_\textsc{o}(q, d_j) > f_\textsc{o}(q, d_i)},
    \label{equation:rank}
\end{equation}
where $\mathbbm{1}_c$ is $1$ when the predicate $c$ holds and $0$ otherwise. In words, and ignoring the subtleties introduced by the presence of score ties, the rank of document $d$ is the count of documents whose score is larger than the score of $d$.

Hybrid retrieval operates on the product space of $\prod \textsc{o}_i$ with metric $f_\textsc{Fusion}: \prod f_{\textsc{o}_i} \rightarrow \mathbb{R}$. Without loss of generality, in this work, we restrict $\prod \textsc{o}_i$ to be $\textsc{Lex} \times \textsc{Sem}$. That is, we only consider the problem of fusing two retrieval scores, but note that much of the anlysis can be trivially extended to the fusion of multiple retrieval systems. We refer to this hybrid metric as a \emph{fusion} function.

A fusion function $f_\textsc{Fusion}$ is typically applied to documents in the union of retrieved sets $\mathcal{U}^k(q) = \bigcup_\textsc{o} R^{k}_\textsc{o}(q)$, which we simply call \emph{the union set}. When a document $d$ in the union set is not present in one of the top-$k$ sets (i.e., $d \in \mathcal{U}^k(q)$ but $d \notin R^{k}_{\textsc{o}_i}(q)$ for some $\textsc{o}_i$), we compute its missing score (i.e., $f_{\textsc{o}_i}(q, d)$) prior to fusion.

\subsection{Empirical Setup}

\begin{table*}[t]
\caption{Datasets used in this work for evaluation purposes along with select statistics.}
\label{table:dataset-stats}
\begin{center}
\begin{sc}
\begin{tabular}{c|cc}
\toprule
Dataset & Document Count & Query Count \\
\midrule
MS MARCO Passage v1 & $8.8$M & $6{,}980$ \\
Natural Questions (NQ) & $2.68$M & $3{,}452$ \\
\textsc{Quora} & $523$K & $10{,}000$ \\
\textsc{NFCorpus} & $3.6$K & $323$ \\
\textsc{HotpotQA} & $5.23$M & $7{,}405$ \\
\textsc{Fever} & $5.42$M & $6{,}666$ \\
\textsc{SciFact} & $5$K & $300$ \\
\textsc{DBPedia} & $4.63$M & $400$ \\
\textsc{FiQA} & $57$K & $648$ \\
\bottomrule
\end{tabular}
\end{sc}
\end{center}
\end{table*}

\textbf{Datasets}: We evaluate our methods on a variety of publicly available benchmark datasets, summarized in Table~\ref{table:dataset-stats} both in in-domain and out-of-domain, zero-shot settings. One of the datasets is the MS MARCO\footnote{Available at \url{https://microsoft.github.io/msmarco/}} Passage Retrieval v1 dataset~\cite{nguyen2016msmarco}, a publicly available retrieval and ranking collection from Microsoft. It consists of roughly $8.8$ million short passages which, along with queries in natural language, originate from Bing. The queries are split into train, dev, and eval non-overlapping subsets. We use the train queries for any learning or tuning and evaluate exclusively on the small dev query set (consisting of $6{,}980$ queries) in our analysis. Included in the dataset also are relevance labels.

We additionally experiment with $8$ datasets from the BeIR collection~\cite{thakur2021beir}\footnote{Available at \url{https://github.com/beir-cellar/beir}}: Natural Questions (NQ, question answering), \textsc{Quora} (duplicate detection), \textsc{NFCorpus} (medical), \textsc{HotpotQA} (question answering), \textsc{Fever} (fact extraction), \textsc{SciFact} (scientific claim verification), \textsc{DBPedia} (entity search), and \textsc{FiQA} (financial). For a more detailed description of each dataset, we refer the reader to~\cite{thakur2021beir}.

\textbf{Lexical search}: We use PISA~\cite{MSMS2019} for keyword-based lexical retrieval. We tokenize queries and documents by space and apply stemming available in PISA---we do not employ any other preprocessing steps such as stopword removal, lemmatization, or expansion. We use BM25 with the same hyperparameters as~\cite{chen2022ecir} (k1=$0.9$ and b=$0.4$) to retrieve the top $1{,}000$ candidates.

\textbf{Semantic search}: We use the \emph{all-MiniLM-L6-v2} model checkpoint available on HuggingFace\footnote{Available at \url{https://huggingface.co/sentence-transformers/all-MiniLM-L6-v2}} to project queries and documents into $384$-dimensional vectors, which can subsequently be used for indexing and top-$k$ retrieval using cosine similarity. This model has been shown to achieve competitive quality on an array of benchmark datasets while remaining compact in size and efficient to infer\footnote{c.f. \url{https://sbert.net} for details.}, thereby allowing us to conduct extensive experiments with results that are competitive with existing state-of-the-art models. This model has been fine-tuned on a large number of datasets, exceeding a total of 1 billion pairs of text, including NQ, MS MARCO Passage, and \textsc{Quora}. As such, we consider all experiments on these three datasets as in-domain, and the rest as out-of-domain. We use the exact search for inner product algorithm (\textsc{IndexFlatIP}) from FAISS~\cite{Johnson2021faiss} to retrieve top $1{,}000$ nearest neighbors.

\textbf{Supplementary models and fusions}: Our primary set of experiments presented in the body of this manuscript focus exclusively on the fusion of BM25 with the \emph{all-MiniLM-L6-v2} model as representatives of lexical and semantic retrieval functions. We do, however, note that our analysis of fusion functions is not limited to lexical-semantic search \emph{per se}: all normalization and fusion functions studied in this work can be applied to arbitrary scoring functions! As such, we conduct additional experiments using a variety of pairs of retrieval models to confirm the generality of the main theoretical findings of this work. We report these results in Appendix~\ref{appendix:splade-bm25} through Appendix~\ref{appendix:tasb-minilm}.

In Appendix~\ref{appendix:splade-bm25}, we examine the fusion of the \textsc{Splade} model with BM25. \textsc{Splade}\footnote{Pre-trained checkpoint from HuggingFace available at \url{https://huggingface.co/naver/splade-cocondenser-ensembledistil}}~\cite{formal2022splade} is a deep learning model that produces sparse representations for a given piece of text, where each non-zero entry in the resulting embedding is the importance weight of a term in the BERT~\cite{devlin2019bert} WordPiece~\cite{wordpiece} vocabulary comprising of roughly $30,000$ terms. Appendix~\ref{appendix:tasb-bm25} studies the fusion of BM25 with the \textsc{Tas-B}~\cite{tas-b} model.\footnote{Available at \url{https://huggingface.co/sentence-transformers/msmarco-distilbert-base-tas-b}} \textsc{Tas-B} is a bi-encoder model that was trained using supervision from a cross-encoder and a ColBERT model. In Appendix~\ref{appendix:tasb-splade} we fuse \textsc{Splade} and \textsc{Tas-B}, and in Appendix~\ref{appendix:tasb-minilm} \textsc{Tas-B} and \emph{all-MiniLM-L6-v2}.

We note that, both \textsc{Splade} and \textsc{Tas-B} were fine-tuned on the MS MARCO dataset. As such, in all the supplementary experiments, results reported on the MS MARCO dataset should be considered ``in-domain'' while the remaining datasets represent out-of-domain distributions.

\textbf{Evaluation}: Unless noted otherwise, we form the union set for every query from the candidates retrieved by the lexical and semantic search systems. We then compute missing scores where required, compute $f_\textsc{Fusion}$ on the union set, and re-order according to the hybrid scores. We then measure Recall@$1000$ and NDCG@$1000$ to quantify ranking quality, as recommended by Zamani \emph{et al.}~\cite{Zamani2022cascade}. Due to the much smaller size of \textsc{SciFact} and \textsc{NFCorpus}, we evaluate Recall and NDCG at rank cutoff $100$ instead, retrieving roughly $2\%$ and $2.7\%$ of the size of the dataset, respectively. We note that, this choice of cutoff does not affect the outcome of our experiments or change our conclusions, but it more clearly highlights the differences between the various methods; recall approaches $1$ regardless of the retrieval method if rank cutoff was $1{,}000$ (or $20\%$ and $27\%$ of the size of the datasets). Further note that, we choose to evaluate deep (i.e., with a larger rank cut-off) rather than shallow metrics per the discussion in~\cite{wang2021bert} to understand the performance of each system more completely.

\section{Analysis of Convex Combination of Retrieval Scores}
\label{section:analysis:convex}

We are interested in understanding the behavior and properties of fusion functions. In the remainder of this work, we study through that lens two popular methods that are representative of existing ideas in the literature, beginning with a convex combination of scores.

As noted earlier, most existing works use a convex combination of lexical and semantic scores as follows: $f_\textsc{Convex}(q, d)= \alpha f_\textsc{Sem}(q, d) +  (1-\alpha) f_\textsc{Lex}(q, d)$ for some $0 \leq \alpha \leq 1$. When $\alpha=1$ the above collapses to semantic scores and when it is $0$, to lexical scores.

An interesting property of this fusion is that it takes into account the distribution of scores. In other words, the distance between lexical (or semantic) scores of two documents plays a significant role in determining their final hybrid score. One disadvantage, however, is that the range of $f_\textsc{Sem}$ can be very different from $f_\textsc{Lex}$. Moreover, as with TF-IDF in lexical search or with inner product in semantic search, the range of individual functions $f_\textsc{o}$ may depend on the norm of the query and document vectors (e.g., BM25 is a function of the number of query terms). As such any constant $\alpha$ is likely to yield inconsistently-scaled hybrid scores.

The problem above is trivially addressed by applying score normalization prior to fusion~\cite{lin2021pretrained,wang2021bert}. Suppose we have collected a union set $\mathcal{U}^k(q)$ for $q$, and that for every candidate we have computed both lexical and semantic scores. Now, consider the min-max scaling of scores $\phi_\textsc{mm}: \mathbb{R} \rightarrow [0, 1]$ below:
\begin{equation}
    \phi_\textsc{mm}(f_\textsc{o}(q,d)) = \frac{f_\textsc{o}(q, d) - m_q}{M_q - m_q},
    \label{equation:minmax}
\end{equation}
where $m_q = \min_{d \in \mathcal{U}^k(q)} f_\textsc{o}(q, d)$ and $M_q = \max_{d \in \mathcal{U}^k(q)} f_\textsc{o}(q, d)$. We note that, min-max scaling is the \emph{de facto} method in the literature, but other choices of $\phi_\textsc{o}(\cdot)$ in the more general expression below:
\begin{equation}
    f_\textsc{Convex}(q, d) = \alpha \phi_\textsc{Sem}(f_\textsc{Sem}(q, d)) + (1 - \alpha) \phi_\textsc{Lex}(f_\textsc{Lex}(q, d)),
    \label{equation:convex:general}
\end{equation}
are valid as well so long as $\phi_\textsc{Sem}, \phi_\textsc{Lex}: \mathbb{R} \rightarrow \mathbb{R}$ are monotone in their argument. For example, for reasons that will become clearer later, we can redefine the normalization by replacing the minimum of the set with the theoretical minimum of the function (i.e., the maximum value that is always less than or equal to all values attainable by the scoring function, or its infimum) to arrive at:
\begin{equation}
    \phi_\textsc{tmm}(f_\textsc{o}(q,d)) = \frac{f_\textsc{o}(q, d) - \inf f_\textsc{o}(q, \cdot)}{M_q - \inf f_\textsc{o}(q, \cdot)}.
    \label{equation:infmax}
\end{equation}
As an example, when $f_\textsc{Lex}$ is BM25, then its infimum is $0$. When $f_\textsc{Sem}$ is cosine similarity, then that quantity is $-1$.

Another popular choice is the standard score (z-score) normalization which is defined as follows:
\begin{equation}
    \phi_\textsc{z}(f_\textsc{o}(q,d)) = \frac{f_\textsc{o}(q, d) - \mu}{\sigma},
    \label{equation:zscore}
\end{equation}
where $\mu$ and $\sigma$ denote the mean and standard deviation of the set of scores $f_\textsc{o}(q, \cdot)$ for query $q$.

We will return to normalization shortly, but we make note of one small but important fact: in cases where the variance of lexical (semantic) scores in the union set is $0$, we may skip the fusion step altogether because retrieval quality will be unaffected by lexical (semantic) scores. The case where the variance is arbitrarily close to $0$, however, creates challenges for certain normalization functions. While this would make for an interesting theoretical analysis, we do not study this particular setting in this work as, empirically, we do observe a reasonably large variance among scores in the union set on all datasets using state-of-the-art lexical and semantic retrieval functions.

\subsection{Suitability of Convex Combination}
\label{section:analysis:convex:suitability}
A convex combination of scores is a natural choice for creating a mixture of two retrieval systems, but is it a reasonable choice? It has been established in many past empirical studies that $f_\textsc{Convex}$ with min-max normalization often serves as a strong baseline. So the answer to our question appears to be positive. Nonetheless, we believe it is important to understand precisely why this fusion works.

We investigate this question empirically, by visualizing lexical and semantic scores of query-document pairs from an array of datasets. Because we operate in a two-dimensional space, observing the pattern of positive (where document is relevant to query) and negative samples in a plot can reveal a lot about whether and how they are separable and how the fusion function behaves. To that end, we sample up to $20{,}000$ positive and up to the same number of negative query-document pairs from the validation split of each dataset, and illustrate the collected points in a scatter plot in Figure~\ref{figure:convex:suitability}.

\begin{figure*}[t]
\begin{center}
\centerline{
\subfloat[MS MARCO]{
\includegraphics[height=1.6in]{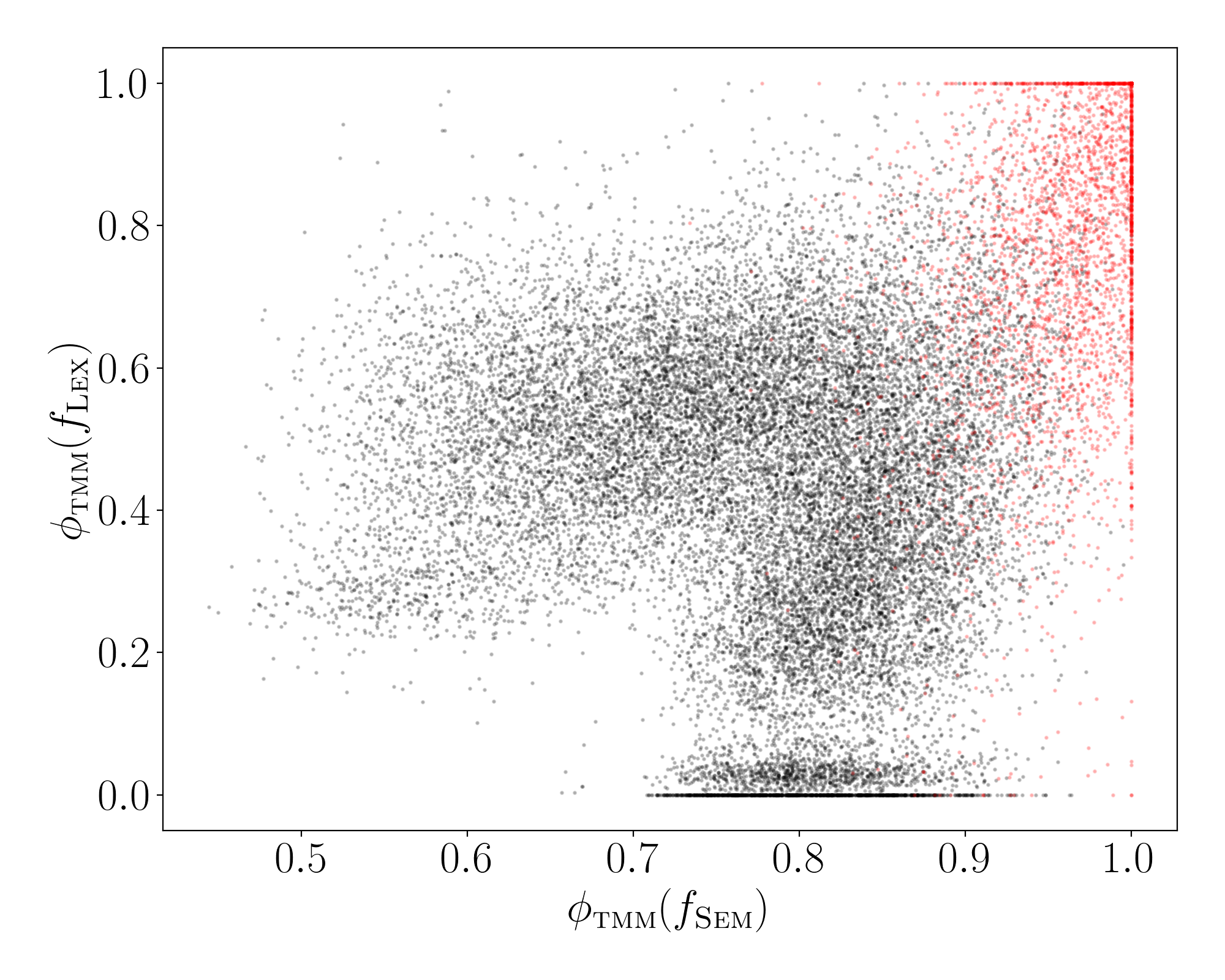}}
\subfloat[\textsc{Quora}]{
\includegraphics[height=1.6in]{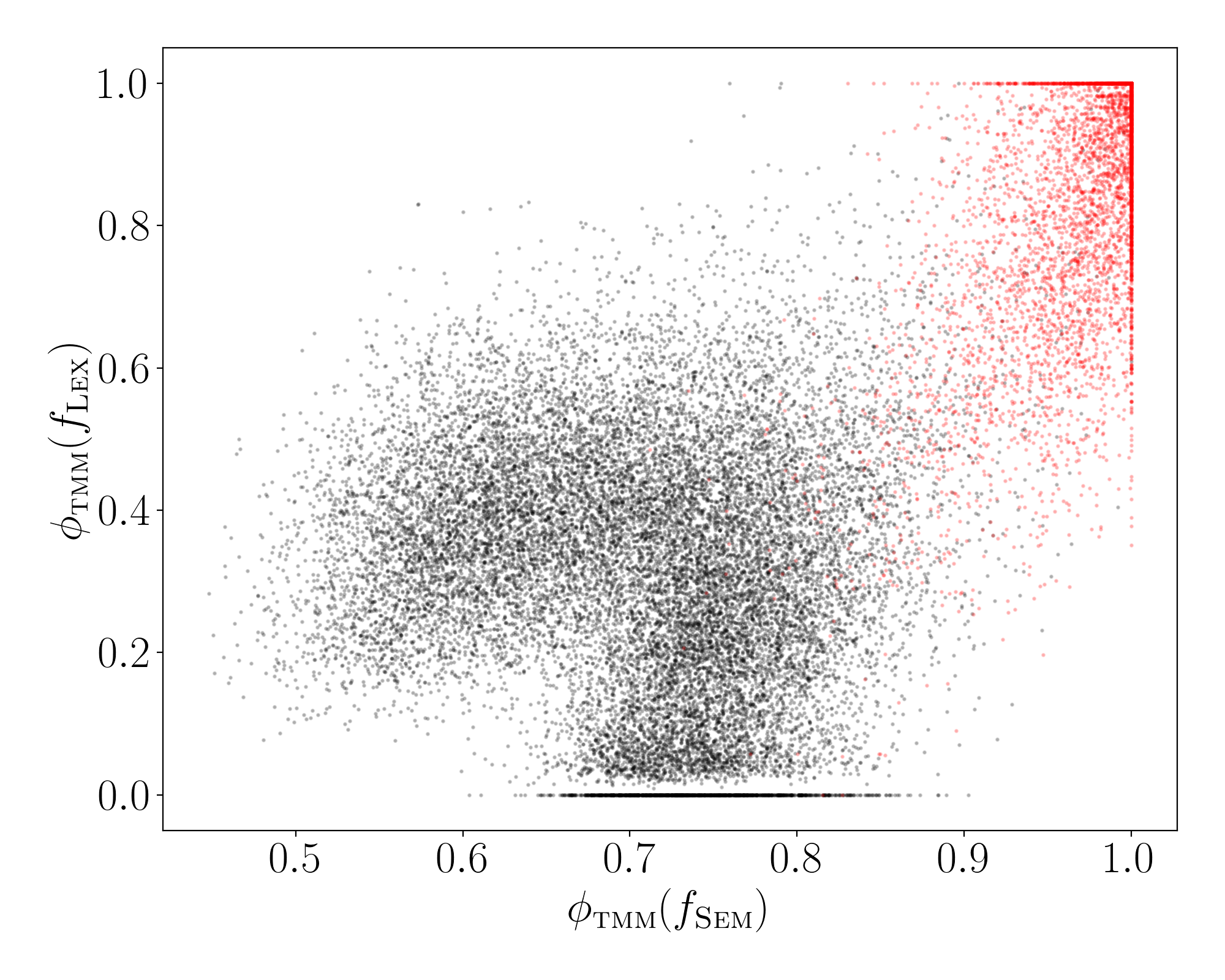}}
}
\centerline{
\subfloat[\textsc{NQ}]{
\includegraphics[height=1.6in]{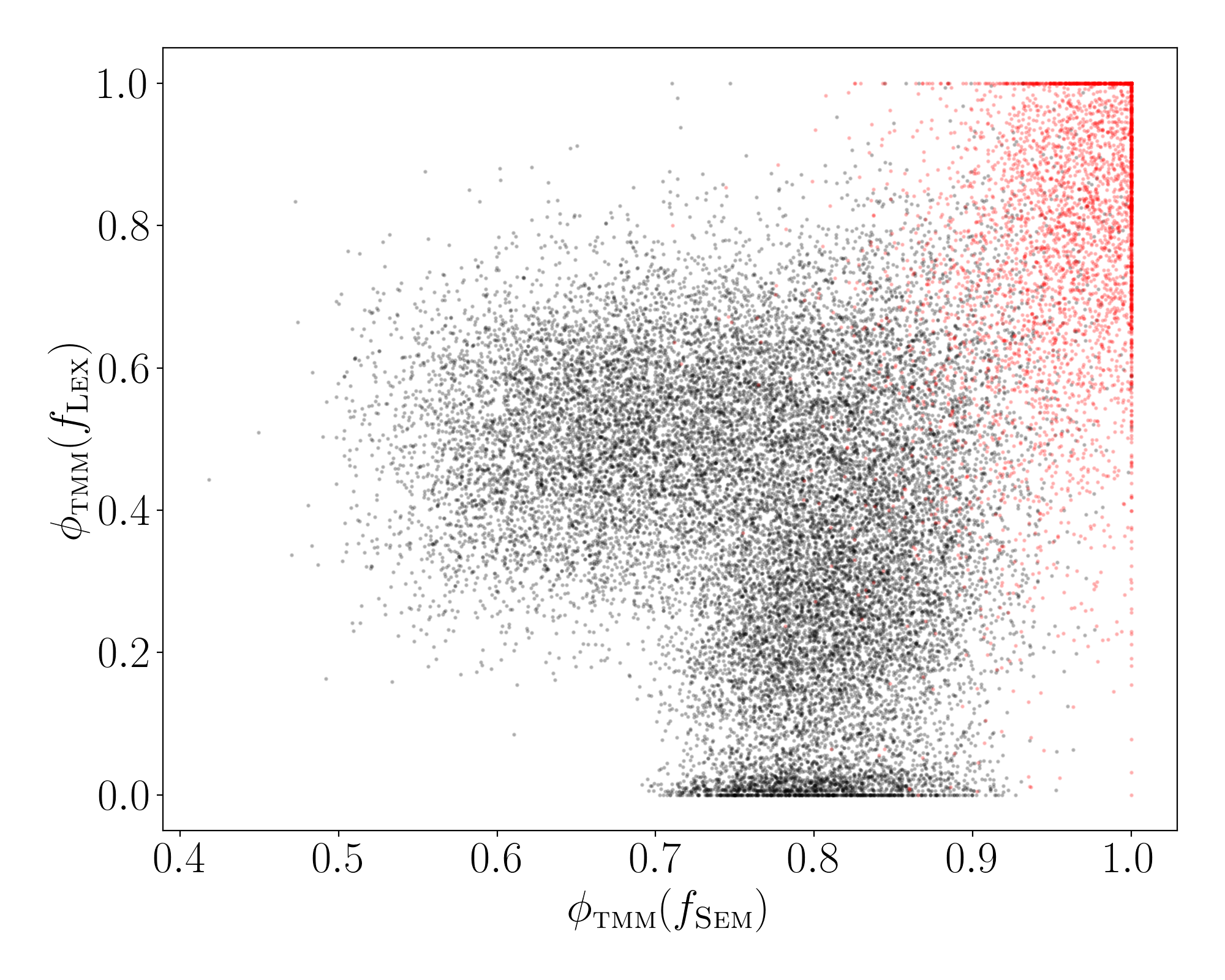}}
\subfloat[\textsc{FiQA}]{
\includegraphics[height=1.6in]{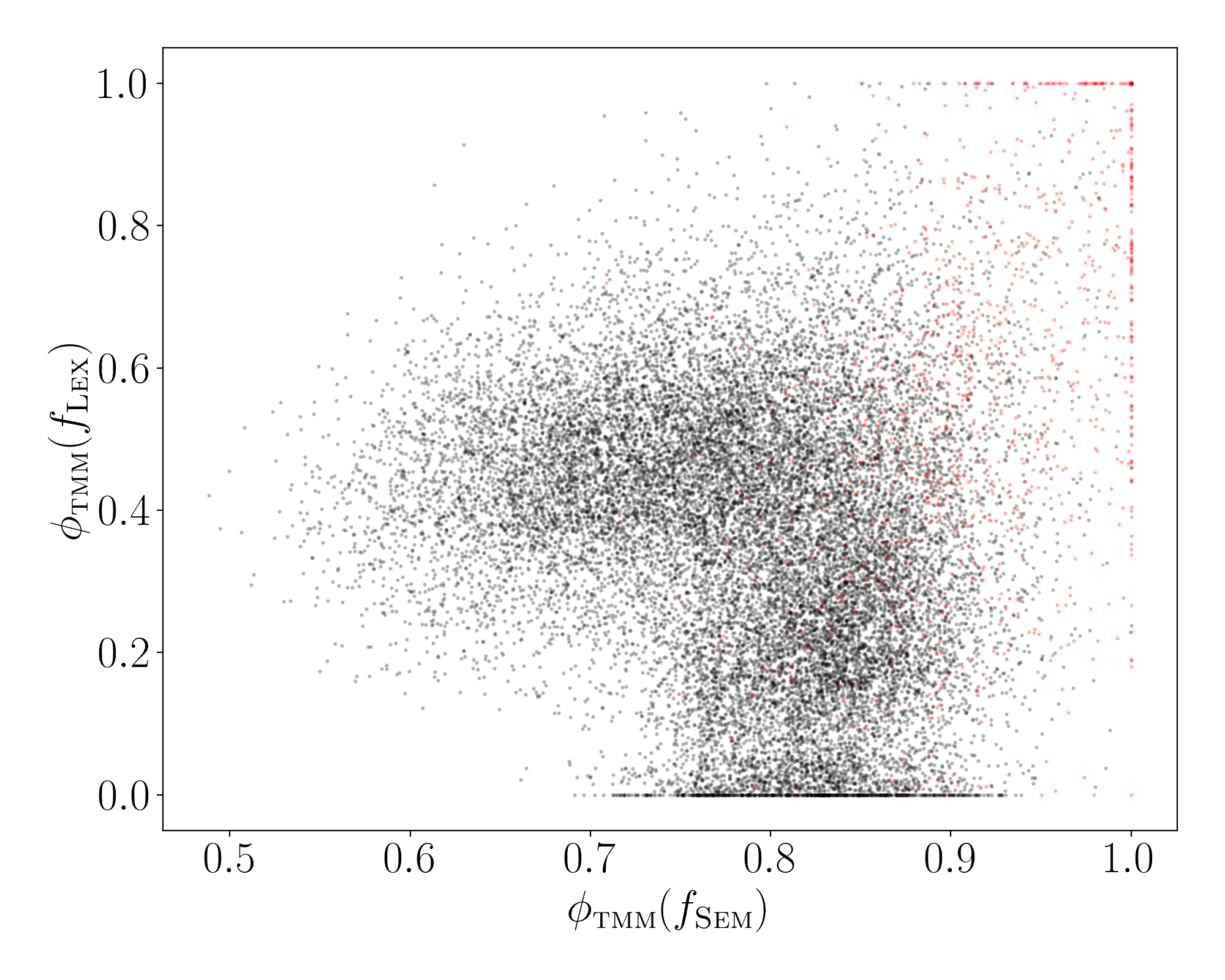}}
}
\centerline{
\subfloat[\textsc{HotpotQA}]{
\includegraphics[height=1.6in]{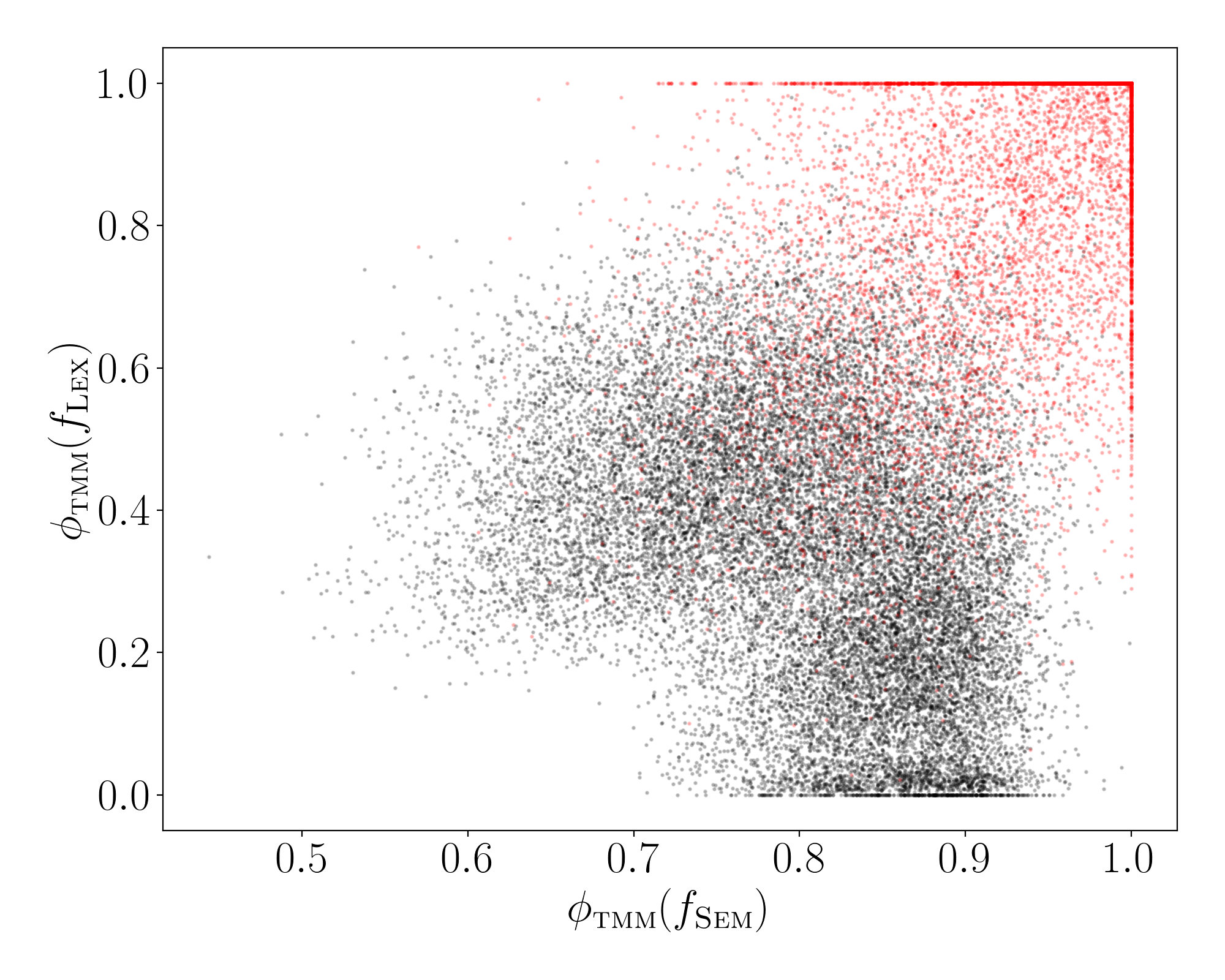}}
\subfloat[\textsc{Fever}]{
\includegraphics[height=1.6in]{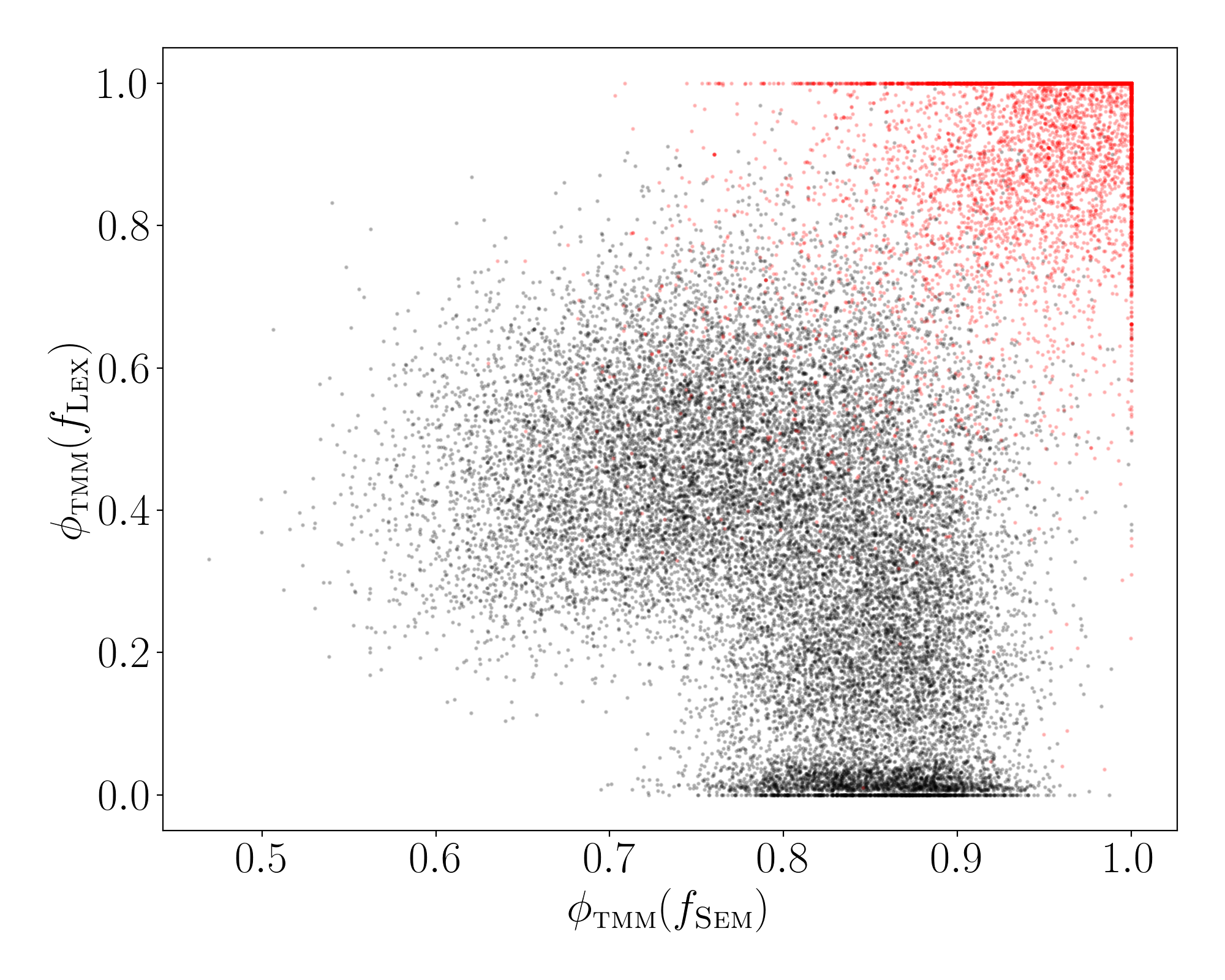}}
}
\caption{Visualization of the normalized lexical ($\phi_\textsc{tmm}(f_\textsc{Lex})$) and semantic ($\phi_\textsc{tmm}(f_\textsc{Sem})$) scores of query-document pairs sampled from the validation split of each dataset. Shown in red are up to $20{,}000$ positive samples where document is relevant to query, and in black up to the same number of negative samples. Adding a lexical (semantic) dimension to query-document pairs helps tease out the relevant documents that would be statistically indistinguishable in a one-dimensional semantic (lexical) view of the data---when samples are projected onto the $x$ ($y$) axis.}
\label{figure:convex:suitability}
\end{center}
\end{figure*}

From these figures, it is clear that positive and negative samples form clusters that are, with some error, separable by a linear function. What is different between datasets is the slope of this separating line. For example, in MS MARCO, \textsc{Quora}, and \textsc{NQ}, which are in-domain datasets, the separating line is almost vertical, suggesting that the semantic scores serve as a sufficiently strong signal for relevance. This is somewhat true of \textsc{FiQA}. In other out-of-domain datasets, however, the line is rotated counter-clockwise, indicating a more balanced weighting of lexical and semantic scores. Said differently, adding a lexical (semantic) dimension to query-document pairs helps tease out the relevant documents that would be statistically indistinguishable in a one-dimensional semantic (lexical) view of the data. Interestingly, across all datasets, there is a higher concentration of negative samples where lexical scores vanish.

This empirical evidence suggests that lexical and semantic scores may indeed be complementary---an observation that is in agreement with prior work~\cite{chen2022ecir}---and a line may be a reasonable choice for distinguishing between positive and negative samples. But while these figures shed light on the shape of positive and negative clusters and their separability, our problem is not classification but \emph{ranking}. We seek to \emph{order} query-document pairs and, as such, separability is less critical and, in fact, not required. It is therefore instructive to understand the effect of a particular convex combination on pairs of lexical and semantic scores. This is visualized in Figure~\ref{figure:convex:reference} for two values of $\alpha$ in $f_\textsc{Convex}$.

\begin{figure}[t]
\begin{center}
\centerline{
\subfloat[$\alpha=0.6$]{
\includegraphics[height=2in]{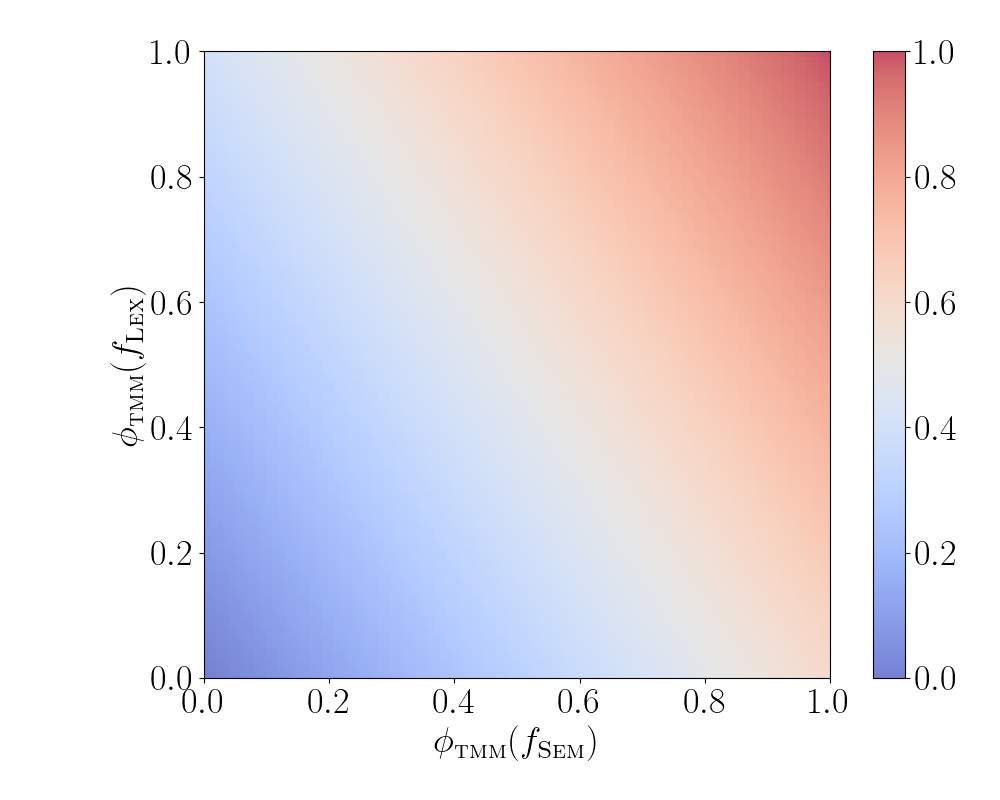}}
\subfloat[$\alpha=0.8$]{
\includegraphics[height=2in]{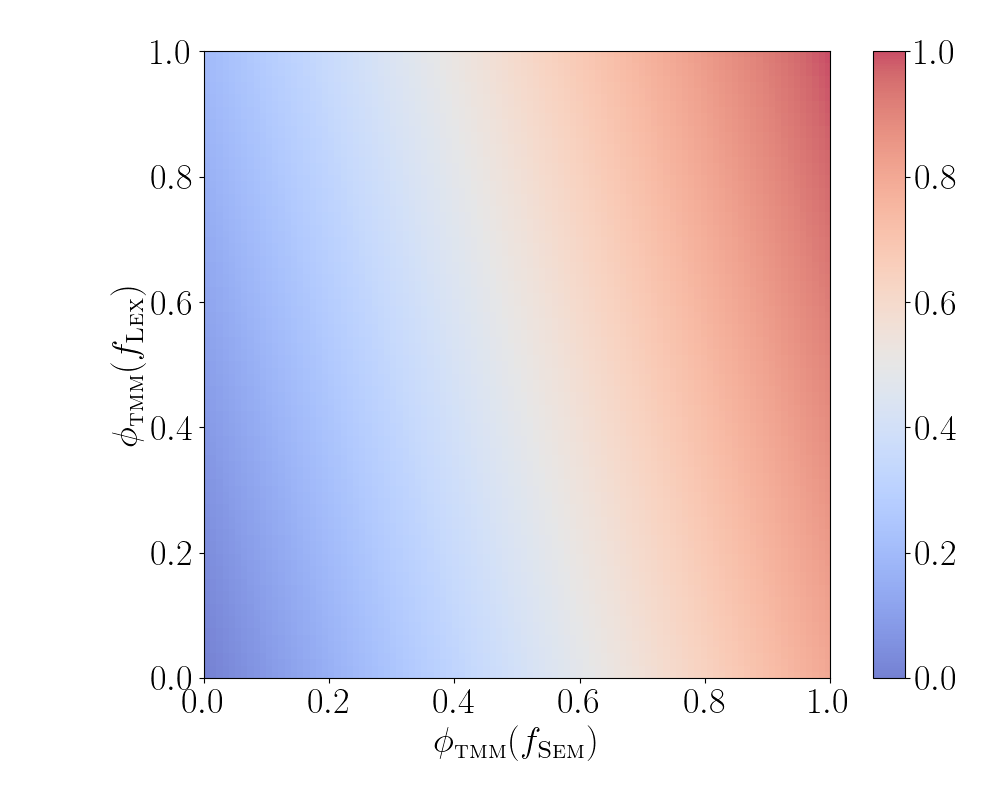}}
}
\caption{Effect of $f_\textsc{Convex}$ on pairs of lexical and semantic scores.}
\label{figure:convex:reference}
\end{center}
\end{figure}

The plots in Figure~\ref{figure:convex:reference} illustrate how the parameter $\alpha$ determines how different regions of the plane are ranked relative to each other. This is a trivial fact, but it is nonetheless interesting to map these patterns to the distributions in Figure~\ref{figure:convex:suitability}. In-domain datasets, for example, form a pattern of positives and negatives that is unsurprisingly more in tune with the $\alpha=0.8$ setting of $f_\textsc{Convex}$ than $\alpha=0.6$.

\subsection{Role of Normalization}

We have thus far used min-max normalization to be consistent with the literature. In this section, we ask the question first raised by Chen \emph{et al.}~\cite{chen2022ecir} on whether and to what extent the choice of normalization matters and how carefully one must choose the normalization protocol. In other words, we wish to examine the effect of $\phi_\textsc{Sem}(\cdot)$ and $\phi_\textsc{Lex}(\cdot)$ on the convex combination in Equation~(\ref{equation:convex:general}).

Before we begin, let us consider the following suite of functions:
\begin{itemize}
    \item $\phi_\textsc{mm}$: Min-max scaling of Equation~(\ref{equation:minmax});
    \item $\phi_\textsc{tmm}$: Theoretical min-max scaling of Equation~(\ref{equation:infmax});
    \item $\phi_\textsc{z}$: z-score normalization of Equation~(\ref{equation:zscore});
    \item $\phi_{\textsc{mm}-\textsc{Lex}}$: Min-max scaling of lexical scores, unnormalized semantic scores;
    \item $\phi_{\textsc{tmm}-\textsc{Lex}}$: Theoretical min-max normalized lexical scores, unnormalized semantic scores;
    \item $\phi_{\textsc{z}-\textsc{Lex}}$: z-score normalized lexical scores, unnormalized semantic scores; and,
    \item $I$: The identity transformation, leaving both semantic and lexical scores unnormalized.
\end{itemize}
We believe these transformations together test the various conditions in our upcoming arguments.

Let us first state the notion of rank-equivalence more formally:
\begin{definition}
We say two functions $f$ and $g$ are \emph{rank-equivalent} on the set $\mathcal{U}$ and write $f \overset{\pi}{=} g$, if the order among documents in a set $\mathcal{U}$ induced by $f$ is the same as that induced by $g$.
\end{definition}

For example, when $\phi_\textsc{Sem}(x) = ax + b$ and $\phi_\textsc{Lex}(x) = cx + d$ are linear transformations of scores for some positive coefficients $a, b$ and real intercepts $b, c$, then they can be reduced to the following rank-equivalent form:
\begin{equation*}
    f_\textsc{Convex}(q, d) \overset{\pi}{=} (a \alpha) f_\textsc{Sem}(q, d) + c(1 - \alpha) f_\textsc{Lex}(q, d).
\end{equation*}
In fact, letting $\alpha^\prime = c\alpha / [c\alpha + c(1-\alpha)]$ transforms the problem to one of learning a convex combination of the original scores with a modified weight. This family of functions includes $\phi_\textsc{mm}$, $\phi_\textsc{z}$, and $\phi_\textsc{tmm}$, and as such solutions for one family can be transformed to solutions for another normalization protocol. More formally:

\begin{lemma}
For every query, given an arbitrary $\alpha$, there exists a $\alpha^\prime$ such that the convex combination of min-max normalized scores with parameter $\alpha$ is rank-equivalent to a convex combination of z-score normalized scores with $\alpha^\prime$, and vice versa.
\end{lemma}
\begin{proof}
Write $m_\textsc{o}$ and $M_\textsc{o}$ for the minimum and maximum scores retrieved by system \textsc{o}, and $\mu_\textsc{o}$ and $\sigma_\textsc{o}$ for their mean and standard deviation. We also write $R_\textsc{o} = M_\textsc{o} - m_\textsc{o}$ for brevity. For every document $d$, we have the following:
\begin{align*}
    \alpha \frac{f_\textsc{Sem}(q, d) - m_\textsc{Sem}}{R_\textsc{Sem}} &+ 
    (1 - \alpha) \frac{f_\textsc{Lex}(q, d) - m_\textsc{Lex}}{R_\textsc{Lex}} 
    \overset{\pi}{=} \frac{\alpha}{R_\textsc{sem}} f_\textsc{Sem}(q, d) + 
    \frac{1 - \alpha}{R_\textsc{Lex}} f_\textsc{Lex}(q, d) \\
    &\overset{\pi}{=} \frac{1}{\sigma_\textsc{Sem} \sigma_\textsc{Lex}}
    \big[ \frac{\alpha}{R_\textsc{Sem}} f_\textsc{Sem}(q, d) +
           \frac{1 - \alpha}{R_\textsc{Lex}} f_\textsc{Lex}(q, d) -
           \frac{\alpha}{R_\textsc{Sem}} \mu_\textsc{Sem} -
           \frac{1 - \alpha}{R_\textsc{Lex}} \mu_\textsc{Lex} \big] \\
    &\overset{\pi}{=} \frac{\alpha}{R_\textsc{Sem} \sigma_\textsc{Lex}} \big(
        \frac{f_\textsc{Sem}(q, d) - \mu_\textsc{Sem}}{\sigma_\textsc{Sem}} \big)
        + \frac{1 - \alpha}{R_\textsc{Lex} \sigma_\textsc{Sem}} \big(
        \frac{f_\textsc{Lex}(q, d) - \mu_\textsc{Lex}}{\sigma_\textsc{Lex}} \big),
\end{align*}
where in every step we either added a constant or multiplied the expression by a positive constant, both rank-preserving operations. Finally, setting
\begin{equation*}
    \alpha^\prime = \frac{\alpha}{R_\textsc{Sem} \sigma_\textsc{Lex}} / (\frac{\alpha}{R_\textsc{Sem} \sigma_\textsc{Lex}}
    + \frac{1 - \alpha}{R_\textsc{Lex} \sigma_\textsc{Sem}})
\end{equation*}
completes the proof. The other direction is similar.
\end{proof}

The fact above implies that the problem of tuning $\alpha$ for a query in a min-max normalization regime is equivalent to learning $\alpha^\prime$ in a z-score normalized setting. In other words, there is a one-to-one relationship between these parameters, and as a result solutions can be mapped from one problem space to the other. However, this statement is only true for individual queries and does not have any implications for the learning of the weight in the convex combination over an entire collection of queries. Let us now consider this more complex setup.

The question we wish to answer is as follows: under what conditions is $f_\textsc{Convex}$ with parameter $\alpha$ and a pair of normalization functions $(\phi_\textsc{Sem}, \phi_\textsc{Lex})$ rank-equivalent to an $f^\prime_\textsc{Convex}$ of a new pair of normalization functions $(\phi_\textsc{Sem}^\prime, \phi_\textsc{Lex}^\prime)$ with weight $\alpha^\prime$? That is, for a constant $\alpha$ with one normalization protocol, when is there a constant $\alpha^\prime$ that produces the same ranked lists for every query but with a different normalization protocol? The answer to this question helps us understand whether and when changing normalization schemes from min-max to z-score, for example, matters. We state the following definitions followed by a theorem that answers this question.

\begin{definition}
    We say $f: \mathbb{R} \rightarrow \mathbb{R}$ is a $\delta$-\emph{expansion} with respect to $g: \mathbb{R} \rightarrow \mathbb{R}$ if for any $x$ and $y$ in the domains of $f$ and $g$ we have that $|f(y) - f(x)| \geq \delta |g(y) - g(x)|$ for some $\delta \geq 1$.
\end{definition}

For example, $\phi_\textsc{mm}(\cdot)$ is an expansion with respect to $\phi_\textsc{tmm}(\cdot)$ with a factor $\delta$ that depends on the range of the scores. As another example, $\phi_\textsc{z}(\cdot)$ is an expansion with respect to $\phi_\textsc{mm}(\cdot)$.

\begin{definition}
    For two pairs of functions $f, g: \mathbb{R} \rightarrow \mathbb{R}$ and $f^\prime, g^\prime: \mathbb{R} \rightarrow \mathbb{R}$, and two points $x$ and $y$ in their domains, we say that $f^\prime$ expands with respect to $f$ \emph{more rapidly} than $g^\prime$ expands with respect to $g$, with a \emph{relative expansion rate} of $\lambda \geq 1$, if the following condition holds:
    \begin{equation*}
        \frac{|f^\prime(y) - f^\prime(x)|}{|f(y) - f(x)|} = \lambda \frac{|g^\prime(y) - g^\prime(x)|}{|g(y) - g(x)|}.
    \end{equation*}
    When $\lambda$ is independent of the points $x$ and $y$, we call this relative expansion \emph{uniform}:
    \begin{equation*}
        \frac{|\Delta f^\prime| / |\Delta f|}{|\Delta g^\prime| / |\Delta g|} = \lambda,\, \forall x, y.
    \end{equation*}
    \label{definition:growth}
\end{definition}

As an example, if $f$ and $g$ are min-max scaling and $f^\prime$ and $g^\prime$ are z-score normalization, then their respective rate of expansion is roughly similar. We will later show that this property often holds empirically across different transformations.

\begin{theorem}
    For every choice of $\alpha$, there exists a constant $\alpha^\prime$ such that the following functions are rank-equivalent on a collection of queries $Q$:
    \begin{equation*}
        f_\textsc{Convex} = \alpha \phi(f_\textsc{Sem}(q, d)) + (1 - \alpha) \omega(f_\textsc{Lex}(q, d)),
    \end{equation*}
    and
    \begin{equation*}
        f^\prime_\textsc{Convex} = \alpha^\prime \phi^\prime(f_\textsc{Sem}(q, d)) + (1 - \alpha^\prime) \omega^\prime(f_\textsc{Lex}(q, d)),
    \end{equation*}
    if for the monotone functions $\phi, \omega, \phi^\prime, \omega^\prime: \mathbb{R} \rightarrow \mathbb{R}$, $\phi^\prime$ expands with respect to $\phi$ more rapidly than $\omega^\prime$ expands with respect to $\omega$ with a uniform rate $\lambda$.
    \label{theorem:main}
\end{theorem}

\begin{proof}
Consider a pair of documents $d_i$ and $d_j$ in the ranked list of a query $q$ such that $d_i$ is ranked above $d_j$ according to $f_\textsc{Convex}$. Shortening $f_\textsc{o}(q, d_k)$ to $f^{(k)}_\textsc{o}$ for brevity, we have that:
\begin{align*}
    f^{(i)}_\textsc{Convex} > f^{(j)}_\textsc{Convex} &\implies
    \alpha \big[ \underbrace{(\phi(f^{(i)}_\textsc{Sem}) - \phi(f^{(j)}_\textsc{Sem}))}_{\Delta\phi_{ij}} +
                \underbrace{(\omega(f^{(j)}_\textsc{Lex}) - \omega(f^{(i)}_\textsc{Lex}))}_{\Delta\omega_{ji}} \big] >
                \omega(f^{(j)}_\textsc{Lex}) - \omega(f^{(i)}_\textsc{Lex})
\end{align*}

This holds if and only if we have the following:

\begin{equation}
\begin{cases}
    \alpha > 1/(1 + \frac{\Delta\phi_{ij}}{\Delta\omega_{ji}}), & \text{if $\Delta\phi_{ij} + \Delta\omega_{ji} > 0$}, \\
    \alpha < 1/(1 + \frac{\Delta\phi_{ij}}{\Delta\omega_{ji}}), & \text{otherwise}.
\end{cases}
\label{equation:convex:conditions}
\end{equation}

Observe that, because of the monotonicity of a convex combination and the monotonicity of the normalization functions, the case $\Delta\phi_{ij} < 0$ and $\Delta\omega_{ji} > 0$ (which implies that the semantic and lexical scores of $d_j$ are both larger than $d_i$) is not valid as it leads to a reversal of ranks. Similarly, the opposite case $\Delta\phi_{ij} > 0$ and $\Delta\omega_{ji} < 0$ always leads to the correct order regardless of the weight in the convex combination.  We consider the other two cases separately below.

\emph{Case 1: $\Delta\phi_{ij} > 0$ and $\Delta\omega_{ji} > 0$}. Because of the monotonicity property, we can deduce that $\Delta\phi^\prime_{ij} > 0$ and $\Delta\omega^\prime_{ji} > 0$. From Equation~(\ref{equation:convex:conditions}), for the order between $d_i$ and $d_j$ to be preserved under the image of $f^\prime_\textsc{Convex}$, we must therefore have the following:
\begin{equation*}
    \alpha^\prime > 1/(1 + \frac{\Delta\phi^\prime_{ij}}{\Delta\omega^\prime_{ji}}).
\end{equation*}
By assumption, using Definition~\ref{definition:growth}, we observe that:
\begin{equation*}
    \frac{\Delta\phi^\prime_{ij}}{\Delta\phi_{ij}} \geq \frac{\Delta\omega^\prime_{ji}}{\Delta\omega_{ji}} \implies
    \frac{\Delta\phi^\prime_{ij}}{\Delta\omega^\prime_{ji}} \geq \frac{\Delta\phi_{ij}}{\Delta\omega_{ji}}.
\end{equation*}
As such, the lower-bound on $\alpha^\prime$ imposed by documents $d_i$ and $d_j$ of query $q$, $L^\prime_{ij}(q)$, is smaller than the lower-bound on $\alpha$, $L_{ij}(q)$. Like $\alpha$, this case does not additionally constrain $\alpha^\prime$ from above (i.e., the upper-bound does not change: $U^\prime_{ij}(q) = U_{ij}(q) = 1$).

\emph{Case 2: $\Delta\phi_{ij} < 0$, $\Delta\omega_{ji} < 0$}. Once again, due to monotonicity, it is easy to see that $\Delta\phi^\prime_{ij} < 0$ and $\Delta\omega^\prime_{ji} < 0$. Equation~(\ref{equation:convex:conditions}) tells us that, for the order to be preserved under $f^\prime_\textsc{Convex}$, we must similarly have that:
\begin{equation*}
    \alpha^\prime < 1/(1 + \frac{\Delta\phi^\prime_{ij}}{\Delta\omega^\prime_{ji}}).
\end{equation*}
Once again, by assumption we have that the upper-bound on $\alpha^\prime$ is a translation of the upper-bound on $\alpha$ to the left. The lower-bound is unaffected and remains $0$.

For $f^\prime_\textsc{Convex}$ to induce the same order as $f_\textsc{Convex}$ among all pairs of documents for all queries in $Q$, the intersection of the intervals produced by the constraints on $\alpha^\prime$ has to be non-empty:

\begin{equation*}
    I^\prime \triangleq \bigcap_{q} [\max_{ij} L^\prime_{ij}(q), \min_{ij} U^\prime_{ij}(q)] = [\max_{q, ij} L^\prime_{ij}(q), \min_{q, ij} U^\prime_{ij}(q)] \neq \emptyset.
\end{equation*}
We next prove that $I^\prime$ is always non-empty to conclude the proof of the theorem.

By Equation~(\ref{equation:convex:conditions}) and the existence of $\alpha$, we know that $\max_{q,ij} L_{ij}(q) \leq \min_{q,ij} U_{ij}(q)$. Suppose that documents $d_i$ and $d_j$ of query $q_1$ maximize the lower-bound, and that documents $d_m$ and $d_n$ of query $q_2$ minimize the upper-bound. We therefore have that:
\begin{equation*}
    1 / (1 + \frac{\Delta \phi_{ij}}{\Delta \omega_{ji}}) \leq 1 / (1 + \frac{\Delta \phi_{mn}}{\Delta \omega_{nm}}) \implies
    \frac{\Delta \phi_{ij}}{\Delta \omega_{ji}} \geq \frac{\Delta \phi_{mn}}{\Delta \omega_{nm}}
\end{equation*}
Because of the uniformity of the relative expansion rate, we can deduce that:
\begin{equation*}
    \frac{\Delta \phi^\prime_{ij}}{\Delta \omega^\prime_{ji}} \geq \frac{\Delta \phi^\prime_{mn}}{\Delta \omega^\prime_{nm}} \implies \max_{q,ij} L^\prime_{ij}(q) \leq \min_{q,ij} U^\prime_{ij}(q).
\end{equation*}
\end{proof}

It is easy to show that the theorem above also holds when the condition is updated to reflect a shift of lower- and upper-bounds to the right, which happens when $\phi^\prime$ \emph{contracts} with respect to $\phi$ more rapidly than $\omega^\prime$ does with respect to $\omega$.

The picture painted by Theorem~\ref{theorem:main} is that switching from min-max scaling to z-score normalization or any other linear transformation that is bounded and does not severely distort the distribution of scores, especially among the top-ranking documents, results in a rank-equivalent function. At most, for any given value of the ranking metric of interest such as NDCG, we should observe a shift of the weight in the convex combination to the right or left. Figure~\ref{figure:convex:normalization} illustrates this effect empirically on select datasets. As anticipated, the peak performance in terms of NDCG shifts to the left or right depending on the type of normalization.

\begin{figure}[t]
\begin{center}
\centerline{
\subfloat[MS MARCO]{
\includegraphics[height=2in]{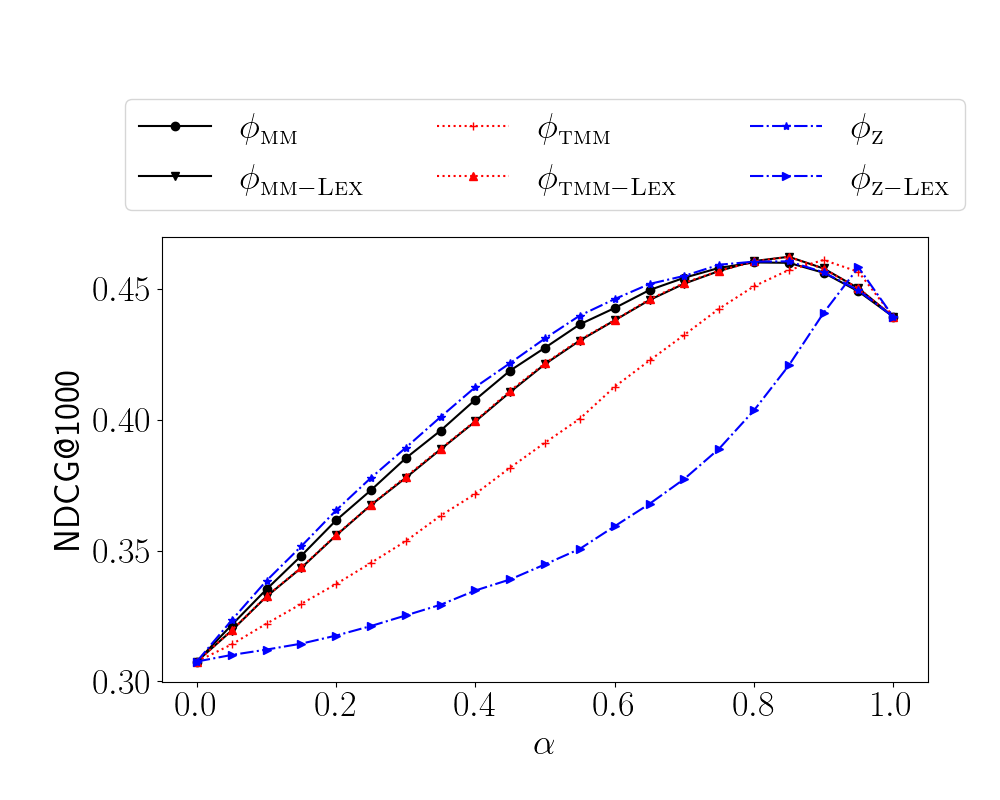}}
\subfloat[\textsc{Quora}]{
\includegraphics[height=2in]{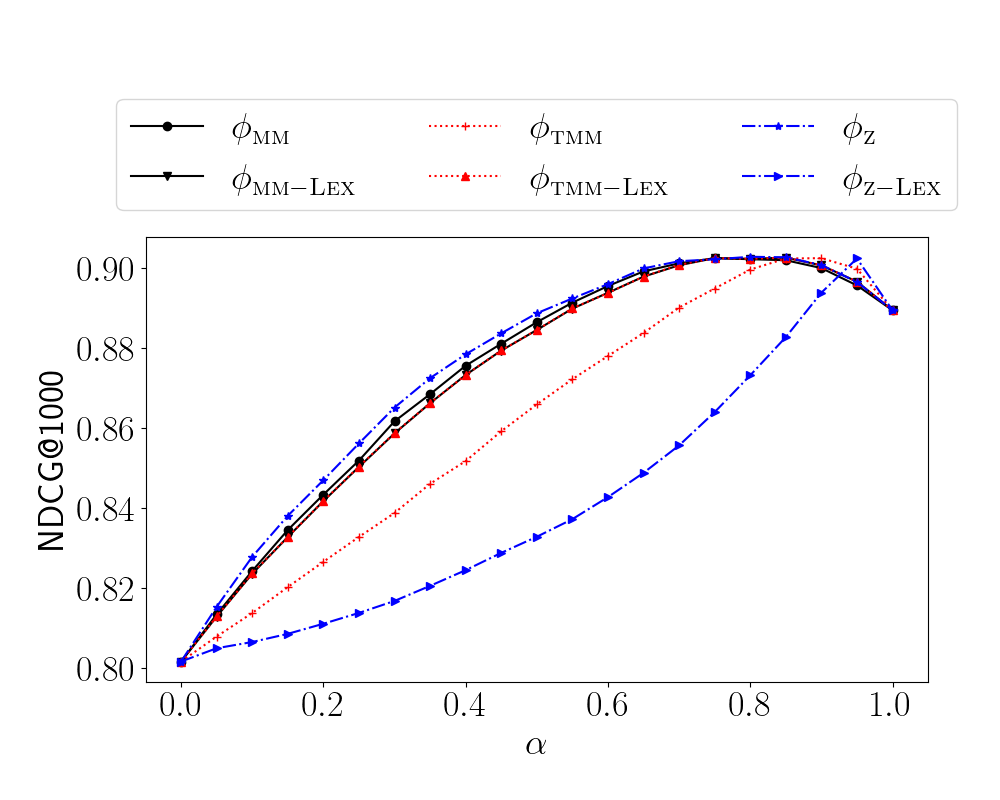}}
}
\centerline{
\subfloat[\textsc{HotpotQA}]{
\includegraphics[height=2in]{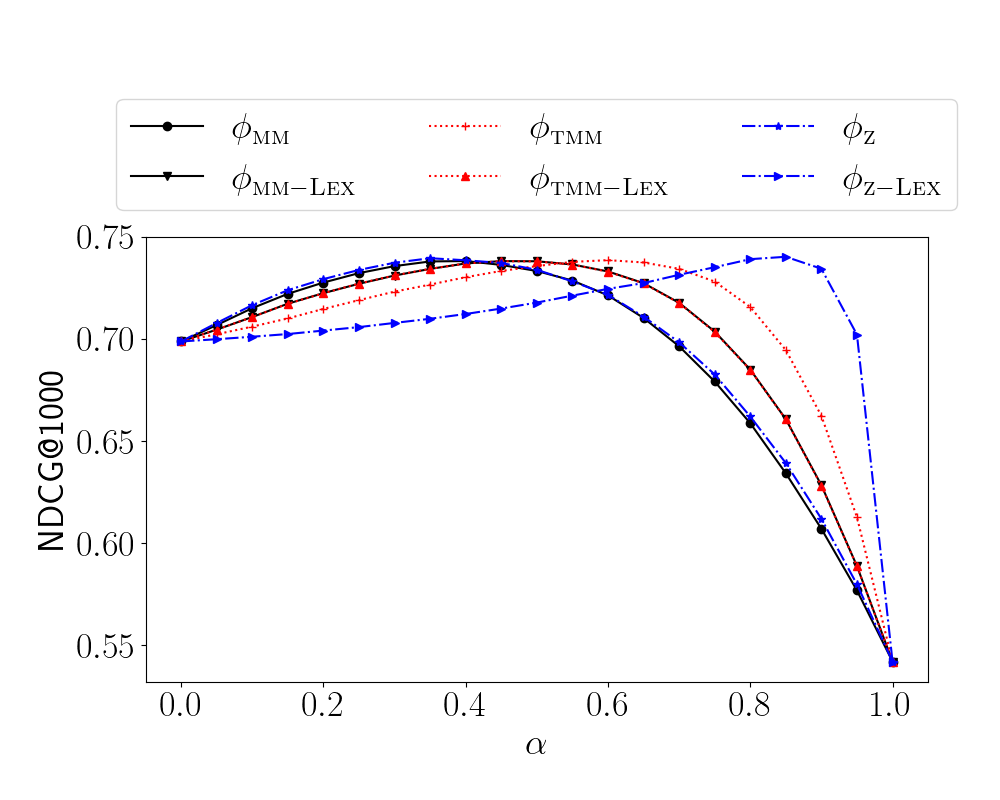}}
\subfloat[\textsc{FiQA}]{
\includegraphics[height=2in]{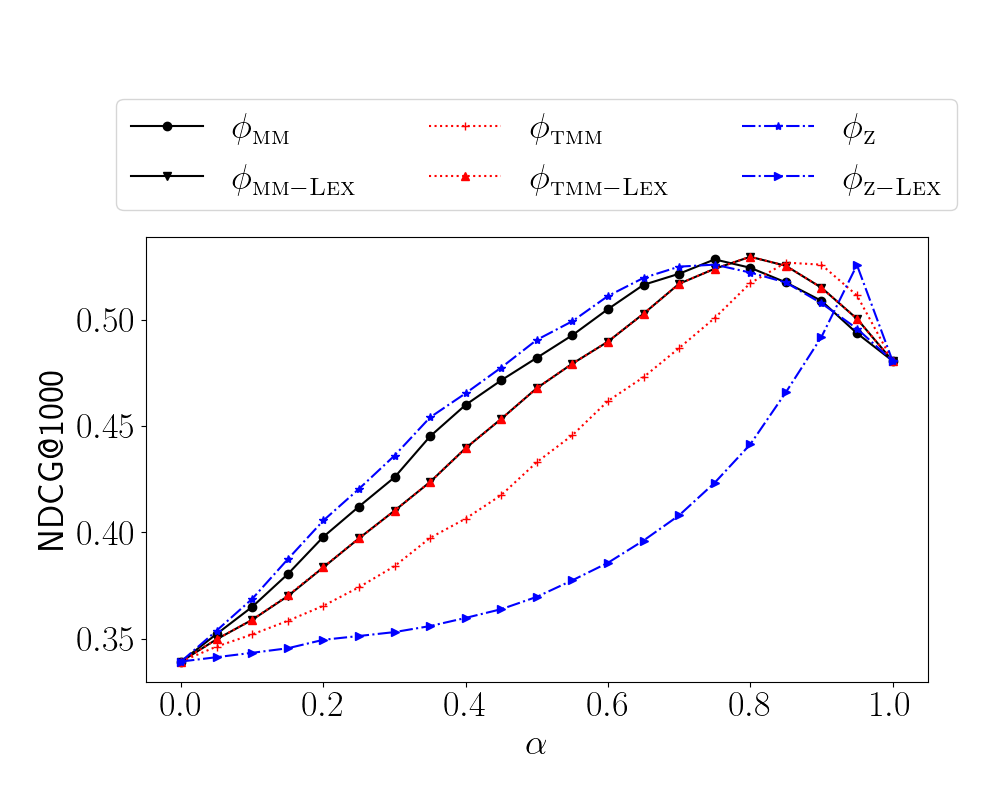}}
}
\caption{Effect of normalization on the performance of $f_\textsc{Convex}$ as a function of $\alpha$ on the validation set.}
\label{figure:convex:normalization}
\end{center}
\end{figure}

The uniformity requirement on the relative expansion rate, $\lambda$, in Theorem~\ref{theorem:main} is not as strict and restrictive as it may appear. First, it is only necessary for $\lambda$ to be stable on the set of ordered pairs of documents as ranked by $f_\textsc{Convex}$:
\begin{equation*}
    \frac{|\Delta \phi^\prime_{ij}| / |\Delta \phi_{ij}|}{|\Delta \omega^\prime_{ji}| / |\Delta \omega_{ji}|} = \lambda,\, \forall
    (d_i, d_j) \textit{ st } f_\textsc{Convex}(d_i) > f_\textsc{Convex}(d_j).
\end{equation*}
Second, as we will observe experimentally, $\lambda$ being \emph{concentrated} around one value, rather than being the same constant everywhere as uniformity requires, is often sufficient for the effect to materialize in practice. We observe this phenomenon empirically by fixing the parameter $\alpha$ in $f_\textsc{Convex}$ with one transformation and forming ranked lists, then choosing another transformation and computing its relative expansion rate $\lambda$ on all ordered pairs of documents. We show the measured relative expansion rate in Figure~\ref{figure:relative_expansion} for various transformations.

\begin{figure}[t]
\begin{center}
\centerline{
\subfloat[MS MARCO]{
\includegraphics[height=2in]{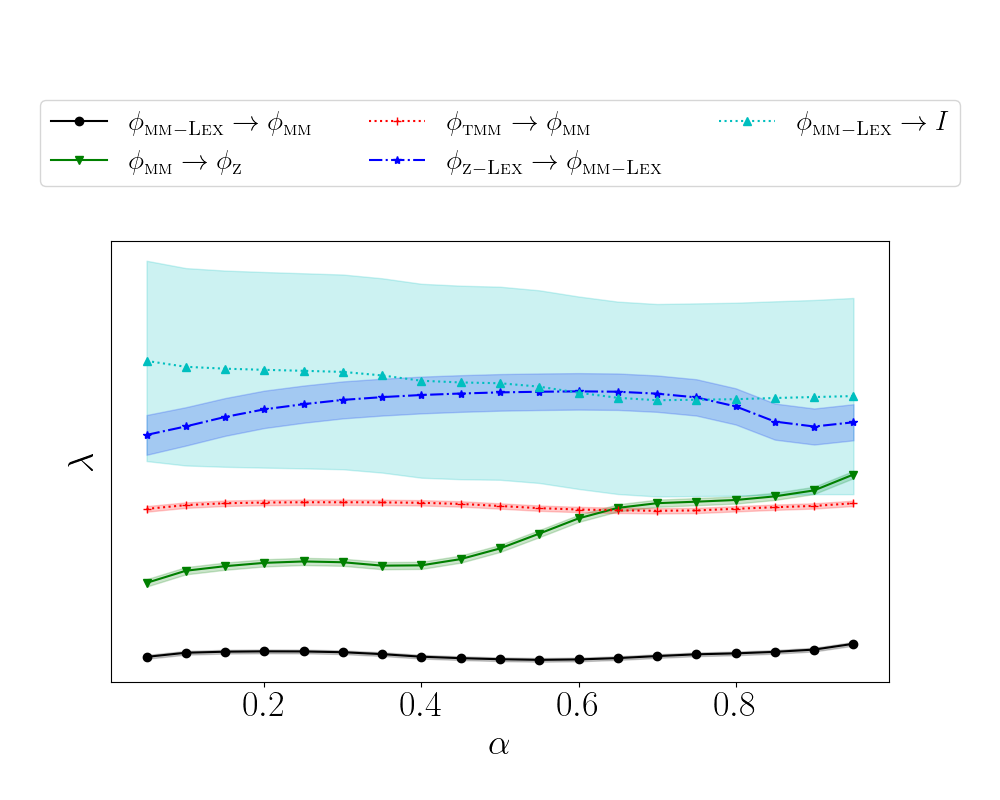}}
\subfloat[\textsc{Quora}]{
\includegraphics[height=2in]{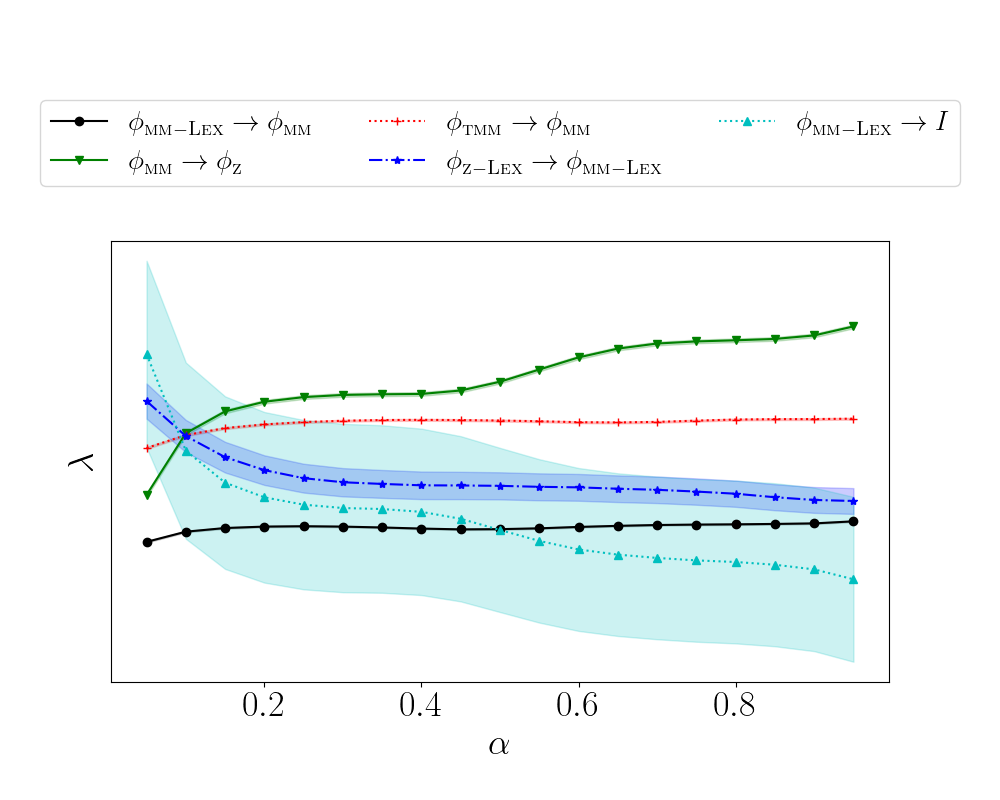}}
}
\centerline{
\subfloat[\textsc{HotpotQA}]{
\includegraphics[height=2in]{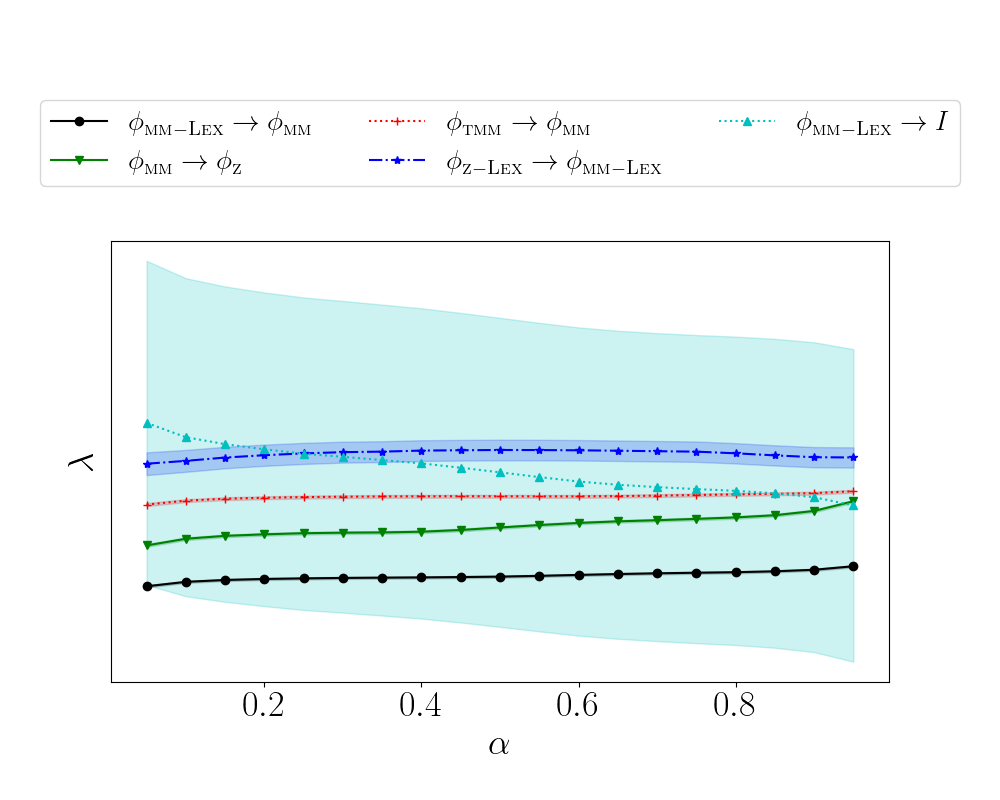}}
\subfloat[\textsc{FiQA}]{
\includegraphics[height=2in]{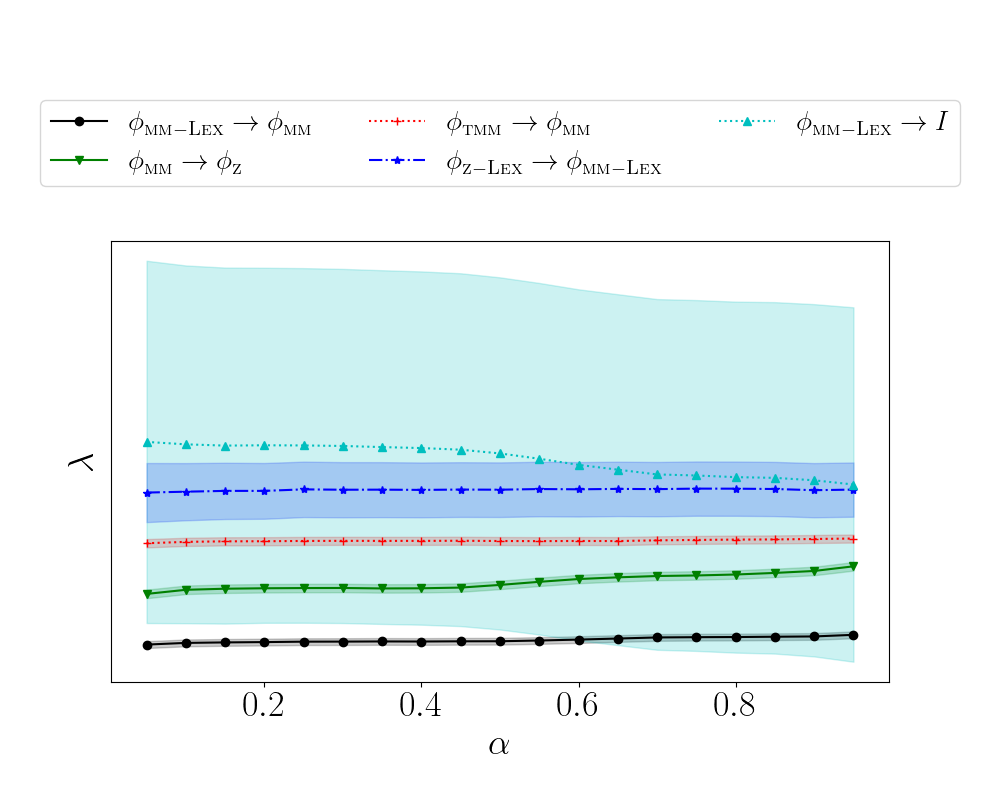}}
}
\caption{Relative expansion rate of semantic scores with respect to lexical scores, $\lambda$, when changing from one transformation to another, with $95\%$ confidence intervals. Prior to visualization, we normalize values of $\lambda$ to bring them into a similar scale---this only affects aesthetics and readability, but is the reason why the vertical axis is not scaled. For most transformations and every value of $\alpha$, we observe a stable relative rate of expansion where $\lambda$ concentrates around one value for the vast majority of queries.}
\label{figure:relative_expansion}
\end{center}
\end{figure}

Figure~\ref{figure:relative_expansion} shows that most pairs of transformations yield a stable relative expansion rate. For example, if $f_\textsc{Convex}$ uses $\phi_\textsc{tmm}$ and $f^\prime_\textsc{Convex}$ uses $\phi_\textsc{mm}$---denoted by $\phi_\textsc{tmm} \rightarrow \phi_\textsc{mm}$---for every choice of $\alpha$, the relative expansion rate $\lambda$ is concentrated around a constant value. This implies that any ranked list obtained from $f_\textsc{Convex}$ can be reconstructed by $f^\prime_\textsc{Convex}$. Interestingly, $\phi_{\textsc{z}-\textsc{Lex}} \rightarrow \phi_{\textsc{mm}-\textsc{Lex}}$ has a comparatively less stable $\lambda$, but removing normalization altogether (i.e., $\phi_{\textsc{mm}-\textsc{Lex}} \rightarrow I$) dramatically distorts the expansion rates. This goes some way to explain why normalization and boundedness are important properties.

This connection between boundedness and the effect that changing the normalization function has on ranking quality, is clearer in the experiments presented in Appendix~\ref{appendix:splade-bm25} through Appendix~\ref{appendix:tasb-minilm}. In general, when a function is unbounded, fusing with normalization versus without results in a relative expansion rate $\lambda$ with a high variance, which leads to relatively different classes of rankings: there exists ranked lists which are produced by a fusion \emph{with} normalization that cannot be reconstructed by a fusion \emph{without} normalization. For pairs of normalization functions whose relative expansion rate is stable and highly concentrated, on the other hand, the curves showing the effect of $\alpha$ on $f_\textsc{Convex}$ are translations of each other, as Theorem~\ref{theorem:main} predicts: every ranked list produced by a fusion using one normalization function can be reproduced by a fusion using another.

In the last two sections, we have answered RQ1: convex combination is an appropriate fusion function and its performance is not sensitive to the choice of normalization so long as the transformation has reasonable properties. Interestingly, the behavior of $\phi_\textsc{tmm}$ appears to be more robust to the data distribution---its peak remains within a small neighborhood as we move from one dataset to another. We believe the reason $\phi_\textsc{tmm}$-normalized scores are more stable is because it has one fewer data-dependent statistic in the transformation (i.e., minimum score in the retrieved set is replaced with minimum feasible value regardless of the candidate set). In the remainder of this work, we use $\phi_\textsc{tmm}$ and denote a convex combination of scores normalized by it by TM2C2 for brevity.
Where the theoretical minimum does not exist (e.g., with the \textsc{Tas-B} model), we use $\phi_\textsc{mm}$ instead and denote it by M2C2.

\section{Analysis of Reciprocal Rank Fusion}
\label{section:analysis:rrf}

Chen \emph{et al.}~\cite{chen2022ecir} show that \ac{rrf} performs better and more reliably than a convex combination of normalized scores. \ac{rrf} is computed as follows:
\begin{equation}
    f_\textsc{RRF}(q, d)= \frac{1}{\eta + \pi_\textsc{Lex}(q, d)} + \frac{1}{\eta + \pi_\textsc{Sem}(q, d)},
    \label{equation:rrf}
\end{equation}
where $\eta$ is a free parameter. The authors of~\cite{chen2022ecir} take a non-parametric view of \ac{rrf}, where the parameter $\eta$ is set to its default value $60$, in order to apply the fusion to out-of-domain datasets in a zero-shot manner. In this work, we additionally take a parametric view of \ac{rrf}, where as we elaborate later, the number of free parameters is the same as the number of functions being fused together, a quantity that is always larger than the number of parameters in a convex combination.

Let us begin by comparing the performance of \ac{rrf} and TM2C2 empirically to get a sense of their relative efficacy. We first verify whether hybrid retrieval leads to significant gains in in-domain and out-of-domain experiments. In a way, we seek to confirm the findings reported in~\cite{chen2022ecir} and compare the two fusion functions in the process.

Table~\ref{table:rq1} summarizes our results for our primary models, with results for the remaining fusions reported in the appendices. We note that, we set \ac{rrf}'s $\eta$ to $60$ per~\cite{chen2022ecir} but tuned TM2C2's $\alpha$ on the validation set of the in-domain datasets and found that $\alpha=0.8$ works well for the three datasets. In the experiments leading to Table~\ref{table:rq1}, we fix $\alpha=0.8$ and evaluate methods on the test split of the datasets. Per~\cite{chen2022ecir,wang2021bert}, we have also included the performance of an oracle system that uses a per-query $\alpha$, to establish an upper-bound---the oracle knows which value of $\alpha$ works best for any given query.

\begin{table*}[t]
\caption{Recall@1000 and NDCG@1000 (except \textsc{SciFact} and \textsc{NFCorpus} where cutoff is $100$) on the test split of various datasets for lexical and semantic search as well as hybrid retrieval using \ac{rrf}~\cite{chen2022ecir} ($\eta=60$) and TM2C2 ($\alpha=0.8$). The symbols $\ddagger$ and $\ast$ indicate statistical significance ($p$-value $< 0.01$) with respect to TM2C2 and RRF respectively, according to a paired two-tailed $t$-test.}
\label{table:rq1}
\begin{center}
\begin{sc}
\begin{tabular}{cc|cccc|ccccc}
& & \multicolumn{4}{c}{Recall} & \multicolumn{5}{c}{NDCG} \\
\toprule
& \small{Dataset} & \small{Lex.} & \small{Sem.} & \small{TM2C2} & \small{\ac{rrf}} & \small{Lex.} & \small{Sem.} & \small{TM2C2} & \small{\ac{rrf}} & \small{Oracle}  \\
\midrule
\parbox[t]{1mm}{\multirow{2}{*}{\rotatebox[origin=c]{90}{in-domain}}}
& \small{MS MARCO} & $0.836^{\ddagger\ast}$ & $0.964^{\ddagger\ast}$ & \textbf{0.974} & $0.969^\ddagger$ &
  $0.309^{\ddagger\ast}$ & $0.441^{\ddagger\ast}$ & \textbf{0.454} & $0.425^{\ddagger}$ & 0.547 \\
& \small{NQ} & $0.886^{\ddagger\ast}$ & $0.978^{\ddagger\ast}$ & \textbf{0.985} & $0.984$ &
  $0.382^{\ddagger\ast}$ & $0.505^{\ddagger}$ & \textbf{0.542} & $0.514^{\ddagger}$ & 0.637 \\
& \small{Quora} & $0.992^{\ddagger\ast}$ & \textbf{0.999} & \textbf{0.999} & \textbf{0.999} &
  $0.800^{\ddagger\ast}$ & $0.889^{\ddagger\ast}$ & \textbf{0.901} & $0.877^{\ddagger}$ & 0.936 \\
& & & & & & & \\
\midrule
\parbox[t]{1mm}{\multirow{6}{*}{\rotatebox[origin=c]{90}{zero-shot}}}
& \small{NFCorpus} & $0.255^{\ddagger\ast}$ & $0.320^{\ddagger\ast}$ & \textbf{0.338} & $0.327$ &
  $0.268^{\ddagger\ast}$ & $0.296^{\ddagger\ast}$ & \textbf{0.327} & $0.312^{\ddagger}$ & 0.371 \\
& \small{HotpotQA} & $0.878^{\ddagger\ast}$ & $0.756^{\ddagger\ast}$ & $0.884$ & \textbf{0.888} &
  $0.682^{\ddagger\ast}$ & $0.520^{\ddagger\ast}$ & \textbf{0.699} & $0.675^{\ddagger}$ & 0.767 \\
& \small{FEVER} & $0.969^{\ddagger\ast}$ & $0.931^{\ddagger\ast}$ & \textbf{0.972} & \textbf{0.972} &
  $0.689^{\ddagger\ast}$ & $0.558^{\ddagger\ast}$ & \textbf{0.744} & $0.721^{\ddagger}$ & 0.814 \\
& \small{SciFact} & $0.900^{\ddagger\ast}$ & $0.932^{\ddagger\ast}$ & \textbf{0.958} & $0.955$ &
  $0.698^{\ddagger\ast}$ & $0.681^{\ddagger\ast}$ & \textbf{0.753} & $0.730^{\ddagger}$ & 0.796 \\
& \small{DBPedia} & $0.540^{\ddagger\ast}$ & $0.408^{\ddagger\ast}$ & $0.564$ & \textbf{0.567} &
  $0.415^{\ddagger\ast}$ & $0.425^{\ddagger\ast}$ & \textbf{0.512} & $0.489^{\ddagger}$ & 0.553 \\
& \small{FiQA} & $0.720^{\ddagger\ast}$ & \textbf{0.908} & $0.907$ & $0.904$ &
  $0.315^{\ddagger\ast}$ & $0.467^{\ddagger}$ & \textbf{0.496} & $0.464^{\ddagger}$ & 0.561 \\
\bottomrule
\end{tabular}
\end{sc}
\end{center}
\end{table*}

Our results show that hybrid retrieval using \ac{rrf} outperforms pure-lexical and pure-semantic retrieval on most datasets. This fusion method is particularly effective on out-of-domain datasets, rendering the observation of~\cite{chen2022ecir} a robust finding and asserting once more the remarkable performance of \ac{rrf} in zeros-shot settings.

Contrary to~\cite{chen2022ecir}, however, we find that TM2C2 significantly outperforms \ac{rrf} on all datasets in terms of NDCG, and does generally better in terms of Recall. Our observation is consistent with~\cite{wang2021bert} that TM2C2 substantially boosts NDCG even on in-domain datasets.

To contextualize the effect of $\alpha$ on ranking quality, we visualize a parameter sweep on the validation split of in-domain datasets in Figure~\subref*{figure:TM2C2_sweep:in_domain}, and for completeness, on the test split of out-of-domain datasets in Figure~\subref*{figure:TM2C2_sweep:out_of_domain}. These figures also compare the performance of TM2C2 with \ac{rrf} by reporting the difference between NDCG of the two methods. These plots show that there always exists an interval of $\alpha$ for which $f_\textsc{TM2C2} \succ f_\textsc{\ac{rrf}}$ with $\succ$ indicating better rank quality.

\begin{figure}[t]
\begin{center}
\centerline{
\subfloat[in-domain]{
\includegraphics[height=2in]{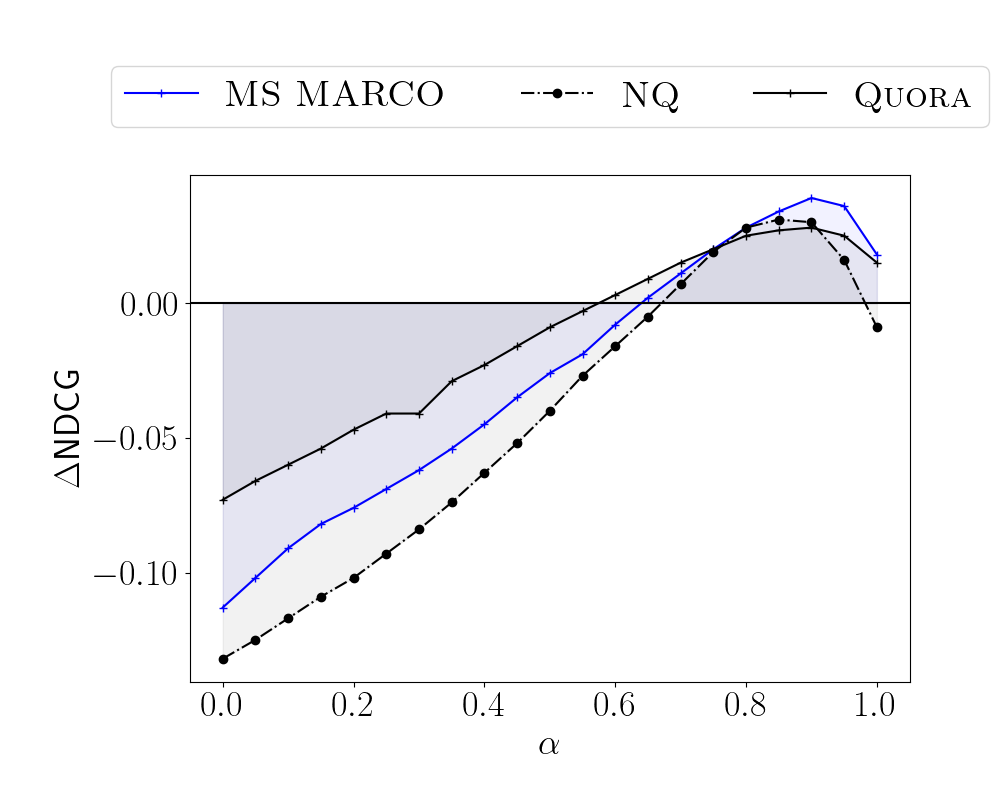}\label{figure:TM2C2_sweep:in_domain}}
\subfloat[out-of-domain]{
\includegraphics[height=2in]{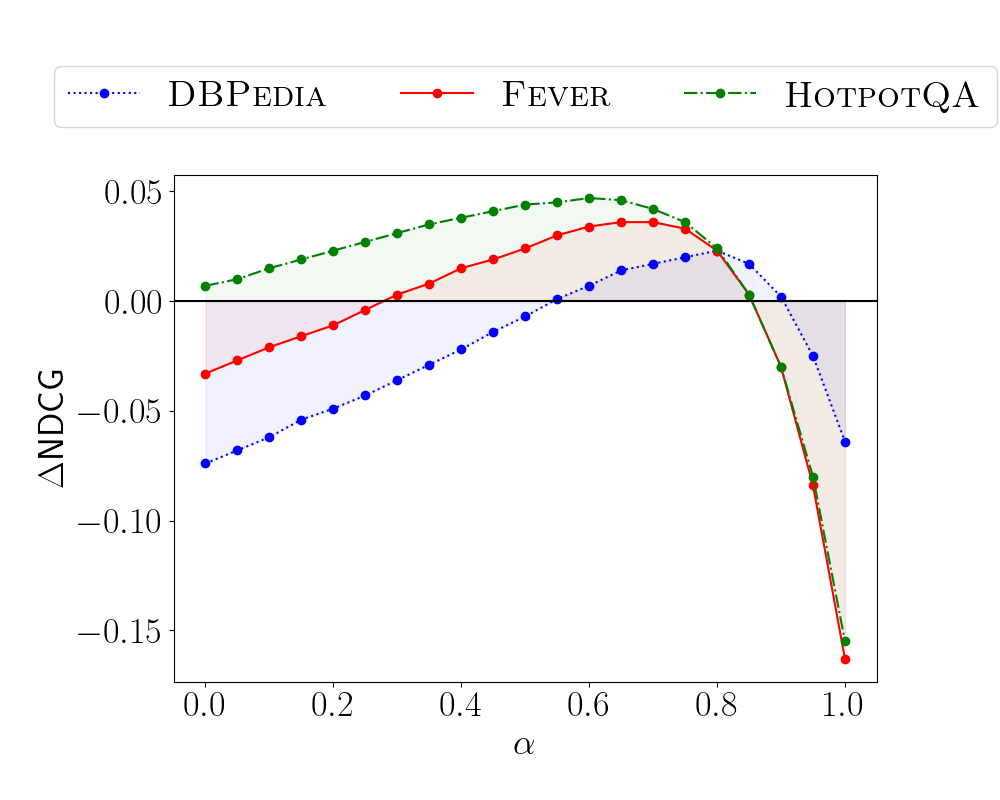}\label{figure:TM2C2_sweep:out_of_domain}}
}
\caption{Difference in NDCG@$1000$ of TM2C2 and \ac{rrf} (positive indicates better ranking quality by TM2C2) as a function of $\alpha$. When $\alpha=0$ the model is rank-equivalent to lexical search while $\alpha=1$ is rank-equivalent to semantic search.}
\label{figure:TM2C2_sweep}
\end{center}
\end{figure}

\subsection{Effect of Parameters}
Chen \emph{et al.}~\cite{chen2022ecir} rightly argue that because \ac{rrf} is merely a function of ranks, rather than scores, it naturally addresses the scale and range problem without requiring normalization---which, as we showed, is not a consequential choice anyway. While that statement is accurate, we believe it introduces new problems that must be recognized too.

The first, more minor issue is that ranks cannot be computed \emph{exactly} unless the entire collection $\mathcal{D}$ is ranked by retrieval system $\textsc{o}$ for every query. That is because, there may be documents that appear in the union set, but not in one of the individual top-$k$ sets. Their true rank is therefore unknown, though is often approximated by ranking documents within the union set. We take this approach when computing ranks.

The second issue is that, unlike TM2C2, \ac{rrf} ignores the raw scores and discards information about their distribution. In this regime, whether or not a document has a low or high semantic score does not matter so long as its rank in $R^k_\textsc{Sem}$ stays the same. It is arguable in this case whether rank is a stronger signal of relevance than score, a measurement in a metric space where \emph{distance} matters greatly. We intuit that, such distortion of distances may result in a loss of valuable information that would lead to better final ranked lists.

To understand these issues better, let us first repeat the exercise in Section~\ref{section:analysis:convex:suitability} for \ac{rrf}. In Figure~\ref{figure:rrf:suitability}, we have plotted the reciprocal rank (i.e., $rr(\pi_\textsc{o}) = 1 / (\eta + \pi_\textsc{o})$ with $\eta=60$) for sampled query-document pairs as before. We choose $\eta=60$ per the setup in~\cite{chen2022ecir}, but note that changing this value leads to different distributions: as $\eta$ approaches $\infty$, for example, all scores will collapse to a single point regardless of the original ranks ($\pi_\textsc{o}$).
\begin{figure*}[t]
\begin{center}
\centerline{
\subfloat[MS MARCO]{
\includegraphics[height=1.6in]{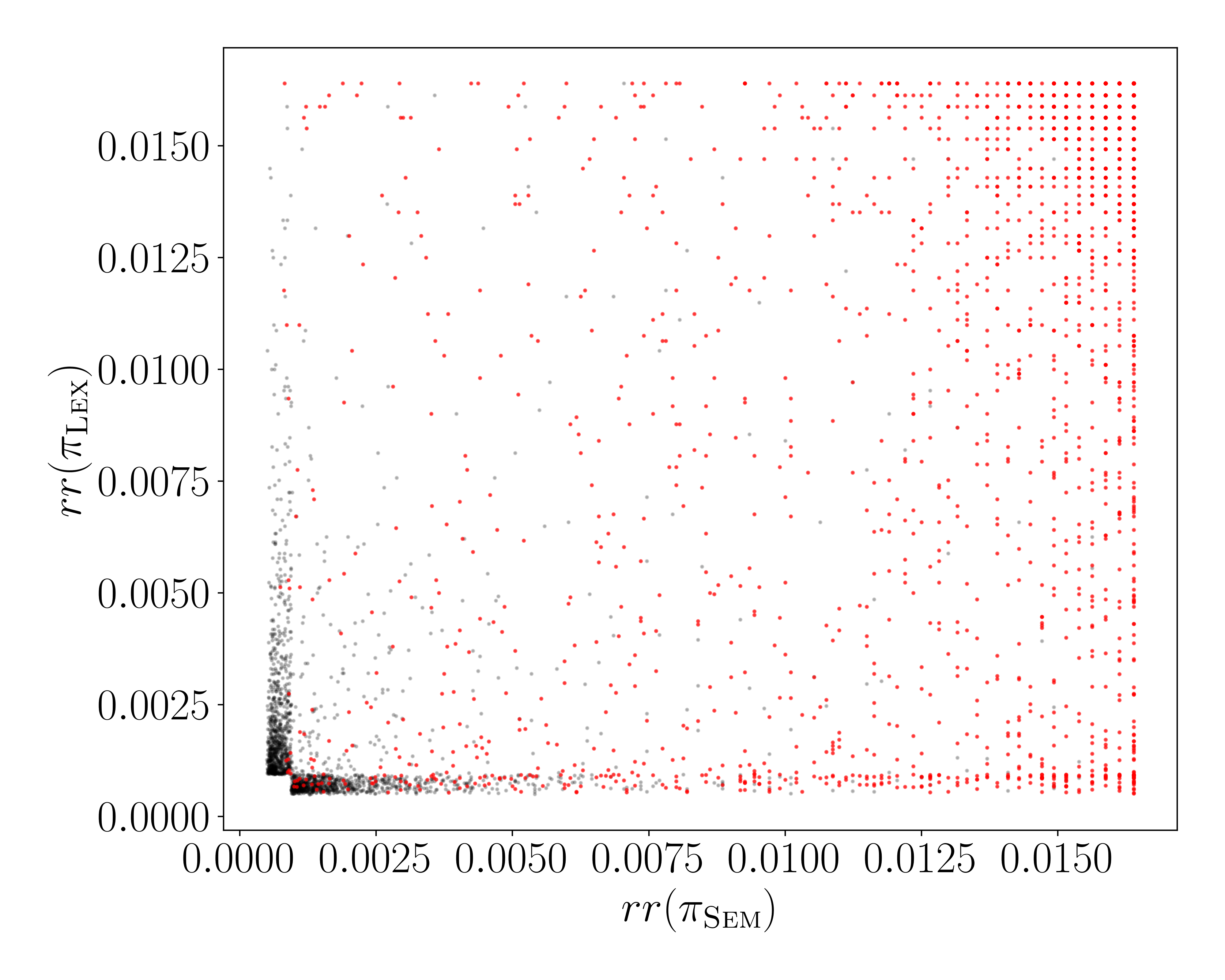}}
\subfloat[\textsc{Quora}]{
\includegraphics[height=1.6in]{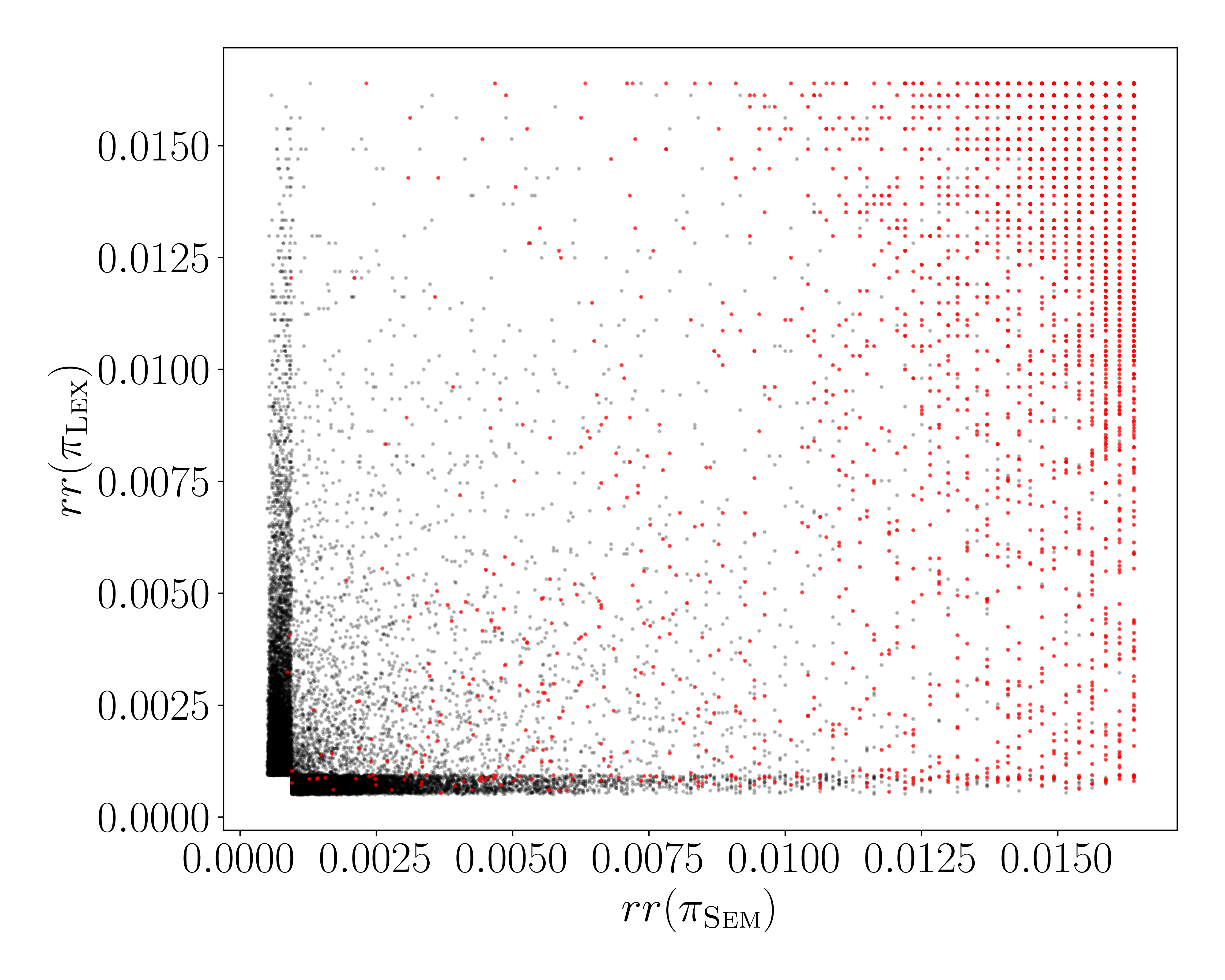}}
}
\centerline{
\subfloat[\textsc{NQ}]{
\includegraphics[height=1.6in]{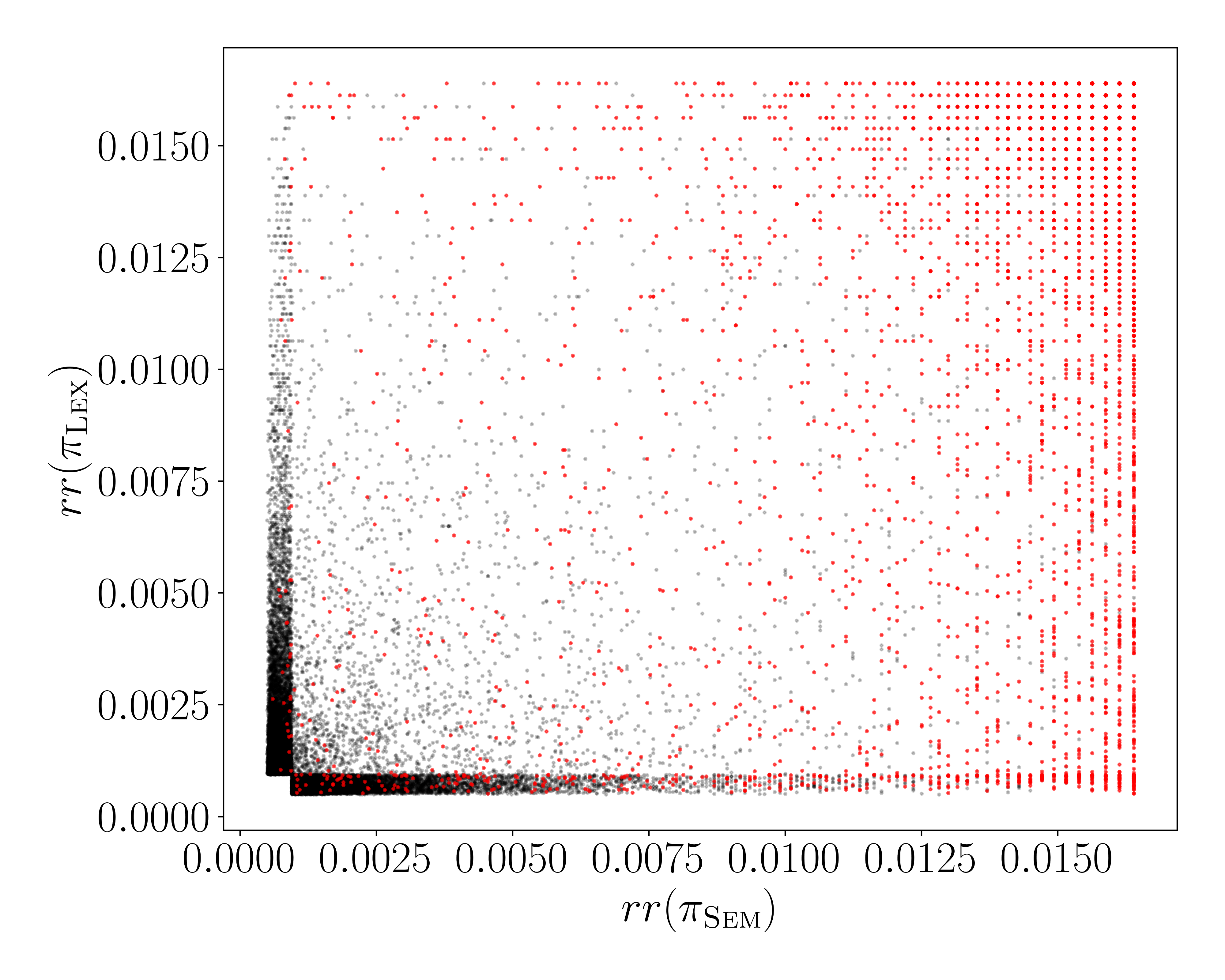}}
\subfloat[\textsc{FiQA}]{
\includegraphics[height=1.6in]{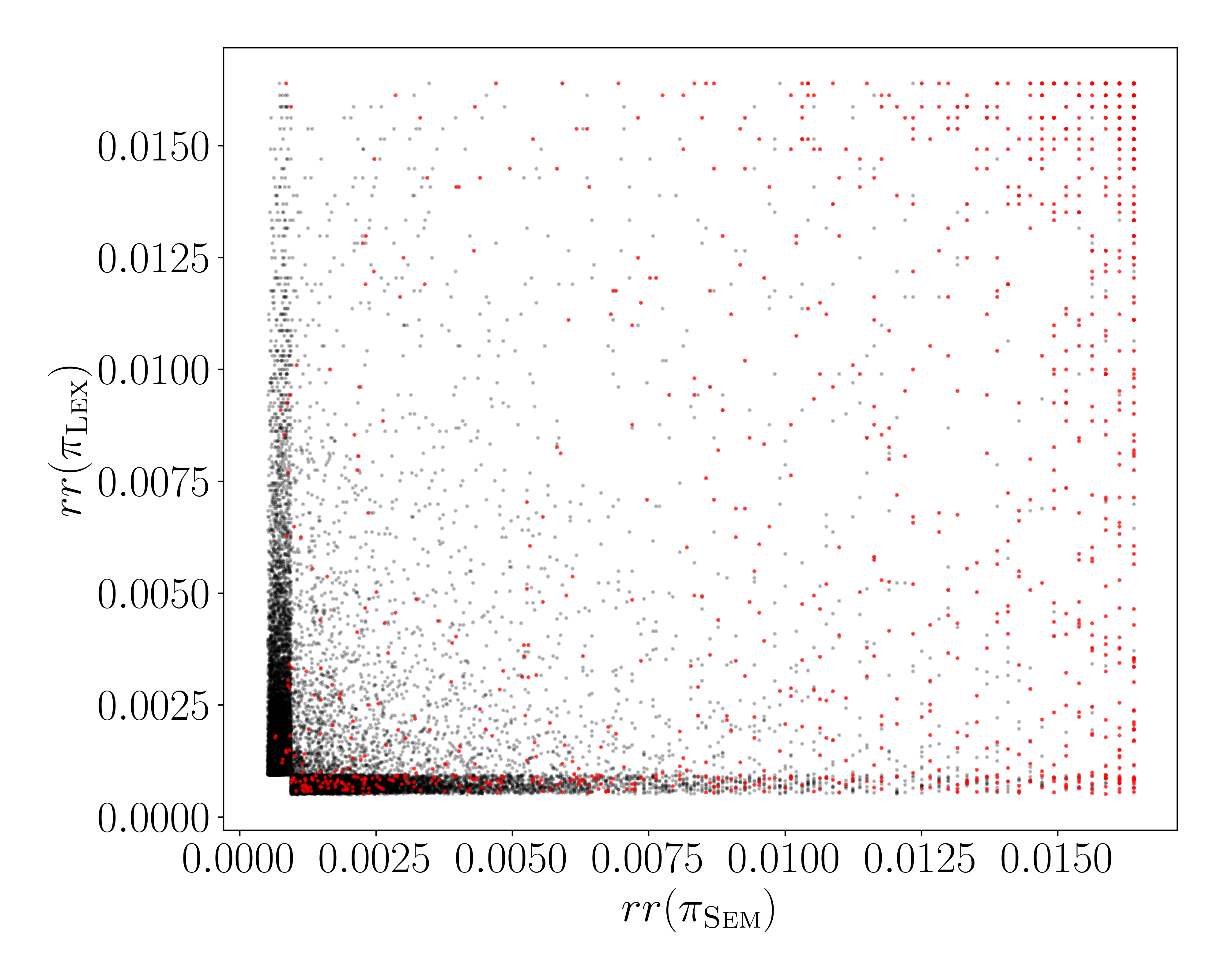}}
}
\centerline{
\subfloat[\textsc{HotpotQA}]{
\includegraphics[height=1.6in]{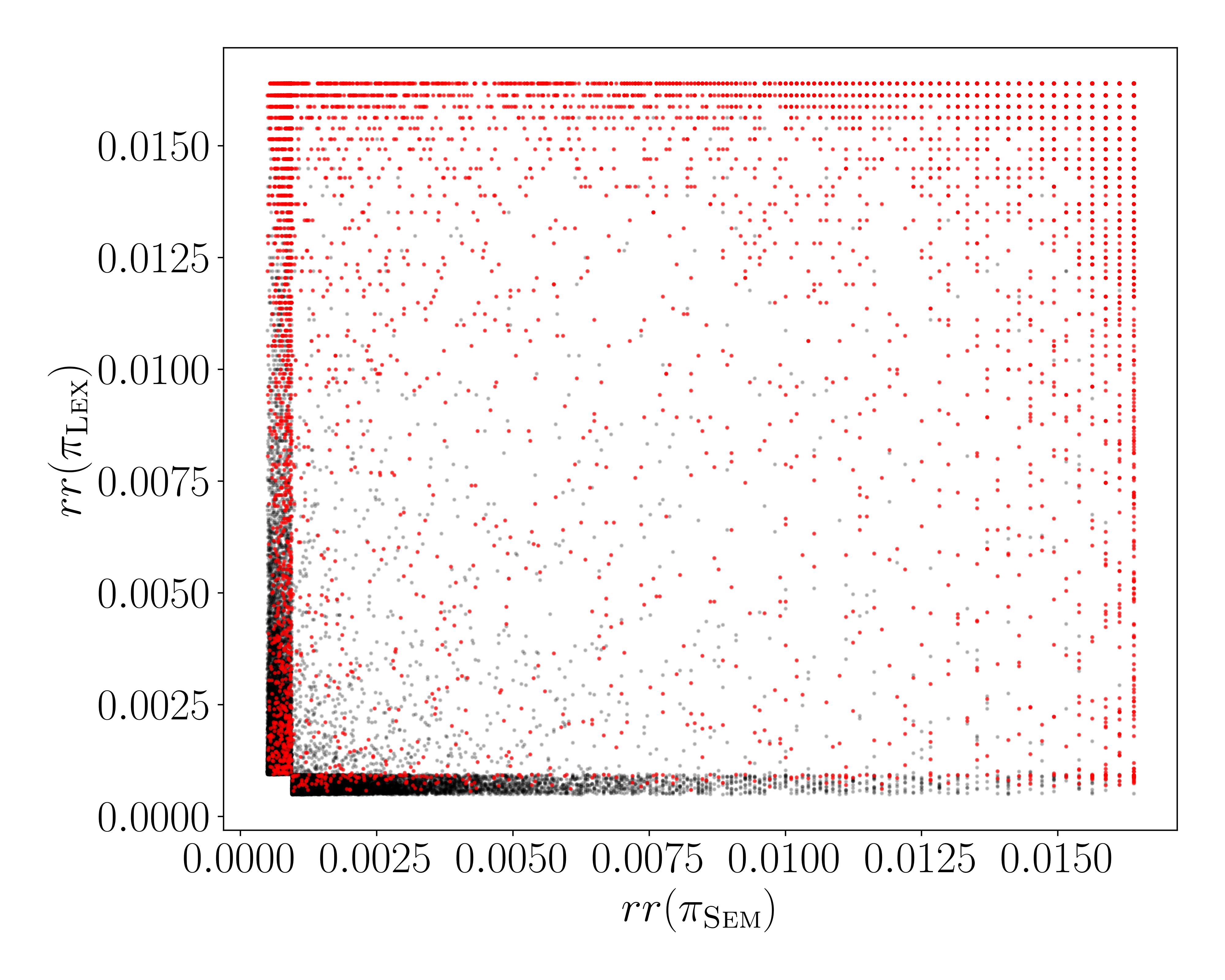}}
\subfloat[\textsc{Fever}]{
\includegraphics[height=1.6in]{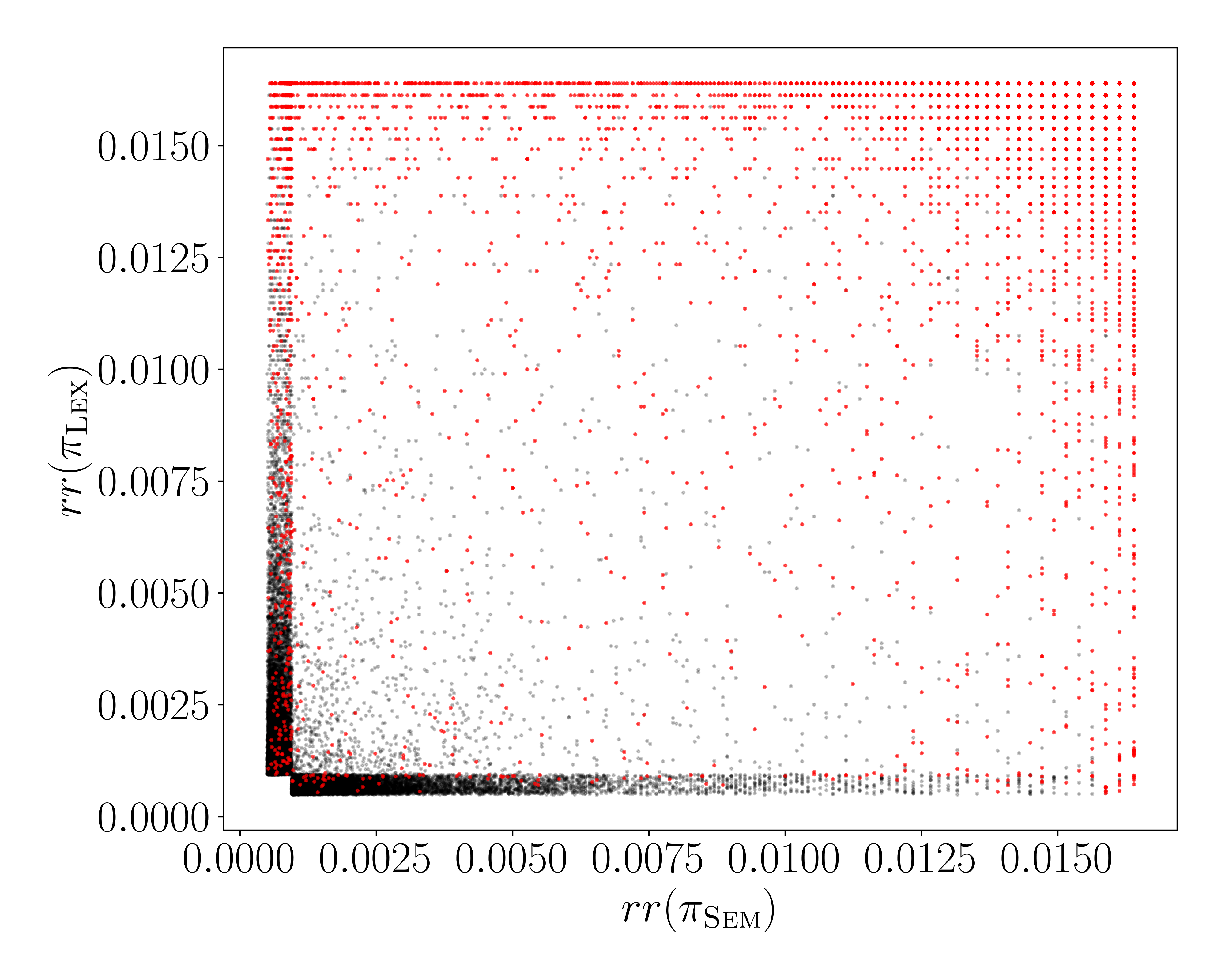}}
}
\caption{Visualization of the reciprocal rank determined by lexical ($rr(\pi_\textsc{Lex}) = 1 / (60 + \pi_\textsc{Lex})$) and semantic ($rr(\pi_\textsc{Sem}) = 1 / (60 + \pi_\textsc{Sem})$) retrieval for query-document pairs sampled from the validation split of each dataset. Shown in red are up to $20{,}000$ positive samples where document is relevant to query, and in black up to the same number of negative samples.}
\label{figure:rrf:suitability}
\end{center}
\end{figure*}
From the figure, 
we can see that samples are pulled towards one of the poles at $(0, 0)$ and $(1/61, 1/61)$. The former attracts a higher concentration of negative samples while the latter positive samples. While this separation is somewhat consistent across datasets, the concentration around poles and axes changes. Indeed, on \textsc{HotpotQA} and \textsc{Fever} there is a higher concentration of positive documents near the top, whereas on \textsc{FiQA} and the in-domain datasets more positive samples end up along the vertical line at $rr(\pi_\textsc{Sem}) = 1/61$, indicating that lexical ranks matter less. This suggests that a simple addition of reciprocal ranks does not behave consistently across domains.

We argued earlier that \ac{rrf} is parametric and that it, in fact, has as many parameters as there are retrieval functions to fuse. To see this more clearly, let us rewrite Equation~(\ref{equation:rrf}) as follows:
\begin{equation}
    f_\textsc{RRF}(q, d)= \frac{1}{\eta_\textsc{Lex} + \pi_\textsc{Lex}(q, d)} + \frac{1}{\eta_\textsc{Sem} + \pi_\textsc{Sem}(q, d)}.
    \label{equation:rrf_parametric}
\end{equation}

We study the effect of parameters on $f_\textsc{RRF}$ by comparing the NDCG obtained from \ac{rrf} with a particular choice of $\eta_\textsc{Lex}$ and $\eta_\textsc{Sem}$ against a realization of \ac{rrf} with $\eta_\textsc{Lex}=\eta_\textsc{Sem}=60$. In this way, we are able to visualize the impact on performance relative to the baseline configuration that is typically used in the literature. This difference in NDCG is rendered as a heatmap in Figure~\ref{figure:parameters} for select datasets---figures for all other datasets show a similar pattern. We note that, we select $\eta_\textsc{Lex}$ and $\eta_\textsc{Sem}$ from the set $\{1, 2, \ldots, 100\}$, but selectively illustrated a subset of values in Figure~\ref{figure:parameters} to make the figures more readable; the omitted combinations of $\eta_\textsc{Lex}$ and $\eta_\textsc{Sem}$ do not bring more insight than what can already be deduced from these figures.

\begin{figure}[t]
\begin{center}
\centerline{
\subfloat[MS MARCO]{
\includegraphics[height=2in]{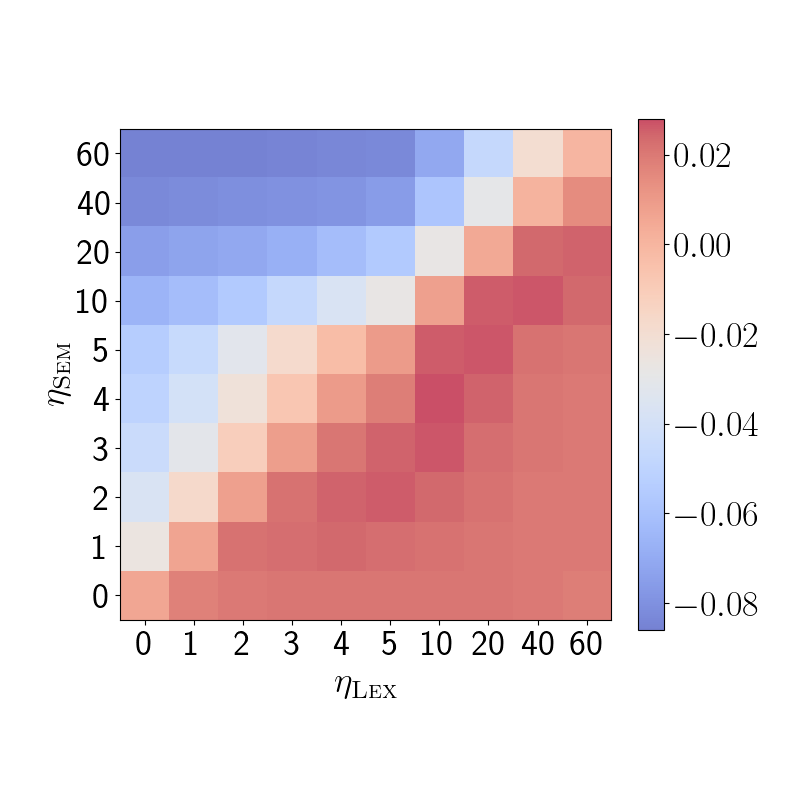}\label{figure:parameters:msmarco}}
\subfloat[\textsc{HotpotQA}]{
\includegraphics[height=2in]{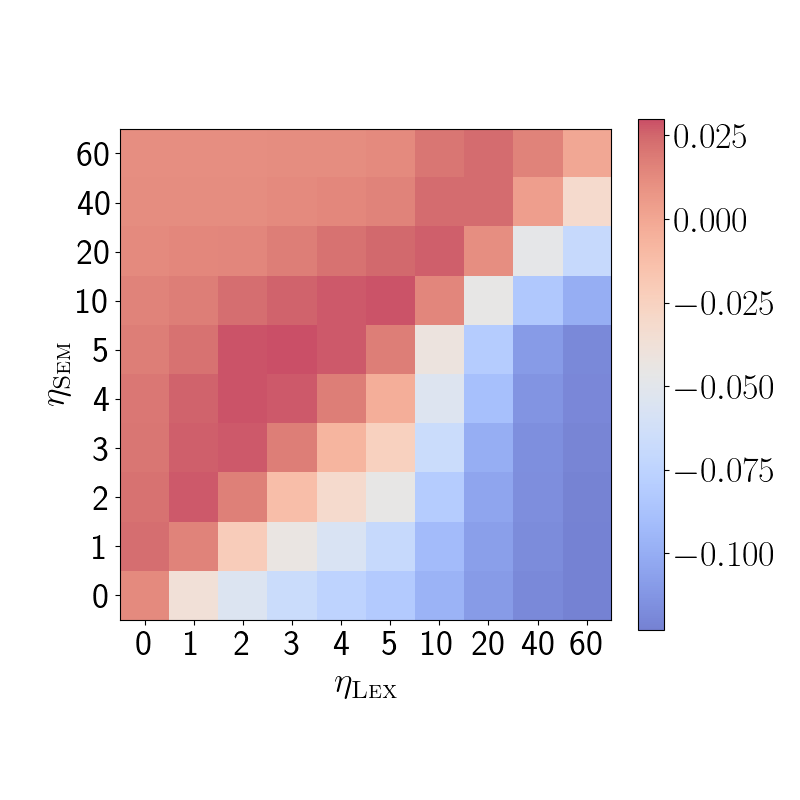}\label{figure:parameters:hotpotqa}}
}
\caption{Difference in NDCG@1000 of $f_\textsc{RRF}$ with distinct values $\eta_\textsc{Lex}$ and $\eta_\textsc{Sem}$, and $f_\textsc{RRF}$ with $\eta_\textsc{Lex}=\eta_\textsc{Sem}=60$ (positive indicates better ranking quality by the former). On MS MARCO, an in-domain dataset, NDCG improves when $\eta_\textsc{Lex} > \eta_\textsc{Sem}$, while the opposite effect can be seen for \textsc{HotpotQA}, an out-of-domain dataset.}
\label{figure:parameters}
\end{center}
\end{figure}

\begin{figure}[t]
\begin{center}
\centerline{
\subfloat[$\eta_\textsc{Lex}=60$, $\eta_\textsc{Sem}=60$]{
\includegraphics[height=1.4in]{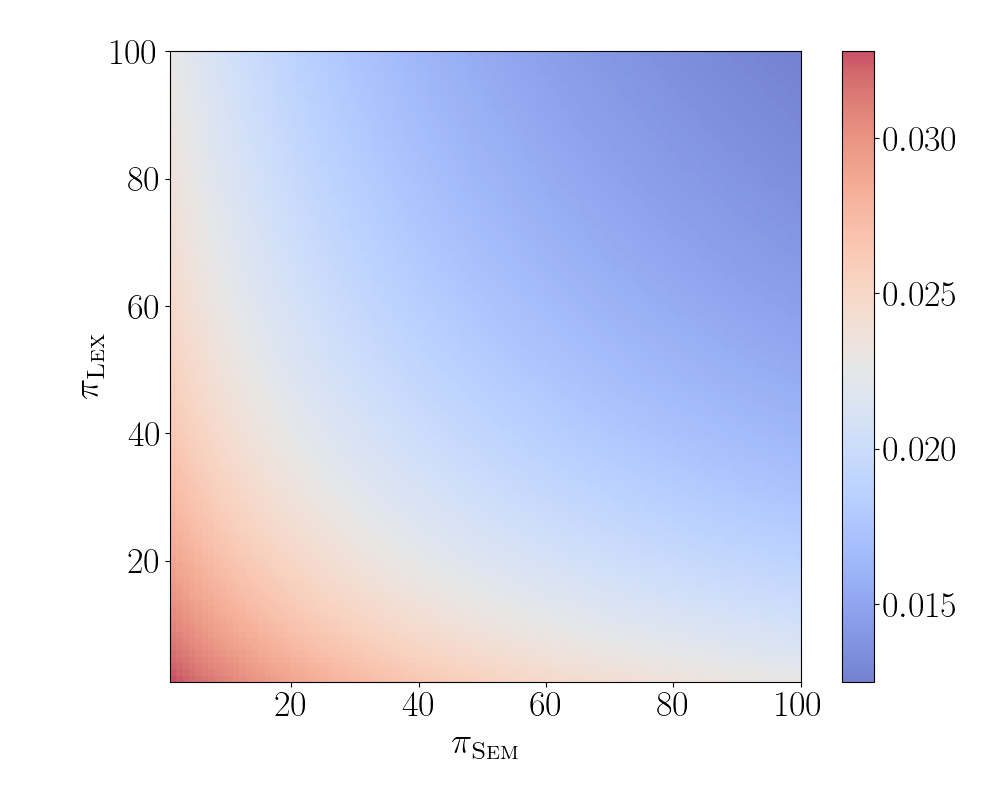}}
\subfloat[$\eta_\textsc{Lex}=10$, $\eta_\textsc{Sem}=4$]{
\includegraphics[height=1.4in]{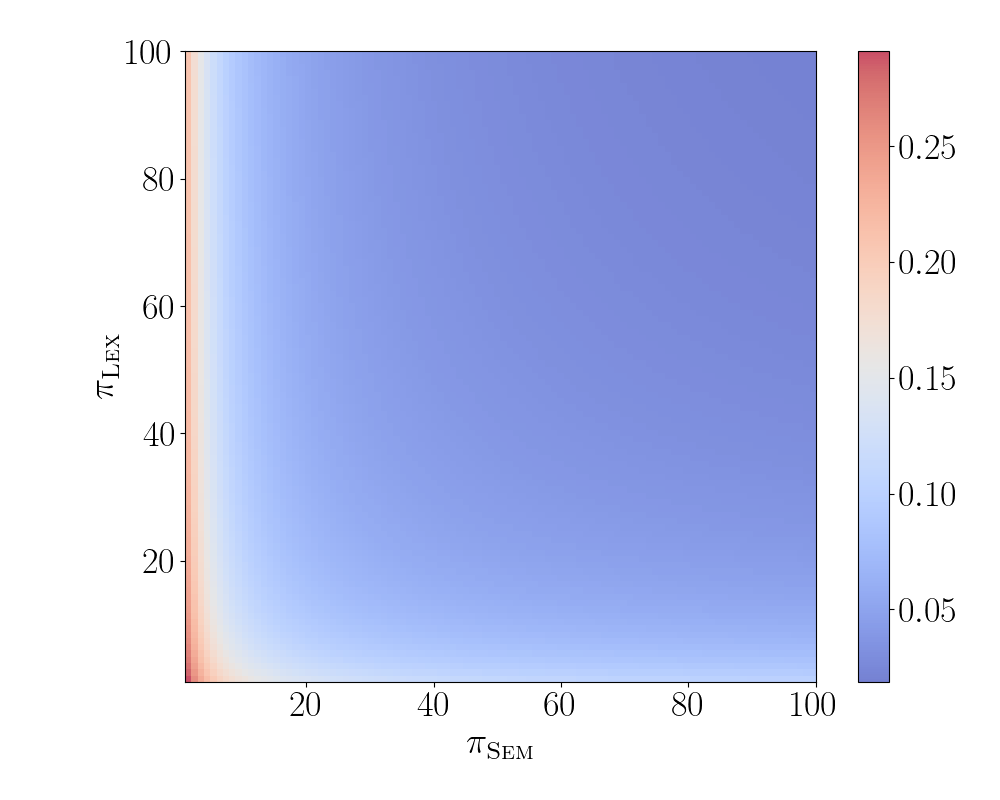}}
\subfloat[\textsc{$\eta_\textsc{Lex}=3$, $\eta_\textsc{Sem}=5$}]{
\includegraphics[height=1.4in]{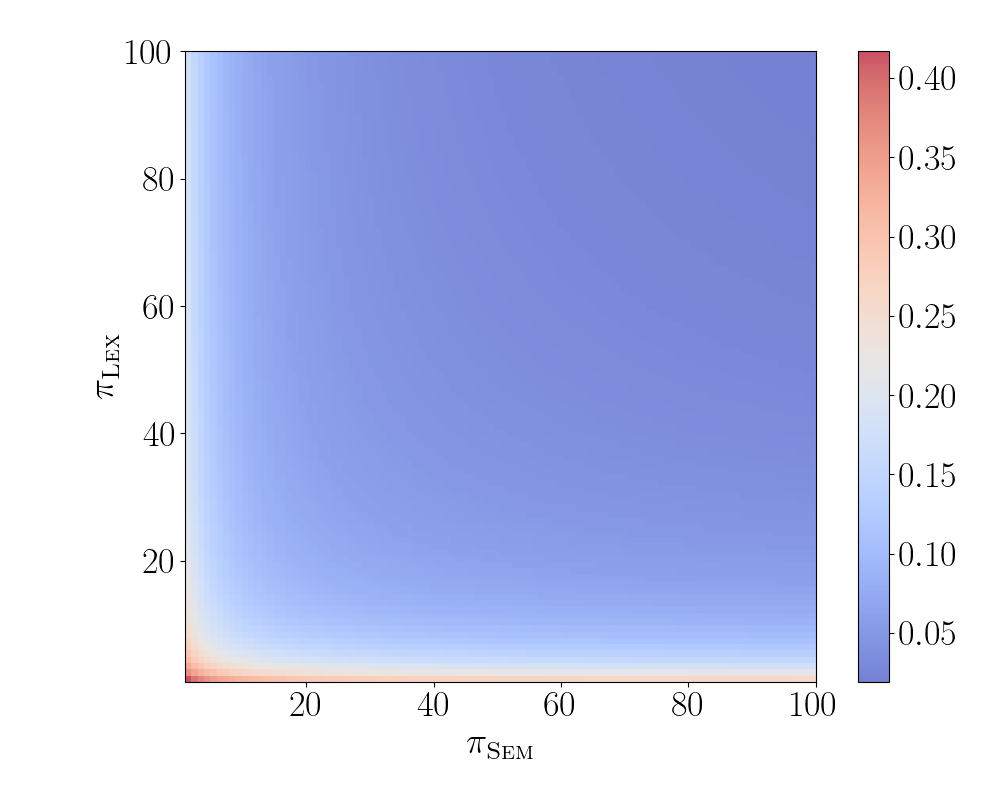}}
}
\caption{Effect of $f_\textsc{RRF}$ with select configurations of $\eta_\textsc{Lex}$ and $\eta_\textsc{Sem}$ on pairs of ranks from lexical and semantic systems. When $\eta_\textsc{Lex} > \eta_\textsc{Sem}$, the fusion function discounts the lexical system's contribution.}
\label{figure:rrf:reference}
\end{center}
\end{figure}

As a general observation, we note that NDCG swings wildly as a function of \ac{rrf} parameters. Crucially, performance improves off-diagonal, where the parameter takes on different values for the semantic and lexical components. On MS MARCO, shown in Figure~\subref*{figure:parameters:msmarco}, NDCG improves when $\eta_\textsc{Lex} > \eta_\textsc{Sem}$, while the opposite effect can be seen for \textsc{HotpotQA}, an out-of-domain dataset. This can be easily explained by the fact that increasing $\eta_\textsc{o}$ for retrieval system $\textsc{o}$ effectively discounts the contribution of ranks from $\textsc{o}$ to the final hybrid score. On in-domain datasets where the semantic model already performs strongly, for example, discounting the lexical system by increasing $\eta_\textsc{Lex}$ leads to better performance.

\begin{table*}[t]
\caption{Mean NDCG@1000 (NDCG@100 for \textsc{SciFact} and \textsc{NFCorpus}) on the test split of various datasets for hybrid retrieval using TM2C2 ($\alpha=0.8$) and \ac{rrf} ($\eta_\textsc{Lex}$, $\eta_\textsc{Sem}$). The symbols $\ddagger$ and $\ast$ indicate statistical significance ($p$-value $< 0.01$) with respect to TM2C2 and baseline RRF ($60, 60$) respectively, according to a paired two-tailed $t$-test.}
\label{table:tuned_rrf}
\begin{center}
\begin{sc}
\begin{tabular}{cc|cc|cc}
& & \multicolumn{4}{c}{NDCG} \\
\toprule
& Dataset & TM2C2 & \ac{rrf} $(60, 60)$ & \ac{rrf} $(5, 5)$ & \ac{rrf} $(10, 4)$\\
\midrule
\parbox[t]{2mm}{\multirow{2}{*}{\rotatebox[origin=c]{90}{in-domain}}}
& MS MARCO & \textbf{0.454} & $0.425^{\ddagger}$ & $0.435^{\ddagger\ast}$ & $0.451^{\ast}$ \\
& NQ & \textbf{0.542} & $0.514^{\ddagger}$ & $0.521^{\ddagger\ast}$ & $0.528^{\ddagger\ast}$ \\
& Quora & \textbf{0.901} & $0.877^{\ddagger}$ & $0.885^{\ddagger\ast}$ &  $0.896^{\ast}$ \\
& & & &  \\
\midrule
\parbox[t]{2mm}{\multirow{6}{*}{\rotatebox[origin=c]{90}{zero-shot}}}
& NFCorpus & \textbf{0.327} & $0.312^{\ddagger}$ & $0.318^{\ddagger\ast}$ & $0.310^{\ddagger}$ \\
& HotpotQA & \textbf{0.699} & $0.675^{\ddagger}$ & $0.693^{\ast}$ & $0.621^{\ddagger\ast}$ \\
& FEVER & \textbf{0.744} & $0.721^{\ddagger}$ & $0.727^{\ddagger\ast}$ & $0.649^{\ddagger\ast}$ \\
& SciFact & \textbf{0.753} & $0.730^{\ddagger}$ & $0.738^{\ddagger}$ & $0.715^{\ddagger\ast}$\\
& DBPedia & \textbf{0.512} & $0.489^{\ddagger}$ & $0.489^{\ddagger}$ & $0.480^{\ddagger\ast}$\\
& FiQA & \textbf{0.496} & $0.464^{\ddagger}$ & $0.470^{\ddagger\ast}$ & $0.482^{\ddagger\ast}$ \\
\bottomrule
\end{tabular}
\end{sc}
\end{center}
\end{table*}

Having observed that tuning \ac{rrf} potentially leads to gains in NDCG, we ask if tuned parameters generalize on out-of-domain datasets. To investigate that question, we tune \ac{rrf} on in-domain datasets and pick the value of parameters that maximize NDCG on the validation split of in-domain datasets, and measure the performance of the resulting function on the test split of all (in-domain and out-of-domain) datasets. We present the results in Table~\ref{table:tuned_rrf}. While tuning a parametric RRF does indeed lead to gains in NDCG on in-domain datasets, the tuned function does not generalize well to out-of-domain datasets.

The poor generalization can be explained by the reversal of patterns observed in Figure~\ref{figure:parameters} where $\eta_\textsc{Lex} > \eta_\textsc{Sem}$ suits in-domain datasets better but the opposite is true for out-of-domain datasets. By modifying $\eta_\textsc{Lex}$ and $\eta_\textsc{Sem}$ we modify the fusion of ranks and boost certain regions and discount others in an imbalanced manner. Figure~\ref{figure:rrf:reference} visualizes this effect on $f_\textsc{RRF}$ for particular values of its parameters. This addresses RQ2.

\subsection{Effect of Lipschitz Continuity}
In the previous section, we stated an intuition that because \ac{rrf} does not preserve the distribution of raw scores, it loses valuable information in the process of fusing retrieval systems. In our final research question, RQ3, we investigate if this indeed matters in practice.

The notion of ``preserving'' information is well captured by the concept of Lipschitz continuity.\footnote{A function $f$ is Lipschitz continous with constant $L$ if $|| f(y) - f(x) ||_o \leq L ||y - x||_i$ for some norm $||\cdot||_o$ and $||\cdot||_i$ on the output and input space of $f$.} When a function is Lipschitz continuous with a small Lipschitz constant, it does not oscillate wildly with a small change to its input. Intuitively, and in the context of this work, this means that a small change to the individual scores does not lead to a sudden and exaggerated effect on the fused score; the fusion function does not dramatically distort the distribution of scores or ranks, an effect we characterize informally as ``preserving information.'' \ac{rrf} does not have this property because the moment one lexical (or semantic) score becomes larger than another the function makes a hard transition to a new value.

We can therefore cast RQ3 as a question of whether Lipschitz continuity is an important property in practice. To put that hypothesis to the test, we design a smooth approximation of \ac{rrf} using known techniques~\cite{bruch2019approxndcg,qin2010general}.

As expressed in Equation~(\ref{equation:rank}), the rank of a document is simply the sum of indicators. It is thus trivial to approximate this quantity using a generalized sigmoid with parameter $\beta$: $\sigma_\beta(x) = 1 / (1 + \exp(-\beta x))$. As $\beta$ approaches $1$, the sigmoid takes its usual \emph{S} shape, while $\beta \rightarrow \infty$ produces a very close approximation of the indicator. Interestingly, the Lipschitz constant of $\sigma_\beta(\cdot)$ is, in fact, $\beta$. As $\beta$ increases, the approximation of ranks becomes more accurate, but the Lipschitz constant becomes larger. When $\beta$ is too small, however, the approximation breaks down but the function transitions more slowly, thereby preserving much of the characteristics of the underlying data distribution.

\ac{rrf} being a function of ranks can now be approximated by plugging in approximate ranks in Equation~(\ref{equation:rrf}), resulting in SRRF:
\begin{equation}
    f_\textsc{SRRF}(q, d)= \frac{1}{\eta + \tilde{\pi}_\textsc{Lex}(q, d)} + \frac{1}{\eta + \tilde{\pi}_\textsc{Sem}(q, d)},
    \label{equation:srrf}
\end{equation}
where $\tilde{\pi}_\textsc{o}(q, d_i) = 0.5 + \sum_{d_j \in R^{k}_\textsc{o}(q)} \sigma_\beta(f_\textsc{o}(q, d_j) - f_\textsc{o}(q, d_i))$.
By increasing $\beta$ we increase the Lipschitz constant of $f_\textsc{SRRF}$. This is the lever we need to test the idea that Lipschitz continuity matters and that functions that do not distort the distributional properties of raw scores lead to better ranking quality.

\begin{figure}[t]
\begin{center}
\centerline{
\subfloat[in-domain]{
\includegraphics[height=2in]{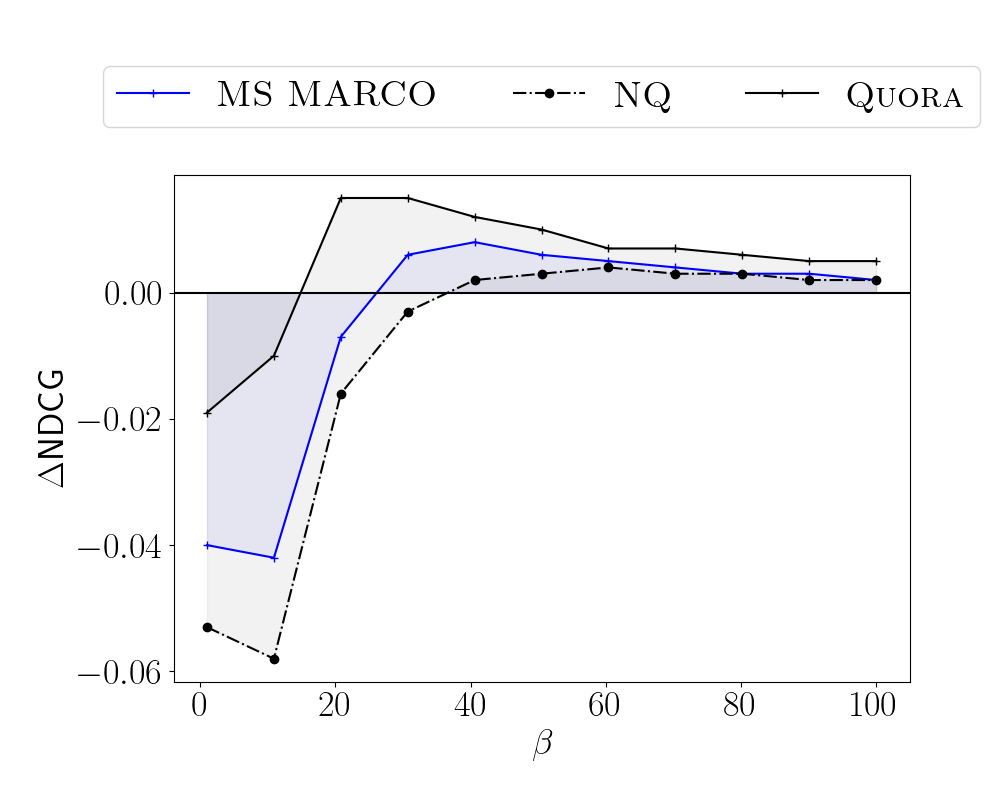}\label{figure:lipschitz:in_domain}}
\subfloat[out-of-domain]{
\includegraphics[height=2in]{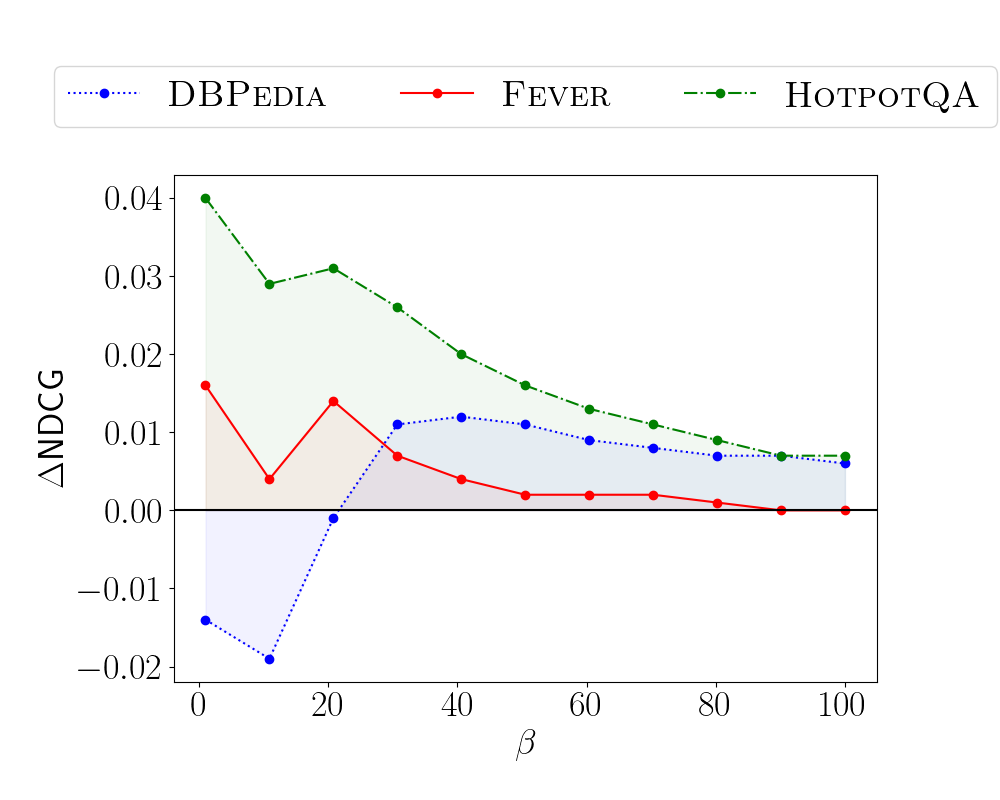}\label{figure:lipschitz:out_of_domain}}
}
\caption{The difference in NDCG@1000 of $f_\textsc{SRRF}$ and $f_\textsc{RRF}$ with $\eta=60$ (positive indicates better ranking quality by \textsc{SRRF}) as a function of $\beta$.}
\label{figure:lipschitz}
\end{center}
\end{figure}

\begin{figure}[t]
\begin{center}
\centerline{
\subfloat[in-domain]{
\includegraphics[height=2in]{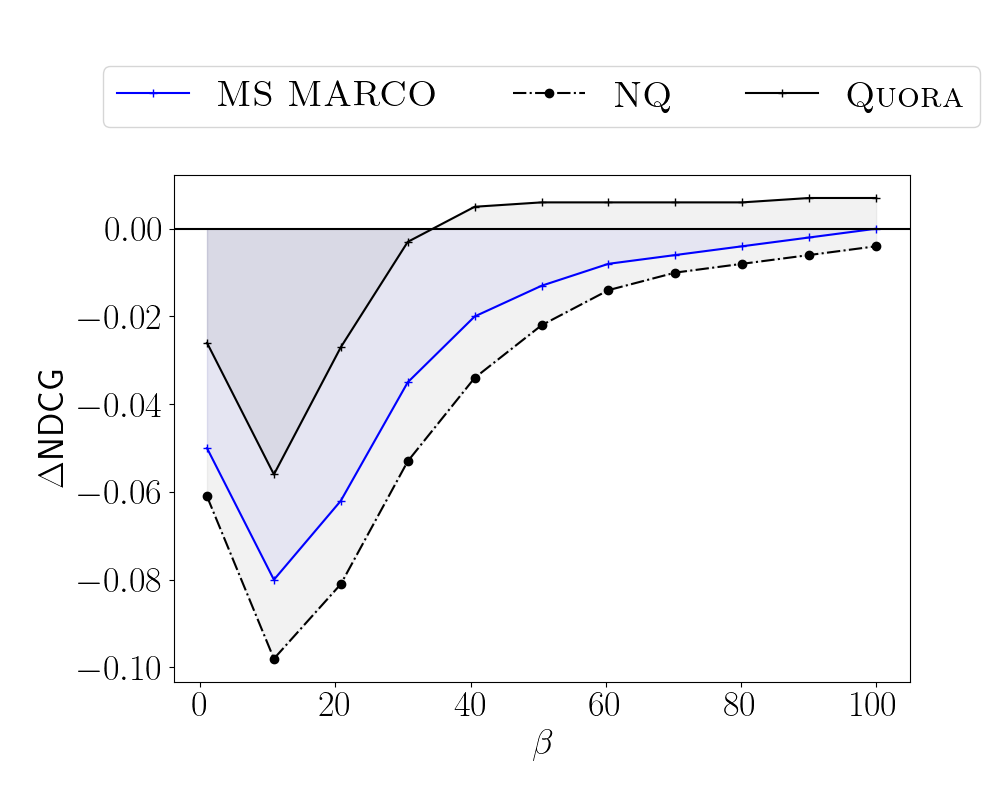}\label{figure:lipschitz:in_domain}}
\subfloat[out-of-domain]{
\includegraphics[height=2in]{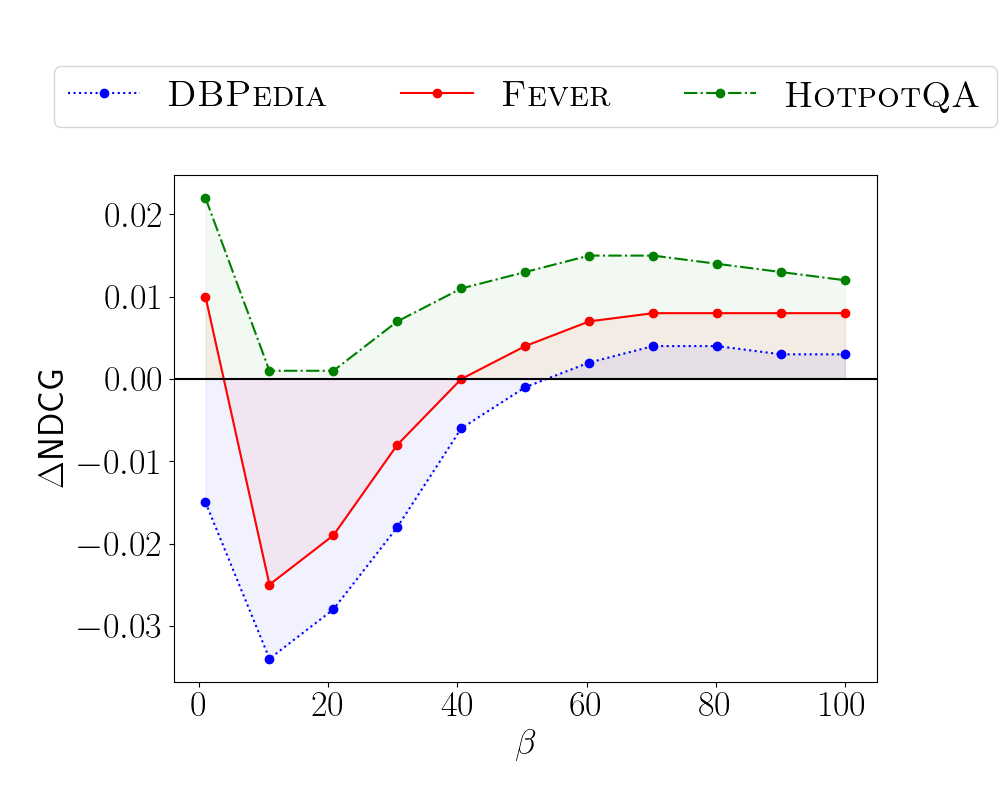}\label{figure:lipschitz:out_of_domain}}
}
\caption{The difference in NDCG@1000 of $f_\textsc{SRRF}$ and $f_\textsc{RRF}$ with $\eta=5$ (positive indicates better ranking quality by \textsc{SRRF}) as a function of $\beta$.}
\label{figure:lipschitz:eta5}
\end{center}
\end{figure}

Figures~\ref{figure:lipschitz} and~\ref{figure:lipschitz:eta5} visualize the difference between SRRF and \ac{rrf} for two settings of $\eta$ selected based on the results in Table~\ref{table:tuned_rrf}. As anticipated, when $\beta$ is too small, the approximation error is large and ranking quality degrades. As $\beta$ becomes larger, ranking quality trends in the direction of \ac{rrf}. Interestingly, as $\beta$ becomes gradually smaller, the performance of SRRF improves over the \ac{rrf} baseline. This effect is more pronounced for the $\eta=60$ setting of \ac{rrf}, as well as on the out-of-domain datasets. Empirical evidence reported for other fusions in Appendices~\ref{appendix:splade-bm25} through~\ref{appendix:tasb-minilm} shows this finding to be robust.

While we acknowledge the possibility that the approximation in Equation~(\ref{equation:srrf}) may cause a change in ranking quality, we expected that change to be a degradation, not an improvement. However, given we do observe gains by smoothing the function, and that the only other difference between SRRF and \ac{rrf} is their Lipschitz constant, we believe these results highlight the role of Lipschitz continuity in ranking quality. For completeness, we have also included a comparison of SRRF, \ac{rrf}, and TM2C2 in Table~\ref{table:rq3}.

\begin{table*}[t]
\caption{Mean NDCG@1000 (NDCG@100 for \textsc{SciFact} and \textsc{NFCorpus}) on the test split of various datasets for hybrid retrieval using TM2C2 ($\alpha=0.8$), \ac{rrf} ($\eta$), and SRRF($\eta$, $\beta$). The parameters $\beta$ are fixed to values that maximize NDCG on the validation split of in-domain datasets. The symbols $\ddagger$ and $\ast$ indicate statistical significance ($p$-value $< 0.01$) with respect to TM2C2 and RRF respectively, according to a paired two-tailed $t$-test.}
\label{table:rq3}
\begin{center}
\begin{sc}
\begin{tabular}{cc|c|cc|cc}
& & \multicolumn{3}{c}{NDCG} \\
\toprule
& Dataset & TM2C2 & \ac{rrf}($60$) & SRRF ($60$, $40$) & \ac{rrf}($5$) & SRRF ($5$, $100$) \\
\midrule
\parbox[t]{2mm}{\multirow{2}{*}{\rotatebox[origin=c]{90}{in-domain}}}
& MS MARCO & \textbf{0.454} & $0.425^{\ddagger}$ & $0.431^{\ddagger\ast}$ & $0.435^{\ddagger}$ & $0.431^{\ddagger\ast}$ \\
& NQ & \textbf{0.542} & $0.514^{\ddagger}$ & $0.516^{\ddagger}$ & $0.521^{\ddagger}$ & $0.517^{\ddagger}$ \\
& Quora & \textbf{0.901} & $0.877^{\ddagger}$ & $0.889^{\ddagger\ast}$ & $0.885^{\ddagger}$ & $0.889^{\ddagger\ast}$ \\
& & & &  \\
\midrule
\parbox[t]{2mm}{\multirow{6}{*}{\rotatebox[origin=c]{90}{zero-shot}}}
& NFCorpus & \textbf{0.327} & $0.312^{\ddagger}$ & $0.323^{\ddagger\ast}$ & $0.318^{\ddagger}$ & $0.322^{\ddagger}$\\
& HotpotQA & 0.699 & $0.675^{\ddagger}$ & $0.695^{\ast}$ & $0.693^{\ddagger}$ & $\textbf{0.705}^{\ddagger\ast}$ \\
& FEVER & \textbf{0.744} & $0.721^{\ddagger}$ & $0.725^{\ddagger}$ & $0.727^{\ddagger}$ & $0.735^{\ddagger\ast}$ \\
& SciFact & \textbf{0.753} & $0.730^{\ddagger}$ & $0.740^{\ddagger}$ & $0.738^{\ddagger}$ & $0.740^{\ddagger}$ \\
& DBPedia & \textbf{0.512} & $0.489^{\ddagger}$ & $0.501^{\ddagger\ast}$ & $0.489^{\ddagger}$ & $0.492^{\ddagger}$ \\
& FiQA & \textbf{0.496} & $0.464^{\ddagger}$ & $0.468^{\ddagger}$ & $0.470^{\ddagger}$ & $0.469^{\ddagger}$ \\
\bottomrule
\end{tabular}
\end{sc}
\end{center}
\end{table*}

\section{Discussion}
\label{section:discussion}

The analysis in this work motivates us to identify and document the properties of a well-behaved fusion function, and present the principles that, we hope, will guide future research in this space. These desiderata are stated below.

\textbf{Monotonicity}: When $f_\textsc{o}$ is positively correlated with a target ranking metric (i.e., ordering documents in decreasing order of $f_\textsc{o}$ must lead to higher quality), then it is natural to require that $f_\textsc{Hybrid}$ be monotone increasing in its arguments. We have already seen and indeed used this property in our analysis of the convex combination fusion function. It is trivial to show why this property is crucial.

\textbf{Homogeneity}: The order induced by a fusion function must be unaffected by a positive re-scaling of query and document vectors. That is: $f_\textsc{Hybrid}(q, d) \overset{\pi}{=} f_\textsc{Hybrid}(q, \gamma d) \overset{\pi}{=} f_\textsc{Hybrid}(\gamma q,d)$ where $\overset{\pi}{=}$ denotes rank-equivalence and $\gamma > 0$. This property prevents any retrieval system from inflating its contribution to the final hybrid score by simply boosting its document or query vectors.

\textbf{Boundedness}: Recall that, a convex combination without score normalization is often ineffective and inconsistent because BM25 is unbounded and that lexical and semantic scores are on different scales. To see this effect we turn to Figure~\ref{figure:boundedness}.

We observe in Figure~\subref*{figure:boundedness:in_domain} that, for in-domain datasets, adding the unnormalized lexical scores using a convex combination leads to a severe \emph{degradation} of ranking quality. We believe this is because of the fact that the semantic retrieval model, which is fine-tuned on these datasets, already produces ranked lists of high quality, and that adding the lexical scores which are on a very different scale distorts the rankings and leads to poor performance. In out-of-domain experiments as shown in Figure~\subref*{figure:boundedness:out_of_domain}, however, the addition of lexical scores leads to often significant gains in quality. We believe this can be explained exactly as the in-domain observations: the semantic model generally does poorly on out-of-domain datasets while the lexical retriever does well. But because the semantic scores are bounded and relatively small, they do not significantly distort the rankings produced by the lexical retriever.

\begin{figure}[t]
\begin{center}
\centerline{
\subfloat[in-domain]{
\includegraphics[height=2in]{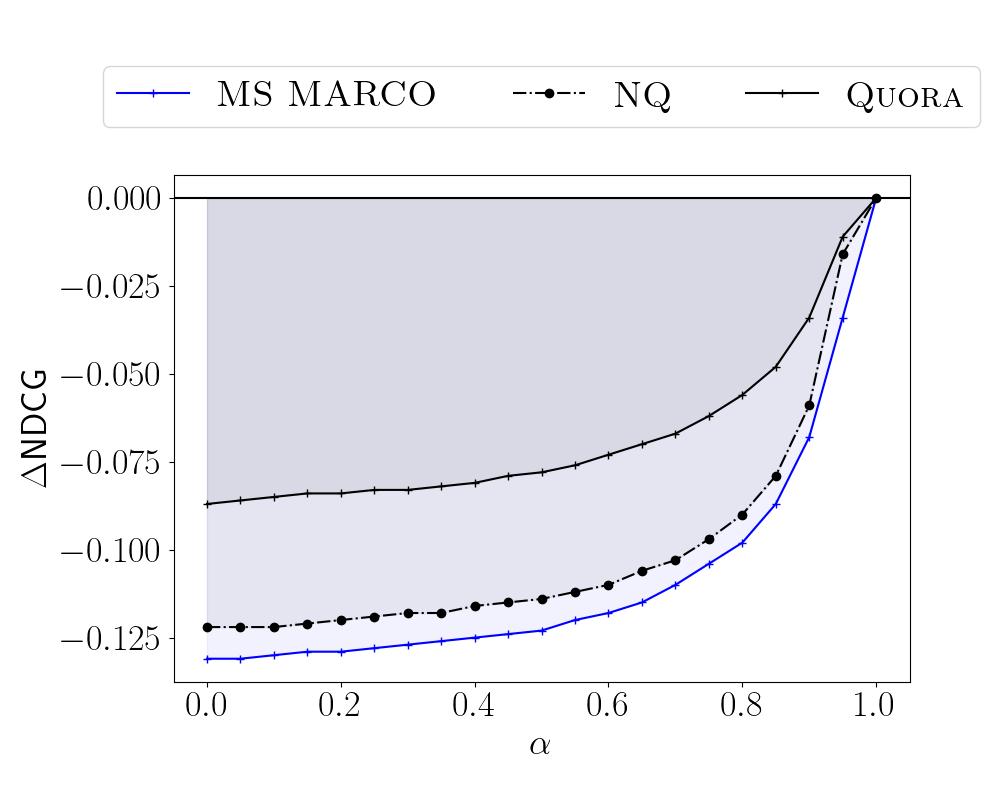}\label{figure:boundedness:in_domain}}
\subfloat[out-of-domain]{
\includegraphics[height=2in]{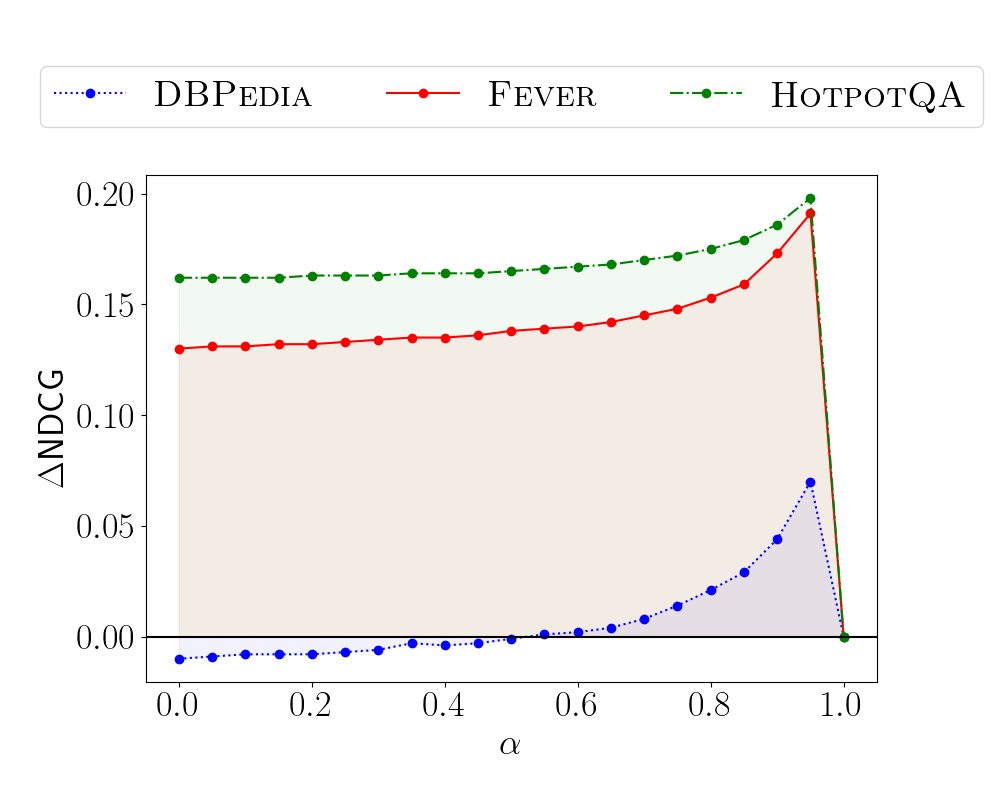}\label{figure:boundedness:out_of_domain}}
}
\caption{The difference in NDCG of convex combination of unnormalized scores and a pure semantic search (positive indicates better ranking quality by a convex combination) as a function of $\alpha$.}
\label{figure:boundedness}
\end{center}
\end{figure}

To avoid that pitfall, we require that $f_\textsc{Hybrid}$ be bounded: $|f_\textsc{Hybrid}| \leq M$ for some $M > 0$. As we have seen before, normalizing the raw scores addresses this issue.

\textbf{Lipschitz Continuity}: We argued that because \ac{rrf} does not take into consideration the raw scores, it distorts their distribution and thereby loses valuable information. On the other hand, TM2C2 (or any convex combination of scores) is a smooth function of scores and preserves much of the characteristics of its underlying distribution. We formalized this idea using the notion of Lipschitz continuity: A larger Lipschitz constant leads to a larger distortion of retrieval score distribution.

\begin{figure}[!h]
\begin{center}
\centerline{
\subfloat[MS MARCO]{
\includegraphics[width=0.45\linewidth]{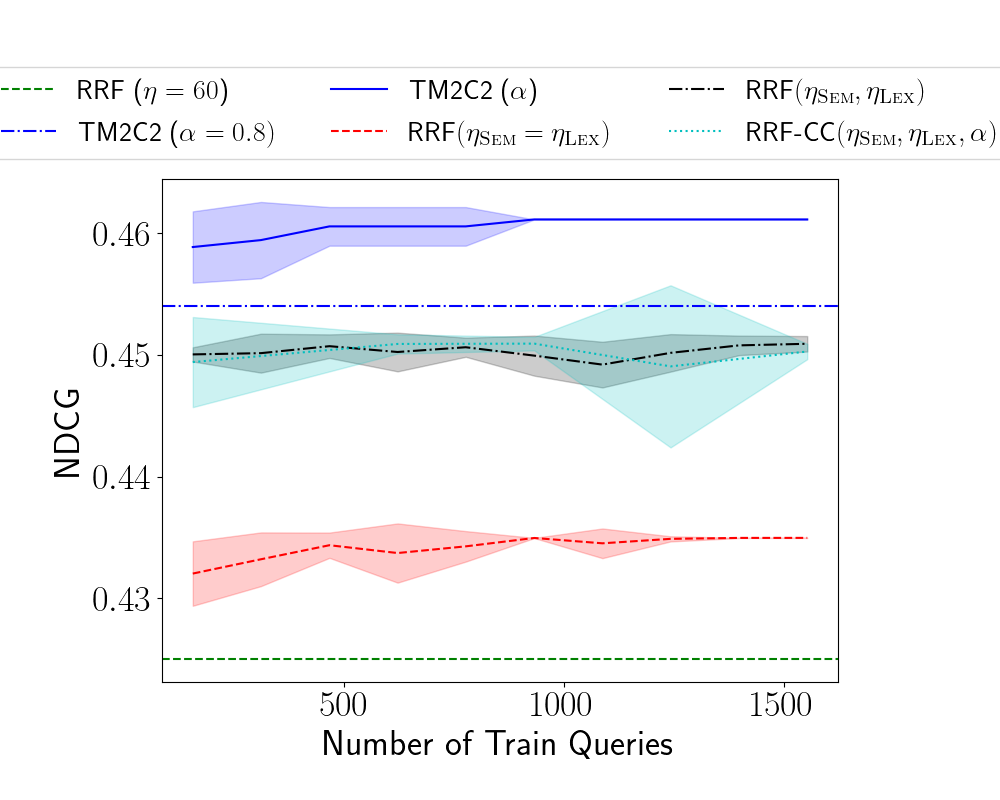}}
\subfloat[\textsc{Quora}]{
\includegraphics[width=0.45\linewidth]{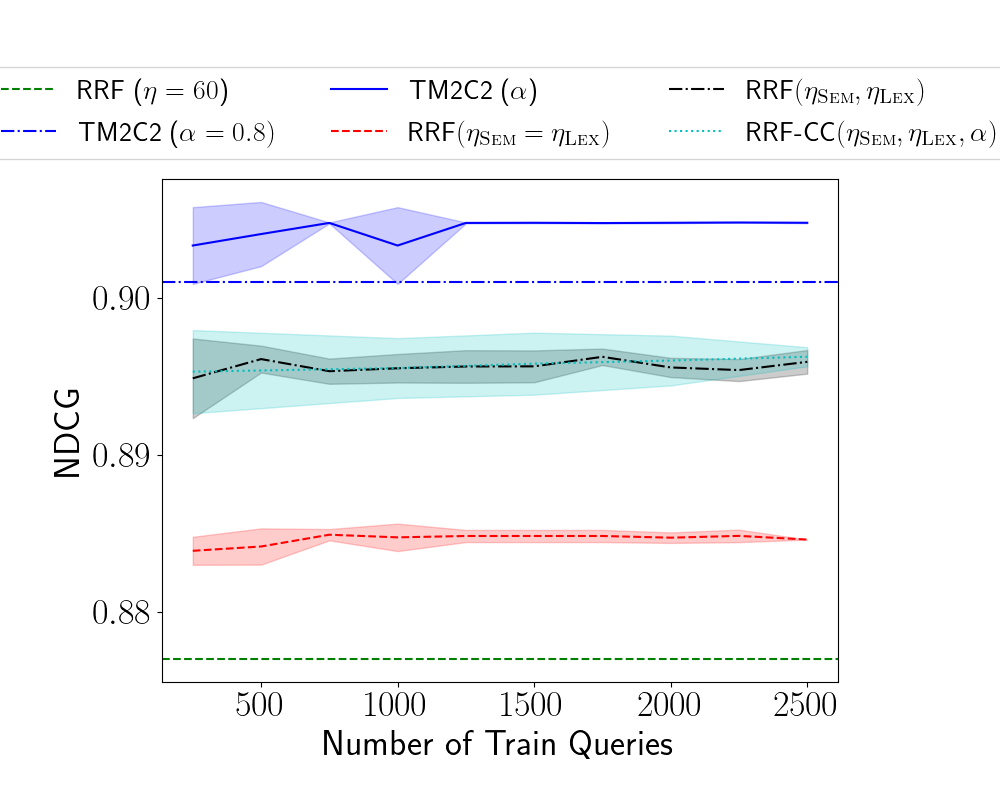}}
}
\centerline{
\subfloat[\textsc{HotpotQA}]{
\includegraphics[trim={0 0 0 4.2cm},clip,width=0.45\linewidth]{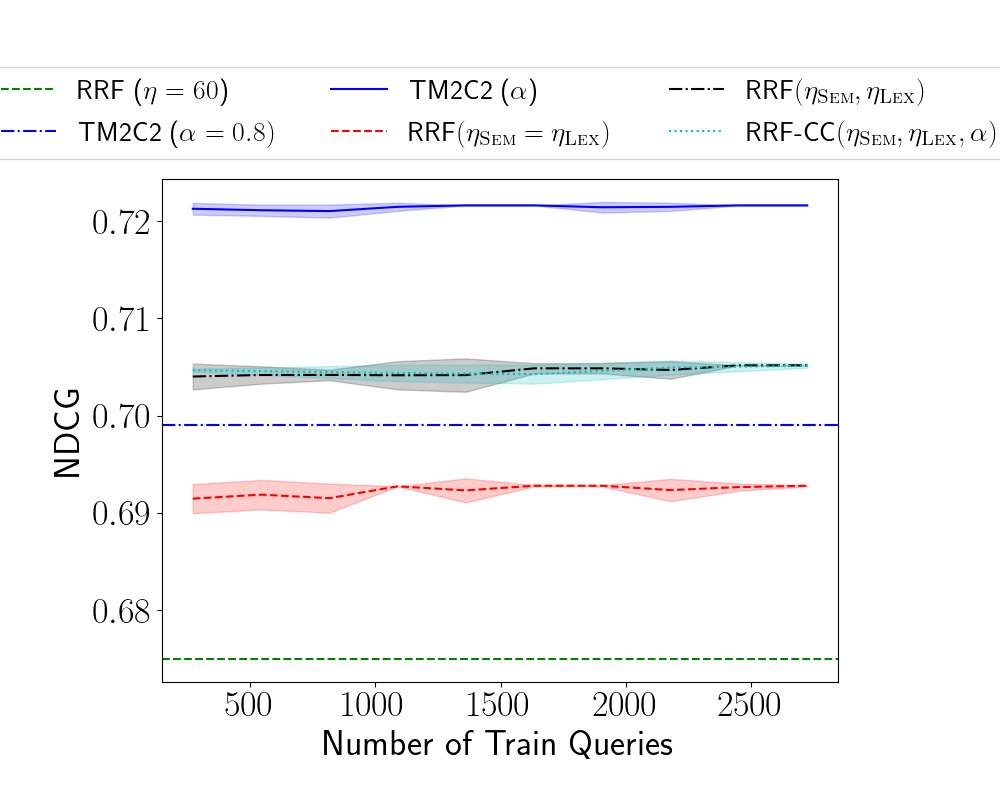}}
\subfloat[\textsc{Fever}]{
\includegraphics[trim={0 0 0 4.2cm},clip,width=0.45\linewidth]{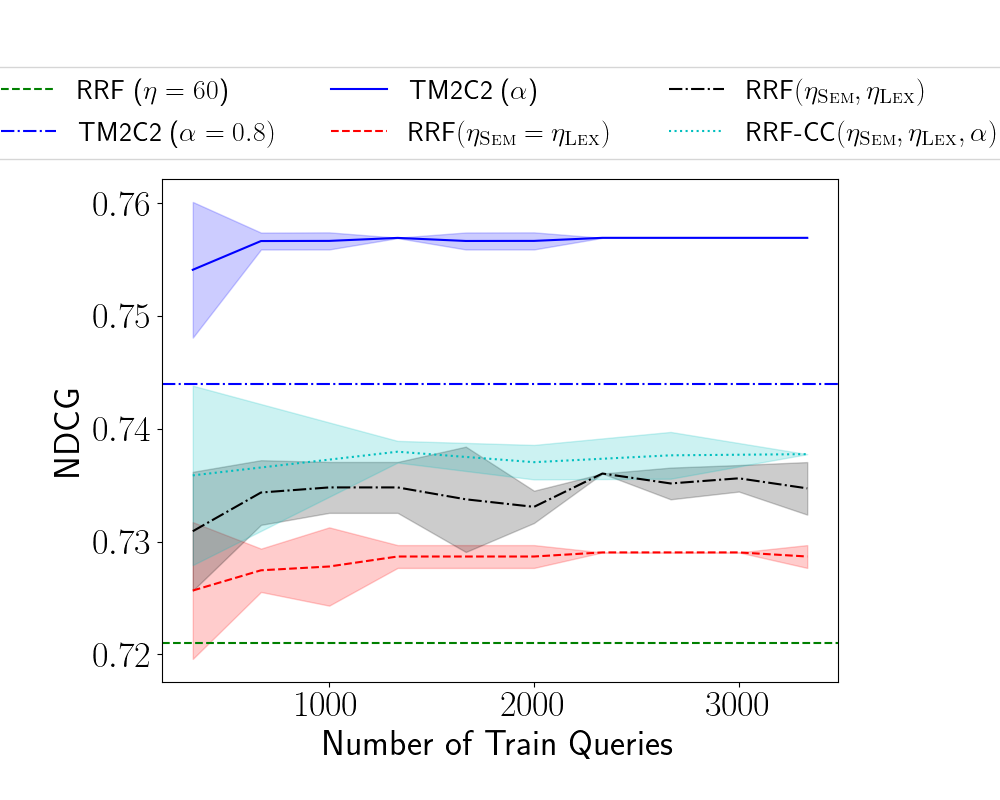}}
}
\centerline{
\subfloat[\textsc{NFCorpus}]{
\includegraphics[trim={0 0 0 4.2cm},clip,width=0.45\linewidth]{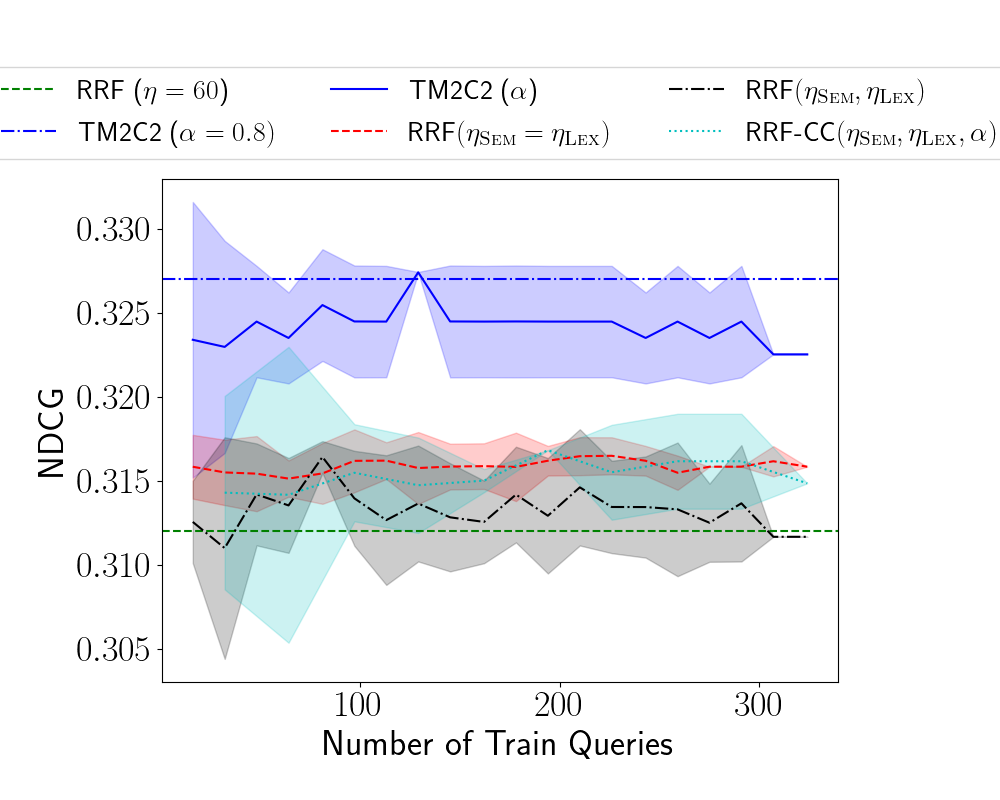}}
\subfloat[\textsc{FiQA}]{
\includegraphics[trim={0 0 0 4.2cm},clip,width=0.45\linewidth]{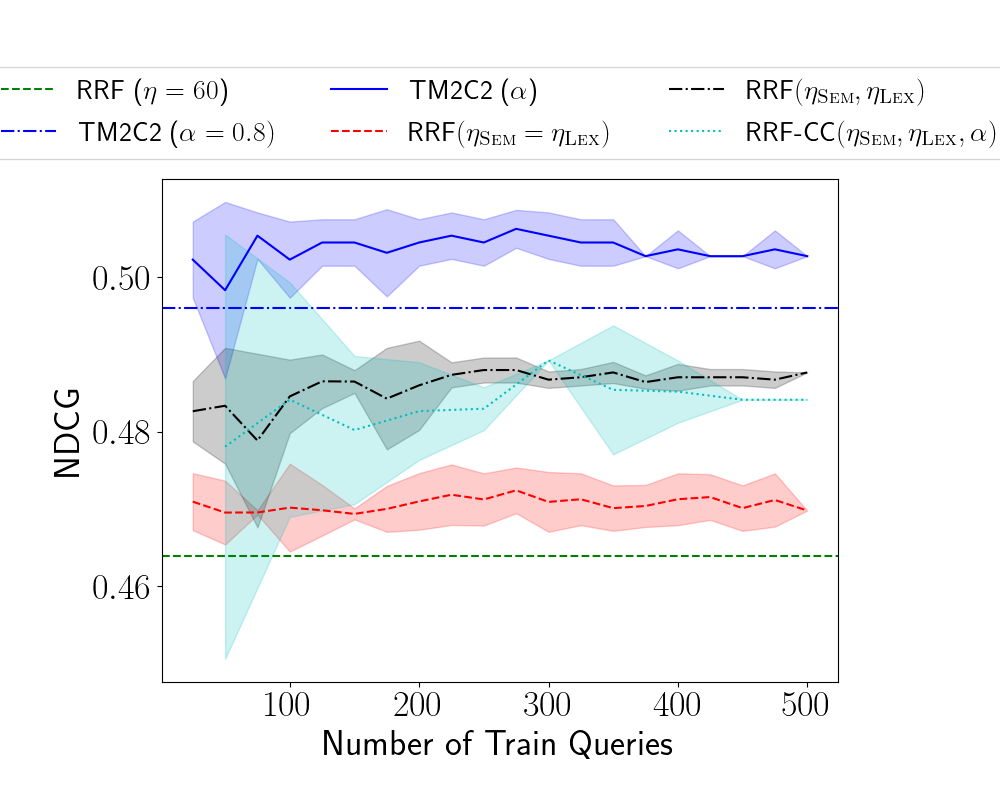}}
}
\caption{Sample efficiency of TM2C2 and the parameterized variants of RRF (single parameter where $\eta_\textsc{Sem}=\eta_\textsc{Lex}$, and two parameters where we allow different values of $\eta_\textsc{Sem}$ and $\eta_\textsc{Lex}$, and a third variation that is a convex combination of \ac{rrf} terms defined in Equation~\ref{equation:rrf-cc}). We sample progressively larger subsets of the training set (with a step size of $5\%$), tune the parameters of each function on the resulting set, and evaluate the resulting function on the test split. These figures depict NDCG@1000 as a function of the size of the tuning set, averaged over $5$ trials with the shaded regions illustrating the $95\%$ confidence intervals. For reference, we have also plotted NDCG on the test split for \ac{rrf} ($\eta=60$) and TM2C2 with $\alpha=0.8$ from Table~\ref{table:rq1}.}
\label{figure:sample_efficiency}
\end{center}
\end{figure}

\textbf{Interpretability and Sample Efficiency}: The question of hybrid retrieval is an important topic in \ac{ir}. What makes it particularly pertinent is its zero-shot applicability, a property that makes deep models \emph{reusable}, reducing computational costs and emissions as a result~\cite{bruch2022reneuir,scells2022sigir-green-ir}, and enabling resource-constrained research labs to innovate. Given the strong evidence supporting the idea that hybrid retrieval is most valuable when applied to out-of-domain datasets~\cite{chen2022ecir}, we believe that $f_\textsc{Hybrid}$ should be robust to distributional shifts and should not need training or fine-tuning on target datasets. This implies that either the function must be non-parametric, that its parameters can be tuned efficiently with respect to the training samples required, or that they are highly interpretable such that their value can be guided by expert knowledge.

In the absence of a truly non-parametric approach, however, we believe a fusion that is more sample-efficient to tune is preferred. Because convex combination has fewer parameters than the fully parameterized \ac{rrf}, we believe it should have this property. To confirm, we ask how many training queries it takes to converge to the correct $\alpha$ on a target dataset.

Figure~\ref{figure:sample_efficiency} visualizes our experiments, where we plot NDCG of \ac{rrf} ($\eta=60$) and TM2C2 with $\alpha=0.8$ from Table~\ref{table:rq1}. Additionally, we take the train split of each dataset and sample from it progressively larger subsets (with a step size of $5\%$), and use it to tune the parameters of each function. We then measure NDCG of the tuned functions on the test split. For the depicted datasets as well as all other datasets in this work, we observe a similar trend: with less than $5\%$ of the training data, which is often a small set of queries, TM2C2's $\alpha$ converges, regardless of the magnitude of domain shift. This sample efficiency is remarkable because it enables significant gains with little labeling effort. Finally, while \ac{rrf} does not settle on a value and its parameters are sensitive to the training sample, its performance does more or less converge. However, the performance of the fully parameterized \ac{rrf} is still sub-optimal compared with TM2C2.

In Figure~\ref{figure:sample_efficiency}, we also include a convex combination of fully parameterized RRF terms, denoted by \ac{rrf}-CC and defined as:
\begin{equation}
    f_\textsc{RRF}(q, d)= (1 - \alpha) \frac{1}{\eta_\textsc{Lex} + \pi_\textsc{Lex}(q, d)} + \alpha \frac{1}{\eta_\textsc{Sem} + \pi_\textsc{Sem}(q, d)},
    \label{equation:rrf-cc}
\end{equation}
where $\alpha$, $\eta_\textsc{Lex}$, and $\eta_\textsc{Sem}$ are tunable parameters. The question this particular formulation tries to answer is whether adding an additional weight to the combination of the \ac{rrf} terms affects retrieval quality. From the figure, it is clear that the addition of this parameter does not have a significant impact on the overall performance. This also serves as additional evidence supporting the claim that Lipschitz continuity is an important property.

\section{Conclusion}
\label{section:conclusion}
We studied the behavior of two popular functions that fuse together lexical and semantic retrieval to produce hybrid retrieval, and identified their advantages and pitfalls. Importantly, we investigated several questions and claims in prior work. We established theoretically that the choice of normalization is not as consequential as once thought for a convex combination-based fusion function. We found that \ac{rrf} is sensitive to its parameters. We also observed empirically that convex combination of normalized scores outperforms \ac{rrf} on in-domain and out-of-domain datasets---a finding that is in disagreement with~\cite{chen2022ecir}.

We believe that a convex combination with theoretical minimum-maximum normalization (TM2C2) indeed enjoys properties that are important in a fusion function. Its parameter, too, can be tuned sample-efficiently or set to a reasonable value based on domain knowledge. In our experiments, for example, we found the range $\alpha \in [0.6, 0.8]$ to consistently lead to improvements.

While we observed that a line appears to be appropriate for a collection of query-document pairs, we acknowledge that, that may change if our analysis was conducted on a per-query basis---itself a rather non-trivial effort. For example, it is unclear if bringing non-linearity to the design of the fusion function or the normalization itself leads to a more accurate prediction of $\alpha$ on a per-query basis. We leave an exploration of this question to future work.

We also note that, while our analysis does not exclude the use of multiple retrieval engines as input, and indeed can be extended, both theoretically and empirically, to a setting where we have more than just lexical and semantic scores, it is nonetheless important to conduct experiments and validate that our findings generalize. In particular, we ask if the role of normalization changes when fusing three or more models, and if so, what is the behavior of convex combination given a particular normalization function, and how does that compare with \ac{rrf}. We believe, however, that our current assumptions are practical and are reflective of the current state of hybrid search where we typically fuse only lexical and semantic retrieval systems. As such, we leave an extended analysis of fusion on multiple retrieval systems to future work.

\begin{acks}
We benefited greatly from conversations with Brian Hentschel, Edo Liberty, and Michael Bendersky. We are grateful to them for their insight and time. We further thank the anonymous reviewers for their meticulous examination of our claims and for their insightful feedback.
\end{acks}

\newpage
\bibliographystyle{ACM-Reference-Format}
\bibliography{main}


\begin{thebibliography}{47}


\ifx \showCODEN    \undefined \def \showCODEN     #1{\unskip}     \fi
\ifx \showDOI      \undefined \def \showDOI       #1{#1}\fi
\ifx \showISBNx    \undefined \def \showISBNx     #1{\unskip}     \fi
\ifx \showISBNxiii \undefined \def \showISBNxiii  #1{\unskip}     \fi
\ifx \showISSN     \undefined \def \showISSN      #1{\unskip}     \fi
\ifx \showLCCN     \undefined \def \showLCCN      #1{\unskip}     \fi
\ifx \shownote     \undefined \def \shownote      #1{#1}          \fi
\ifx \showarticletitle \undefined \def \showarticletitle #1{#1}   \fi
\ifx \showURL      \undefined \def \showURL       {\relax}        \fi
\providecommand\bibfield[2]{#2}
\providecommand\bibinfo[2]{#2}
\providecommand\natexlab[1]{#1}
\providecommand\showeprint[2][]{arXiv:#2}

\bibitem[Asadi(2013)]%
        {asadi2013phd}
\bibfield{author}{\bibinfo{person}{Nima Asadi}.}
  \bibinfo{year}{2013}\natexlab{}.
\newblock \bibinfo{booktitle}{\emph{Multi-Stage Search Architectures for
  Streaming Documents}}.
\newblock \bibinfo{publisher}{University of Maryland}.
\newblock


\bibitem[Asadi and Lin(2013)]%
        {asadi2013efficiency}
\bibfield{author}{\bibinfo{person}{Nima Asadi} {and} \bibinfo{person}{Jimmy
  Lin}.} \bibinfo{year}{2013}\natexlab{}.
\newblock \showarticletitle{Effectiveness/Efficiency Tradeoffs for Candidate
  Generation in Multi-Stage Retrieval Architectures}. In
  \bibinfo{booktitle}{\emph{Proceedings of the 36th International ACM SIGIR
  Conference on Research and Development in Information Retrieval}} (Dublin,
  Ireland). \bibinfo{pages}{997–1000}.
\newblock


\bibitem[Bruch et~al\mbox{.}(2022)]%
        {bruch2022reneuir}
\bibfield{author}{\bibinfo{person}{Sebastian Bruch}, \bibinfo{person}{Claudio
  Lucchese}, {and} \bibinfo{person}{Franco~Maria Nardini}.}
  \bibinfo{year}{2022}\natexlab{}.
\newblock \showarticletitle{ReNeuIR: Reaching Efficiency in Neural Information
  Retrieval}. In \bibinfo{booktitle}{\emph{Proceedings of the 45th
  International ACM SIGIR Conference on Research and Development in Information
  Retrieval}} (Madrid, Spain). \bibinfo{pages}{3462–3465}.
\newblock


\bibitem[Bruch et~al\mbox{.}(2019)]%
        {bruch2019approxndcg}
\bibfield{author}{\bibinfo{person}{Sebastian Bruch}, \bibinfo{person}{Masrour
  Zoghi}, \bibinfo{person}{Michael Bendersky}, {and} \bibinfo{person}{Marc
  Najork}.} \bibinfo{year}{2019}\natexlab{}.
\newblock \showarticletitle{Revisiting Approximate Metric Optimization in the
  Age of Deep Neural Networks}. In \bibinfo{booktitle}{\emph{Proceedings of the
  42nd International ACM SIGIR Conference on Research and Development in
  Information Retrieval}} (Paris, France). \bibinfo{pages}{1241–1244}.
\newblock


\bibitem[Chen et~al\mbox{.}(2022)]%
        {chen2022ecir}
\bibfield{author}{\bibinfo{person}{Tao Chen}, \bibinfo{person}{Mingyang Zhang},
  \bibinfo{person}{Jing Lu}, \bibinfo{person}{Michael Bendersky}, {and}
  \bibinfo{person}{Marc Najork}.} \bibinfo{year}{2022}\natexlab{}.
\newblock \showarticletitle{Out-of-Domain Semantics to the Rescue! Zero-Shot
  Hybrid Retrieval Models}. In \bibinfo{booktitle}{\emph{Advances in
  Information Retrieval: 44th European Conference on IR Research, ECIR 2022,
  Stavanger, Norway, April 10–14, 2022, Proceedings, Part I}} (Stavanger,
  Norway). \bibinfo{pages}{95–110}.
\newblock


\bibitem[Cormack et~al\mbox{.}(2009)]%
        {cormack2009rrf}
\bibfield{author}{\bibinfo{person}{Gordon~V. Cormack}, \bibinfo{person}{Charles
  L~A Clarke}, {and} \bibinfo{person}{Stefan Buettcher}.}
  \bibinfo{year}{2009}\natexlab{}.
\newblock \showarticletitle{Reciprocal Rank Fusion Outperforms Condorcet and
  Individual Rank Learning Methods}. In \bibinfo{booktitle}{\emph{Proceedings
  of the 32nd International ACM SIGIR Conference on Research and Development in
  Information Retrieval}} (Boston, MA, USA). \bibinfo{pages}{758–759}.
\newblock


\bibitem[Dang et~al\mbox{.}(2013)]%
        {ecir13}
\bibfield{author}{\bibinfo{person}{Van Dang}, \bibinfo{person}{Michael
  Bendersky}, {and} \bibinfo{person}{W~Bruce Croft}.}
  \bibinfo{year}{2013}\natexlab{}.
\newblock \showarticletitle{Two-Stage learning to rank for information
  retrieval}.
\newblock In \bibinfo{booktitle}{\emph{Advances in Information Retrieval}}.
  \bibinfo{publisher}{Springer}, \bibinfo{pages}{423--434}.
\newblock


\bibitem[Devlin et~al\mbox{.}(2019)]%
        {devlin2019bert}
\bibfield{author}{\bibinfo{person}{Jacob Devlin}, \bibinfo{person}{Ming-Wei
  Chang}, \bibinfo{person}{Kenton Lee}, {and} \bibinfo{person}{Kristina
  Toutanova}.} \bibinfo{year}{2019}\natexlab{}.
\newblock \showarticletitle{{BERT}: Pre-training of Deep Bidirectional
  Transformers for Language Understanding}. In
  \bibinfo{booktitle}{\emph{Proceedings of the 2019 Conference of the North
  {A}merican Chapter of the Association for Computational Linguistics: Human
  Language Technologies, Volume 1 (Long and Short Papers)}}.
  \bibinfo{pages}{4171--4186}.
\newblock


\bibitem[Formal et~al\mbox{.}(2022)]%
        {formal2022splade}
\bibfield{author}{\bibinfo{person}{Thibault Formal}, \bibinfo{person}{Carlos
  Lassance}, \bibinfo{person}{Benjamin Piwowarski}, {and}
  \bibinfo{person}{St\'{e}phane Clinchant}.} \bibinfo{year}{2022}\natexlab{}.
\newblock \showarticletitle{From Distillation to Hard Negative Sampling: Making
  Sparse Neural IR Models More Effective}. In
  \bibinfo{booktitle}{\emph{Proceedings of the 45th International ACM SIGIR
  Conference on Research and Development in Information Retrieval}} (Madrid,
  Spain). \bibinfo{pages}{2353–2359}.
\newblock


\bibitem[Hofst\"{a}tter et~al\mbox{.}(2021)]%
        {tas-b}
\bibfield{author}{\bibinfo{person}{Sebastian Hofst\"{a}tter},
  \bibinfo{person}{Sheng-Chieh Lin}, \bibinfo{person}{Jheng-Hong Yang},
  \bibinfo{person}{Jimmy Lin}, {and} \bibinfo{person}{Allan Hanbury}.}
  \bibinfo{year}{2021}\natexlab{}.
\newblock \showarticletitle{Efficiently Teaching an Effective Dense Retriever
  with Balanced Topic Aware Sampling}. In \bibinfo{booktitle}{\emph{Proceedings
  of the 44th International ACM SIGIR Conference on Research and Development in
  Information Retrieval}} (Virtual Event, Canada). \bibinfo{pages}{113–122}.
\newblock


\bibitem[J{\"a}rvelin and Kek{\"a}l{\"a}inen(2000)]%
        {jarvelin2000ir}
\bibfield{author}{\bibinfo{person}{Kalervo J{\"a}rvelin} {and}
  \bibinfo{person}{Jaana Kek{\"a}l{\"a}inen}.} \bibinfo{year}{2000}\natexlab{}.
\newblock \showarticletitle{IR evaluation methods for retrieving highly
  relevant documents}. In \bibinfo{booktitle}{\emph{Proceedings of the 23rd
  annual international ACM SIGIR conference on Research and development in
  information retrieval}}. ACM, \bibinfo{pages}{41--48}.
\newblock


\bibitem[Johnson et~al\mbox{.}(2021)]%
        {Johnson2021faiss}
\bibfield{author}{\bibinfo{person}{Jeff Johnson}, \bibinfo{person}{Matthijs
  Douze}, {and} \bibinfo{person}{Herv{\'e} J{\'e}gou}.}
  \bibinfo{year}{2021}\natexlab{}.
\newblock \showarticletitle{Billion-Scale Similarity Search with GPUs}.
\newblock \bibinfo{journal}{\emph{IEEE Transactions on Big Data}}
  \bibinfo{volume}{7} (\bibinfo{year}{2021}), \bibinfo{pages}{535--547}.
\newblock


\bibitem[Karpukhin et~al\mbox{.}(2020)]%
        {karpukhin-etal-2020-dense}
\bibfield{author}{\bibinfo{person}{Vladimir Karpukhin}, \bibinfo{person}{Barlas
  Oguz}, \bibinfo{person}{Sewon Min}, \bibinfo{person}{Patrick Lewis},
  \bibinfo{person}{Ledell Wu}, \bibinfo{person}{Sergey Edunov},
  \bibinfo{person}{Danqi Chen}, {and} \bibinfo{person}{Wen-tau Yih}.}
  \bibinfo{year}{2020}\natexlab{}.
\newblock \showarticletitle{Dense Passage Retrieval for Open-Domain Question
  Answering}. In \bibinfo{booktitle}{\emph{Proceedings of the 2020 Conference
  on Empirical Methods in Natural Language Processing (EMNLP)}}.
\newblock


\bibitem[Kuzi et~al\mbox{.}(2020)]%
        {Kuzi2020LeveragingSA}
\bibfield{author}{\bibinfo{person}{Saar Kuzi}, \bibinfo{person}{Mingyang
  Zhang}, \bibinfo{person}{Cheng Li}, \bibinfo{person}{Michael Bendersky},
  {and} \bibinfo{person}{Marc Najork}.} \bibinfo{year}{2020}\natexlab{}.
\newblock \showarticletitle{Leveraging Semantic and Lexical Matching to Improve
  the Recall of Document Retrieval Systems: A Hybrid Approach}.
\newblock  (\bibinfo{year}{2020}).
\newblock
\showeprint[arxiv]{2010.01195}~[cs.IR]


\bibitem[Li et~al\mbox{.}(2022)]%
        {li2022prf}
\bibfield{author}{\bibinfo{person}{Hang Li}, \bibinfo{person}{Shuai Wang},
  \bibinfo{person}{Shengyao Zhuang}, \bibinfo{person}{Ahmed Mourad},
  \bibinfo{person}{Xueguang Ma}, \bibinfo{person}{Jimmy Lin}, {and}
  \bibinfo{person}{Guido Zuccon}.} \bibinfo{year}{2022}\natexlab{}.
\newblock \showarticletitle{To Interpolate or Not to Interpolate: PRF, Dense
  and Sparse Retrievers}. In \bibinfo{booktitle}{\emph{Proceedings of the 45th
  International ACM SIGIR Conference on Research and Development in Information
  Retrieval}} (Madrid, Spain). \bibinfo{pages}{2495–2500}.
\newblock


\bibitem[Lin et~al\mbox{.}(2021)]%
        {lin2021pretrained}
\bibfield{author}{\bibinfo{person}{Jimmy Lin}, \bibinfo{person}{Rodrigo
  Nogueira}, {and} \bibinfo{person}{Andrew Yates}.}
  \bibinfo{year}{2021}\natexlab{}.
\newblock \bibinfo{title}{Pretrained Transformers for Text Ranking: BERT and
  Beyond}.
\newblock
\newblock
\showeprint[arxiv]{2010.06467}~[cs.IR]


\bibitem[Liu(2009)]%
        {LTR-Liu09}
\bibfield{author}{\bibinfo{person}{Tie{-}Yan Liu}.}
  \bibinfo{year}{2009}\natexlab{}.
\newblock \showarticletitle{Learning to Rank for Information Retrieval}.
\newblock \bibinfo{journal}{\emph{Foundations and Trends in Information
  Retrieval}} \bibinfo{volume}{3}, \bibinfo{number}{3} (\bibinfo{year}{2009}),
  \bibinfo{pages}{225--331}.
\newblock


\bibitem[Luan et~al\mbox{.}(2021)]%
        {luan-etal-2021-sparse}
\bibfield{author}{\bibinfo{person}{Yi Luan}, \bibinfo{person}{Jacob
  Eisenstein}, \bibinfo{person}{Kristina Toutanova}, {and}
  \bibinfo{person}{Michael Collins}.} \bibinfo{year}{2021}\natexlab{}.
\newblock \showarticletitle{Sparse, Dense, and Attentional Representations for
  Text Retrieval}.
\newblock \bibinfo{journal}{\emph{Transactions of the Association for
  Computational Linguistics}}  \bibinfo{volume}{9} (\bibinfo{year}{2021}),
  \bibinfo{pages}{329--345}.
\newblock


\bibitem[Ma et~al\mbox{.}(2020)]%
        {Ma2020HybridFR}
\bibfield{author}{\bibinfo{person}{Ji Ma}, \bibinfo{person}{Ivan Korotkov},
  \bibinfo{person}{Keith Hall}, {and} \bibinfo{person}{Ryan~T. McDonald}.}
  \bibinfo{year}{2020}\natexlab{}.
\newblock \showarticletitle{Hybrid First-stage Retrieval Models for Biomedical
  Literature}. In \bibinfo{booktitle}{\emph{CLEF}}.
\newblock


\bibitem[Ma et~al\mbox{.}(2021)]%
        {Ma2021ARS}
\bibfield{author}{\bibinfo{person}{Xueguang Ma}, \bibinfo{person}{Kai Sun},
  \bibinfo{person}{Ronak Pradeep}, {and} \bibinfo{person}{Jimmy~J. Lin}.}
  \bibinfo{year}{2021}\natexlab{}.
\newblock \showarticletitle{A Replication Study of Dense Passage Retriever}.
\newblock  (\bibinfo{year}{2021}).
\newblock
\showeprint[arxiv]{2004.04906}~[cs.CL]


\bibitem[Macdonald et~al\mbox{.}(2013)]%
        {macdonald2013whens}
\bibfield{author}{\bibinfo{person}{Craig Macdonald},
  \bibinfo{person}{Rodrygo~LT Santos}, {and} \bibinfo{person}{Iadh Ounis}.}
  \bibinfo{year}{2013}\natexlab{}.
\newblock \showarticletitle{The whens and hows of learning to rank for web
  search}.
\newblock \bibinfo{journal}{\emph{Information Retrieval}} \bibinfo{volume}{16},
  \bibinfo{number}{5} (\bibinfo{year}{2013}), \bibinfo{pages}{584--628}.
\newblock


\bibitem[Malkov and Yashunin(2016)]%
        {malkov2016hnsw}
\bibfield{author}{\bibinfo{person}{Yu.~A. Malkov} {and} \bibinfo{person}{D.~A.
  Yashunin}.} \bibinfo{year}{2016}\natexlab{}.
\newblock \bibinfo{title}{Efficient and robust approximate nearest neighbor
  search using Hierarchical Navigable Small World graphs}.
\newblock
\newblock
\showeprint[arxiv]{1603.09320}~[cs.DS]


\bibitem[Mallia et~al\mbox{.}(2019)]%
        {MSMS2019}
\bibfield{author}{\bibinfo{person}{Antonio Mallia}, \bibinfo{person}{Michal
  Siedlaczek}, \bibinfo{person}{Joel Mackenzie}, {and} \bibinfo{person}{Torsten
  Suel}.} \bibinfo{year}{2019}\natexlab{}.
\newblock \showarticletitle{{PISA:} Performant Indexes and Search for
  Academia}. In \bibinfo{booktitle}{\emph{Proceedings of the Open-Source {IR}
  Replicability Challenge co-located with 42nd International {ACM} {SIGIR}
  Conference on Research and Development in Information Retrieval Paris,
  France, July 25, 2019.}} \bibinfo{pages}{50--56}.
\newblock


\bibitem[Matsubara et~al\mbox{.}(2020)]%
        {matsubara2020multistage}
\bibfield{author}{\bibinfo{person}{Yoshitomo Matsubara}, \bibinfo{person}{Thuy
  Vu}, {and} \bibinfo{person}{Alessandro Moschitti}.}
  \bibinfo{year}{2020}\natexlab{}.
\newblock \showarticletitle{Reranking for Efficient Transformer-Based Answer
  Selection}. In \bibinfo{booktitle}{\emph{Proceedings of the 43rd
  International ACM SIGIR Conference on Research and Development in Information
  Retrieval}}. \bibinfo{pages}{1577–1580}.
\newblock


\bibitem[Mitra et~al\mbox{.}(2016)]%
        {mitra2016dual}
\bibfield{author}{\bibinfo{person}{Bhaskar Mitra}, \bibinfo{person}{Eric
  Nalisnick}, \bibinfo{person}{Nick Craswell}, {and} \bibinfo{person}{Rich
  Caruana}.} \bibinfo{year}{2016}\natexlab{}.
\newblock \showarticletitle{A dual embedding space model for document ranking}.
\newblock  (\bibinfo{year}{2016}).
\newblock
\showeprint[arxiv]{1602.01137}~[cs.IR]


\bibitem[Nguyen et~al\mbox{.}(2016)]%
        {nguyen2016msmarco}
\bibfield{author}{\bibinfo{person}{Tri Nguyen}, \bibinfo{person}{Mir
  Rosenberg}, \bibinfo{person}{Xia Song}, \bibinfo{person}{Jianfeng Gao},
  \bibinfo{person}{Saurabh Tiwary}, \bibinfo{person}{Rangan Majumder}, {and}
  \bibinfo{person}{Li Deng}.} \bibinfo{year}{2016}\natexlab{}.
\newblock \showarticletitle{MS MARCO: A Human Generated MAchine Reading
  COmprehension Dataset}.
\newblock  (\bibinfo{date}{November} \bibinfo{year}{2016}).
\newblock


\bibitem[Nogueira and Cho(2020)]%
        {nogueira2020passage}
\bibfield{author}{\bibinfo{person}{Rodrigo Nogueira} {and}
  \bibinfo{person}{Kyunghyun Cho}.} \bibinfo{year}{2020}\natexlab{}.
\newblock \bibinfo{title}{Passage Re-ranking with BERT}.
\newblock
\newblock
\showeprint[arxiv]{1901.04085}~[cs.IR]


\bibitem[Nogueira et~al\mbox{.}(2020)]%
        {nogueira2020monot5}
\bibfield{author}{\bibinfo{person}{Rodrigo Nogueira}, \bibinfo{person}{Zhiying
  Jiang}, \bibinfo{person}{Ronak Pradeep}, {and} \bibinfo{person}{Jimmy Lin}.}
  \bibinfo{year}{2020}\natexlab{}.
\newblock \showarticletitle{Document Ranking with a Pretrained
  Sequence-to-Sequence Model}. In \bibinfo{booktitle}{\emph{Findings of the
  Association for Computational Linguistics: EMNLP 2020}}.
  \bibinfo{pages}{708--718}.
\newblock


\bibitem[Nogueira et~al\mbox{.}(2019a)]%
        {nogueira2019multi}
\bibfield{author}{\bibinfo{person}{Rodrigo Nogueira}, \bibinfo{person}{Wei
  Yang}, \bibinfo{person}{Kyunghyun Cho}, {and} \bibinfo{person}{Jimmy Lin}.}
  \bibinfo{year}{2019}\natexlab{a}.
\newblock \showarticletitle{Multi-stage document ranking with BERT}.
\newblock  (\bibinfo{year}{2019}).
\newblock
\showeprint{1910.14424}~[cs.IR]


\bibitem[Nogueira et~al\mbox{.}(2019b)]%
        {nogueira2019document}
\bibfield{author}{\bibinfo{person}{Rodrigo Nogueira}, \bibinfo{person}{Wei
  Yang}, \bibinfo{person}{Jimmy Lin}, {and} \bibinfo{person}{Kyunghyun Cho}.}
  \bibinfo{year}{2019}\natexlab{b}.
\newblock \showarticletitle{Document Expansion by Query Prediction}.
\newblock  (\bibinfo{year}{2019}).
\newblock
\showeprint[arxiv]{1904.08375}~[cs.IR]


\bibitem[Qin et~al\mbox{.}(2010)]%
        {qin2010general}
\bibfield{author}{\bibinfo{person}{Tao Qin}, \bibinfo{person}{Tie-Yan Liu},
  {and} \bibinfo{person}{Hang Li}.} \bibinfo{year}{2010}\natexlab{}.
\newblock \showarticletitle{A general approximation framework for direct
  optimization of information retrieval measures}.
\newblock \bibinfo{journal}{\emph{Information retrieval}} \bibinfo{volume}{13},
  \bibinfo{number}{4} (\bibinfo{year}{2010}), \bibinfo{pages}{375--397}.
\newblock


\bibitem[Reimers and Gurevych(2019)]%
        {reimers-2019-sentence-bert}
\bibfield{author}{\bibinfo{person}{Nils Reimers} {and} \bibinfo{person}{Iryna
  Gurevych}.} \bibinfo{year}{2019}\natexlab{}.
\newblock \showarticletitle{Sentence-BERT: Sentence Embeddings using Siamese
  BERT-Networks}. In \bibinfo{booktitle}{\emph{Proceedings of the 2019
  Conference on Empirical Methods in Natural Language Processing}}.
  \bibinfo{publisher}{Association for Computational Linguistics}.
\newblock


\bibitem[Robertson and Zaragoza(2009)]%
        {bm25}
\bibfield{author}{\bibinfo{person}{Stephen Robertson} {and}
  \bibinfo{person}{Hugo Zaragoza}.} \bibinfo{year}{2009}\natexlab{}.
\newblock \showarticletitle{The Probabilistic Relevance Framework: BM25 and
  Beyond}.
\newblock \bibinfo{journal}{\emph{Foundations and Trends in Information
  Retrieval}} \bibinfo{volume}{3}, \bibinfo{number}{4} (\bibinfo{date}{April}
  \bibinfo{year}{2009}), \bibinfo{pages}{333--389}.
\newblock
\showISSN{1554-0669}


\bibitem[Robertson et~al\mbox{.}(1994)]%
        {bm25original}
\bibfield{author}{\bibinfo{person}{Stephen~E. Robertson},
  \bibinfo{person}{Steve Walker}, \bibinfo{person}{Susan Jones},
  \bibinfo{person}{Micheline Hancock-Beaulieu}, {and} \bibinfo{person}{Mike
  Gatford}.} \bibinfo{year}{1994}\natexlab{}.
\newblock \showarticletitle{Okapi at TREC-3}. In
  \bibinfo{booktitle}{\emph{TREC}} \emph{(\bibinfo{series}{NIST Special
  Publication}, Vol.~\bibinfo{volume}{500-225})}. \bibinfo{publisher}{National
  Institute of Standards and Technology (NIST)}, \bibinfo{pages}{109--126}.
\newblock


\bibitem[Scells et~al\mbox{.}(2022)]%
        {scells2022sigir-green-ir}
\bibfield{author}{\bibinfo{person}{Harrisen Scells}, \bibinfo{person}{Shengyao
  Zhuang}, {and} \bibinfo{person}{Guido Zuccon}.}
  \bibinfo{year}{2022}\natexlab{}.
\newblock \showarticletitle{Reduce, Reuse, Recycle: Green Information Retrieval
  Research}. In \bibinfo{booktitle}{\emph{Proceedings of the 45th International
  ACM SIGIR Conference on Research and Development in Information Retrieval}}
  (Madrid, Spain). \bibinfo{pages}{2825–2837}.
\newblock


\bibitem[Tao et~al\mbox{.}(2006)]%
        {tao2006expansion}
\bibfield{author}{\bibinfo{person}{Tao Tao}, \bibinfo{person}{Xuanhui Wang},
  \bibinfo{person}{Qiaozhu Mei}, {and} \bibinfo{person}{ChengXiang Zhai}.}
  \bibinfo{year}{2006}\natexlab{}.
\newblock \showarticletitle{Language Model Information Retrieval with Document
  Expansion}. In \bibinfo{booktitle}{\emph{Proceedings of the Main Conference
  on Human Language Technology Conference of the North American Chapter of the
  Association of Computational Linguistics}} (New York, New York).
  \bibinfo{pages}{407–414}.
\newblock


\bibitem[Thakur et~al\mbox{.}(2021)]%
        {thakur2021beir}
\bibfield{author}{\bibinfo{person}{Nandan Thakur}, \bibinfo{person}{Nils
  Reimers}, \bibinfo{person}{Andreas R{\"u}ckl{\'e}}, \bibinfo{person}{Abhishek
  Srivastava}, {and} \bibinfo{person}{Iryna Gurevych}.}
  \bibinfo{year}{2021}\natexlab{}.
\newblock \showarticletitle{{BEIR}: A Heterogeneous Benchmark for Zero-shot
  Evaluation of Information Retrieval Models}. In
  \bibinfo{booktitle}{\emph{Proceedings of the 35th Conference on Neural
  Information Processing Systems Datasets and Benchmarks Track (Round 2)}}.
\newblock


\bibitem[Vaswani et~al\mbox{.}(2017)]%
        {vaswani2017attention}
\bibfield{author}{\bibinfo{person}{Ashish Vaswani}, \bibinfo{person}{Noam
  Shazeer}, \bibinfo{person}{Niki Parmar}, \bibinfo{person}{Jakob Uszkoreit},
  \bibinfo{person}{Llion Jones}, \bibinfo{person}{Aidan~N. Gomez},
  \bibinfo{person}{\L{}ukasz Kaiser}, {and} \bibinfo{person}{Illia
  Polosukhin}.} \bibinfo{year}{2017}\natexlab{}.
\newblock \showarticletitle{Attention is All You Need}. In
  \bibinfo{booktitle}{\emph{Proceedings of the 31st International Conference on
  Neural Information Processing Systems}} (Long Beach, California, USA).
  \bibinfo{pages}{6000–6010}.
\newblock


\bibitem[Wang et~al\mbox{.}(2011)]%
        {wang2011cascade}
\bibfield{author}{\bibinfo{person}{Lidan Wang}, \bibinfo{person}{Jimmy Lin},
  {and} \bibinfo{person}{Donald Metzler}.} \bibinfo{year}{2011}\natexlab{}.
\newblock \showarticletitle{A cascade ranking model for efficient ranked
  retrieval}. In \bibinfo{booktitle}{\emph{Proceedings of the 34th
  international ACM SIGIR conference on Research and development in Information
  Retrieval}}. ACM, \bibinfo{pages}{105--114}.
\newblock


\bibitem[Wang et~al\mbox{.}(2021)]%
        {wang2021bert}
\bibfield{author}{\bibinfo{person}{Shuai Wang}, \bibinfo{person}{Shengyao
  Zhuang}, {and} \bibinfo{person}{Guido Zuccon}.}
  \bibinfo{year}{2021}\natexlab{}.
\newblock \showarticletitle{BERT-Based Dense Retrievers Require Interpolation
  with BM25 for Effective Passage Retrieval}. In
  \bibinfo{booktitle}{\emph{Proceedings of the 2021 ACM SIGIR International
  Conference on Theory of Information Retrieval}} (Virtual Event, Canada).
  \bibinfo{pages}{317–324}.
\newblock


\bibitem[Wu et~al\mbox{.}(2010)]%
        {lambdamart}
\bibfield{author}{\bibinfo{person}{Qiang Wu}, \bibinfo{person}{Christopher~J.C.
  Burges}, \bibinfo{person}{Krysta~M. Svore}, {and} \bibinfo{person}{Jianfeng
  Gao}.} \bibinfo{year}{2010}\natexlab{}.
\newblock \showarticletitle{Adapting boosting for information retrieval
  measures}.
\newblock \bibinfo{journal}{\emph{Information Retrieval}}
  (\bibinfo{year}{2010}).
\newblock


\bibitem[Wu et~al\mbox{.}(2019)]%
        {Wu2019EfficientIP}
\bibfield{author}{\bibinfo{person}{Xiang Wu}, \bibinfo{person}{Ruiqi Guo},
  \bibinfo{person}{David Simcha}, \bibinfo{person}{Dave Dopson}, {and}
  \bibinfo{person}{Sanjiv Kumar}.} \bibinfo{year}{2019}\natexlab{}.
\newblock \showarticletitle{Efficient Inner Product Approximation in Hybrid
  Spaces}.
\newblock  (\bibinfo{year}{2019}).
\newblock
\showeprint[arxiv]{1903.08690}~[cs.LG]


\bibitem[Wu et~al\mbox{.}(2016)]%
        {wordpiece}
\bibfield{author}{\bibinfo{person}{Yonghui Wu}, \bibinfo{person}{Mike
  Schuster}, \bibinfo{person}{Zhifeng Chen}, \bibinfo{person}{Quoc~V. Le},
  \bibinfo{person}{Mohammad Norouzi}, \bibinfo{person}{Wolfgang Macherey},
  \bibinfo{person}{Maxim Krikun}, \bibinfo{person}{Yuan Cao},
  \bibinfo{person}{Qin Gao}, \bibinfo{person}{Klaus Macherey},
  \bibinfo{person}{Jeff Klingner}, \bibinfo{person}{Apurva Shah},
  \bibinfo{person}{Melvin Johnson}, \bibinfo{person}{Xiaobing Liu},
  \bibinfo{person}{{\L}ukasz Kaiser}, \bibinfo{person}{Stephan Gouws},
  \bibinfo{person}{Yoshikiyo Kato}, \bibinfo{person}{Taku Kudo},
  \bibinfo{person}{Hideto Kazawa}, \bibinfo{person}{Keith Stevens},
  \bibinfo{person}{George Kurian}, \bibinfo{person}{Nishant Patil},
  \bibinfo{person}{Wei Wang}, \bibinfo{person}{Cliff Young},
  \bibinfo{person}{Jason Smith}, \bibinfo{person}{Jason Riesa},
  \bibinfo{person}{Alex Rudnick}, \bibinfo{person}{Oriol Vinyals},
  \bibinfo{person}{Greg Corrado}, \bibinfo{person}{Macduff Hughes}, {and}
  \bibinfo{person}{Jeffrey Dean}.} \bibinfo{year}{2016}\natexlab{}.
\newblock \bibinfo{title}{Google's Neural Machine Translation System: Bridging
  the Gap between Human and Machine Translation}.
\newblock
\newblock
\showeprint[arxiv]{1609.08144}~[cs.CL]


\bibitem[Xiong et~al\mbox{.}(2021)]%
        {xiong2021approximate}
\bibfield{author}{\bibinfo{person}{Lee Xiong}, \bibinfo{person}{Chenyan Xiong},
  \bibinfo{person}{Ye Li}, \bibinfo{person}{Kwok-Fung Tang},
  \bibinfo{person}{Jialin Liu}, \bibinfo{person}{Paul Bennett},
  \bibinfo{person}{Junaid Ahmed}, {and} \bibinfo{person}{Arnold Overwijk}.}
  \bibinfo{year}{2021}\natexlab{}.
\newblock \showarticletitle{Approximate Nearest Neighbor Negative Contrastive
  Learning for Dense Text Retrieval}. In
  \bibinfo{booktitle}{\emph{International Conference on Learning
  Representations}}.
\newblock


\bibitem[Yin et~al\mbox{.}(2016)]%
        {yin2016ranking}
\bibfield{author}{\bibinfo{person}{Dawei Yin}, \bibinfo{person}{Yuening Hu},
  \bibinfo{person}{Jiliang Tang}, \bibinfo{person}{Tim Daly},
  \bibinfo{person}{Mianwei Zhou}, \bibinfo{person}{Hua Ouyang},
  \bibinfo{person}{Jianhui Chen}, \bibinfo{person}{Changsung Kang},
  \bibinfo{person}{Hongbo Deng}, \bibinfo{person}{Chikashi Nobata},
  {et~al\mbox{.}}} \bibinfo{year}{2016}\natexlab{}.
\newblock \showarticletitle{Ranking relevance in yahoo search}. In
  \bibinfo{booktitle}{\emph{Proceedings of the 22nd ACM SIGKDD International
  Conference on Knowledge Discovery and Data Mining}}. ACM,
  \bibinfo{pages}{323--332}.
\newblock


\bibitem[Zamani et~al\mbox{.}(2022)]%
        {Zamani2022cascade}
\bibfield{author}{\bibinfo{person}{Hamed Zamani}, \bibinfo{person}{Mike
  Bendersky}, \bibinfo{person}{Donald Metzler}, \bibinfo{person}{Honglei
  Zhuang}, {and} \bibinfo{person}{Marc Najork}.}
  \bibinfo{year}{2022}\natexlab{}.
\newblock \showarticletitle{Stochastic Retrieval-Conditioned Reranking}. In
  \bibinfo{booktitle}{\emph{Proceedings of the 2022 ACM SIGIR International
  Conference on the Theory of Information Retrieval}} (Madrid, Spain).
\newblock


\bibitem[Zhan et~al\mbox{.}(2020)]%
        {zhan2020repbert}
\bibfield{author}{\bibinfo{person}{Jingtao Zhan}, \bibinfo{person}{Jiaxin Mao},
  \bibinfo{person}{Yiqun Liu}, \bibinfo{person}{Min Zhang}, {and}
  \bibinfo{person}{Shaoping Ma}.} \bibinfo{year}{2020}\natexlab{}.
\newblock \bibinfo{title}{RepBERT: Contextualized Text Embeddings for
  First-Stage Retrieval}.
\newblock
\newblock
\showeprint[arxiv]{2006.15498}~[cs.IR]


\end{thebibliography}

\newpage
\appendix
\section{Fusion of \textsc{Splade} and BM25}
\label{appendix:splade-bm25}

\begin{figure}[h]
\begin{center}
\centerline{
\subfloat[MS MARCO]{
\includegraphics[width=0.36\linewidth,height=4cm]{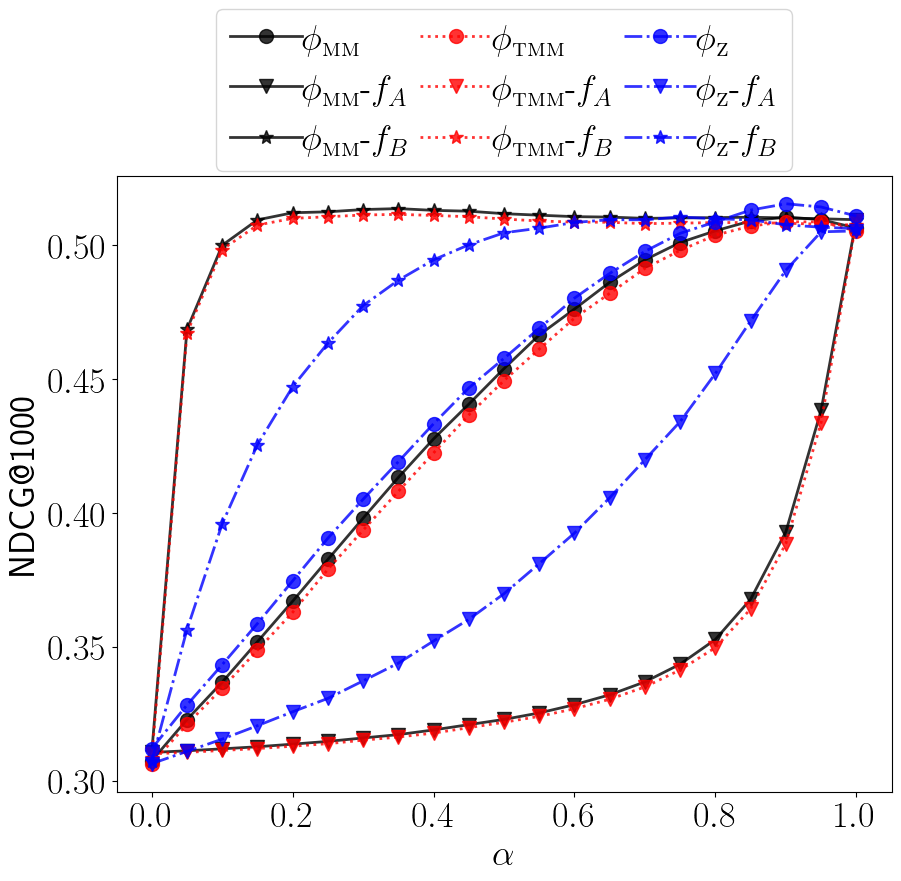}
\includegraphics[width=0.36\linewidth,height=4cm]{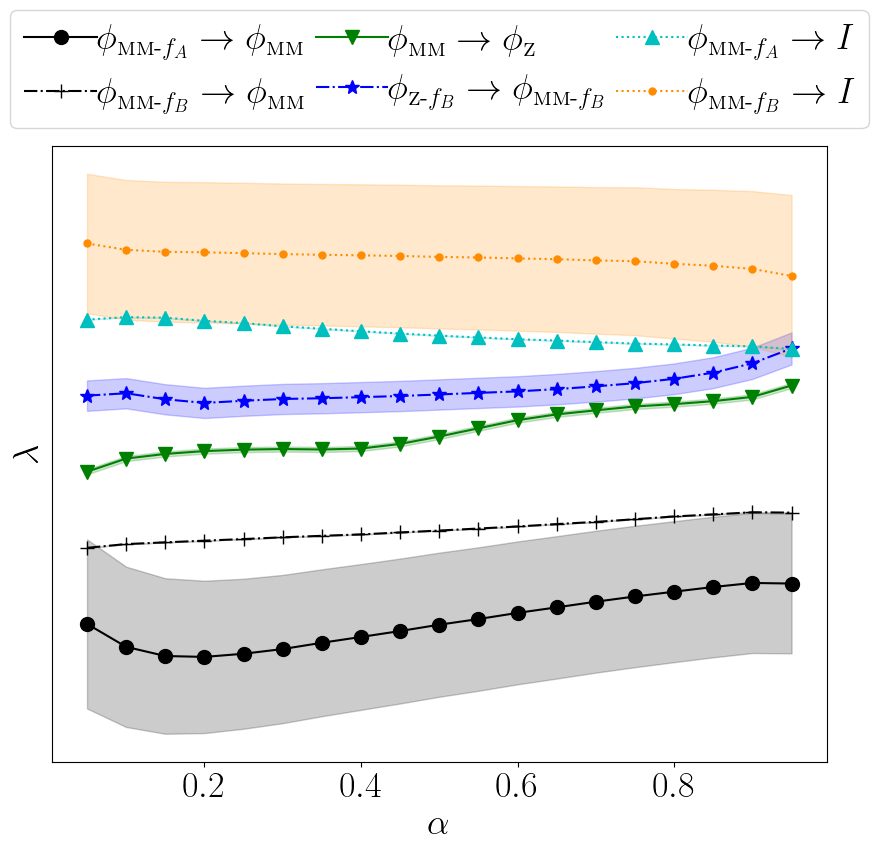}}
}
\centerline{
\subfloat[\textsc{Quora}]{
\includegraphics[trim={0 0 0 4cm},clip,width=0.36\linewidth,height=3cm]{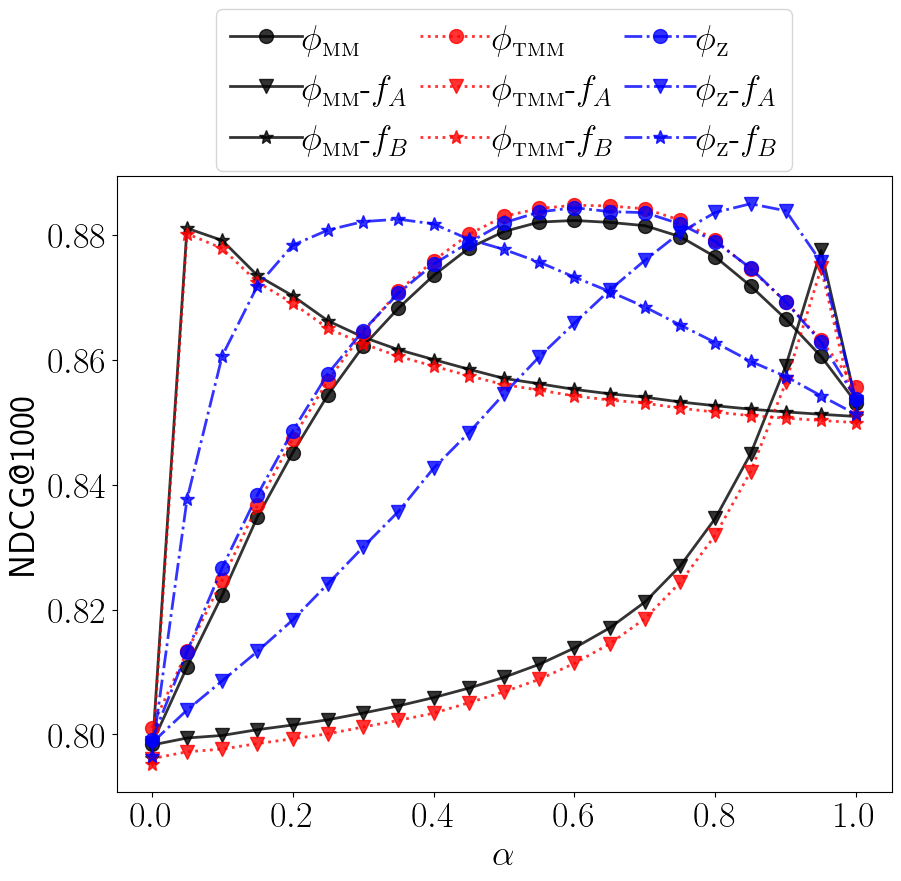}
\includegraphics[trim={0 0 0 3.6cm},clip,width=0.36\linewidth,height=3cm]{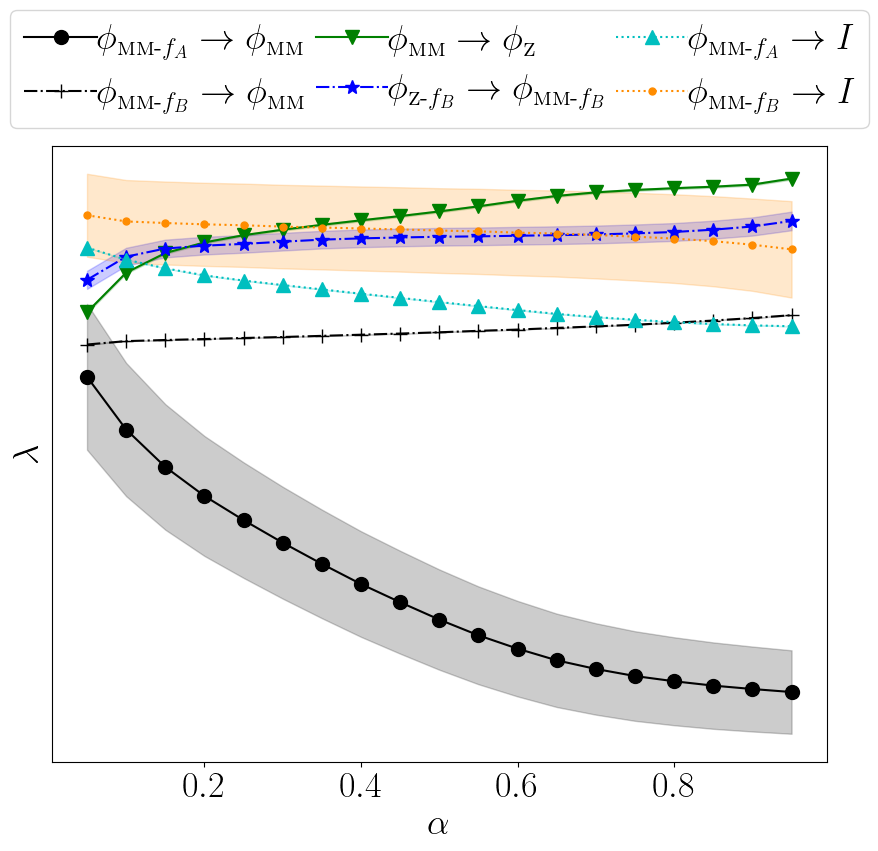}}
}
\centerline{
\subfloat[\textsc{HotpotQA}]{
\includegraphics[trim={0 0 0 4cm},clip,width=0.36\linewidth,height=3cm]{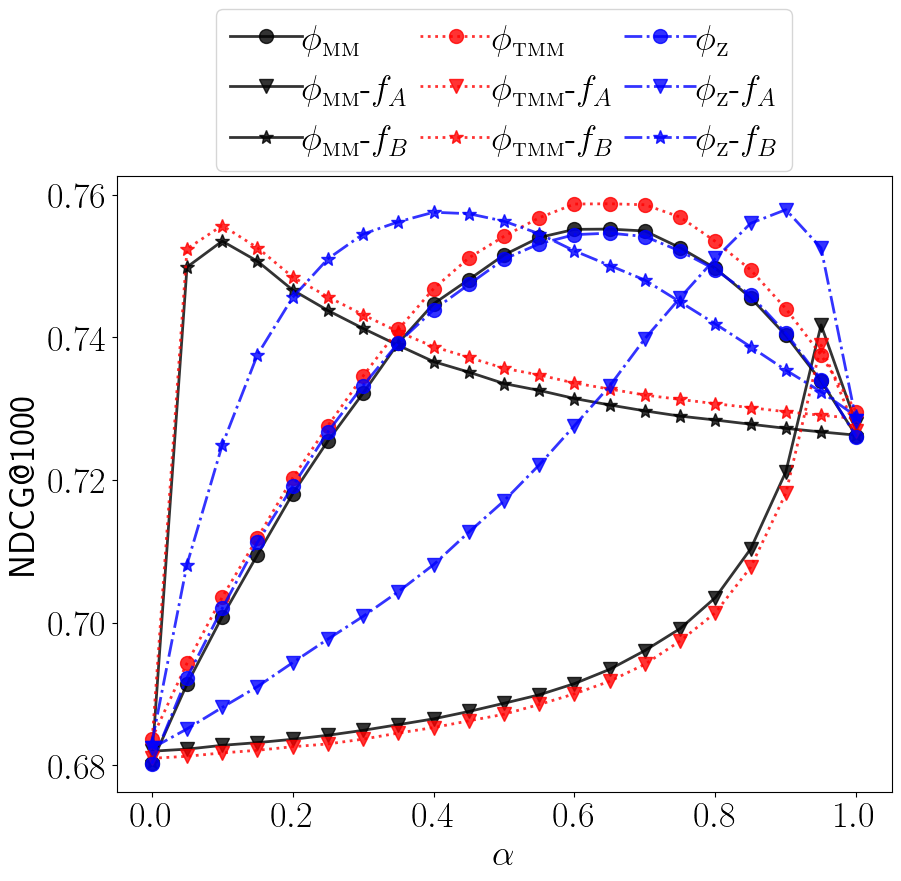}
\includegraphics[trim={0 0 0 3.6cm},clip,width=0.36\linewidth,height=3cm]{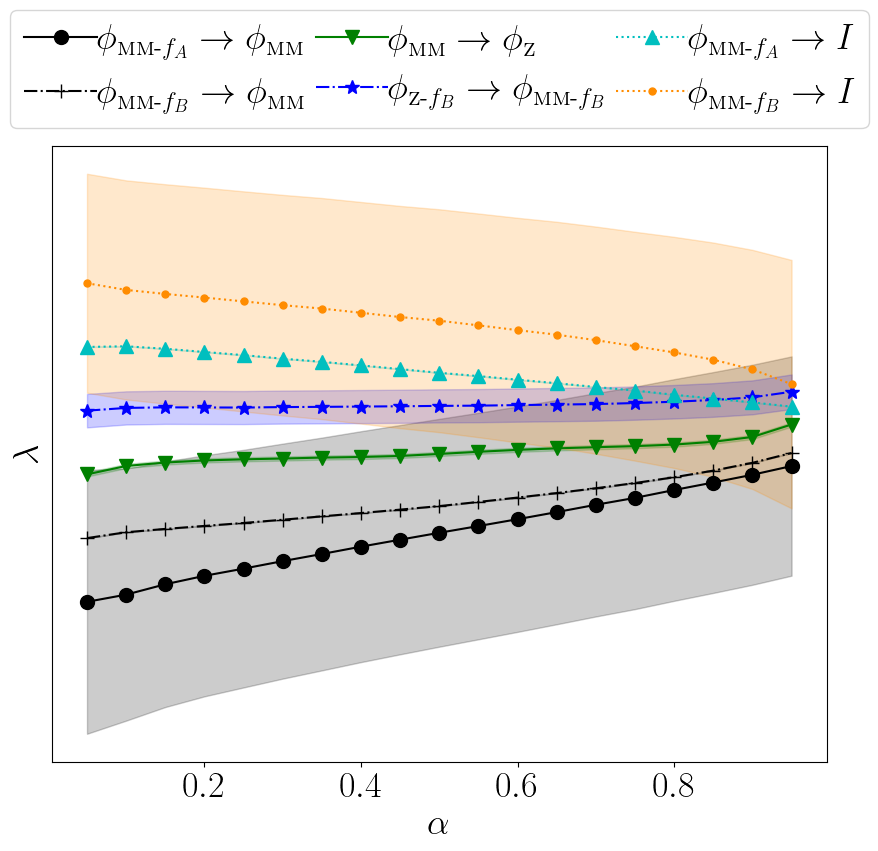}}
}
\centerline{
\subfloat[\textsc{FiQA}]{
\includegraphics[trim={0 0 0 4cm},clip,width=0.36\linewidth,height=3cm]{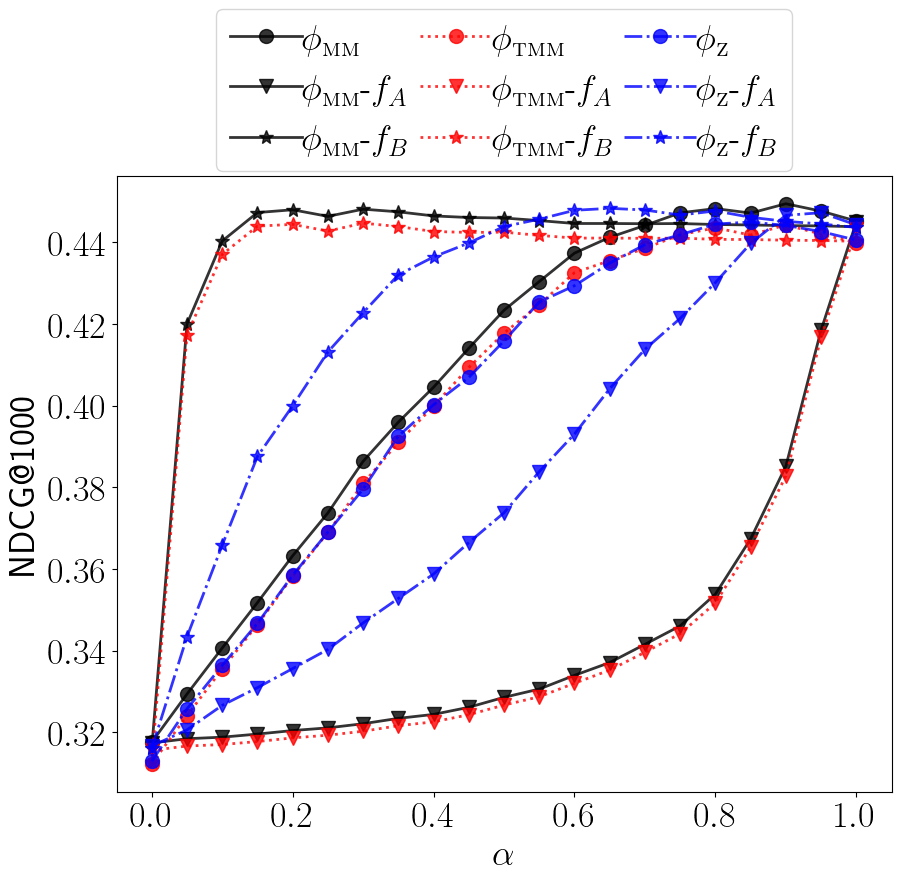}
\includegraphics[trim={0 0 0 3.6cm},clip,width=0.36\linewidth,height=3cm]{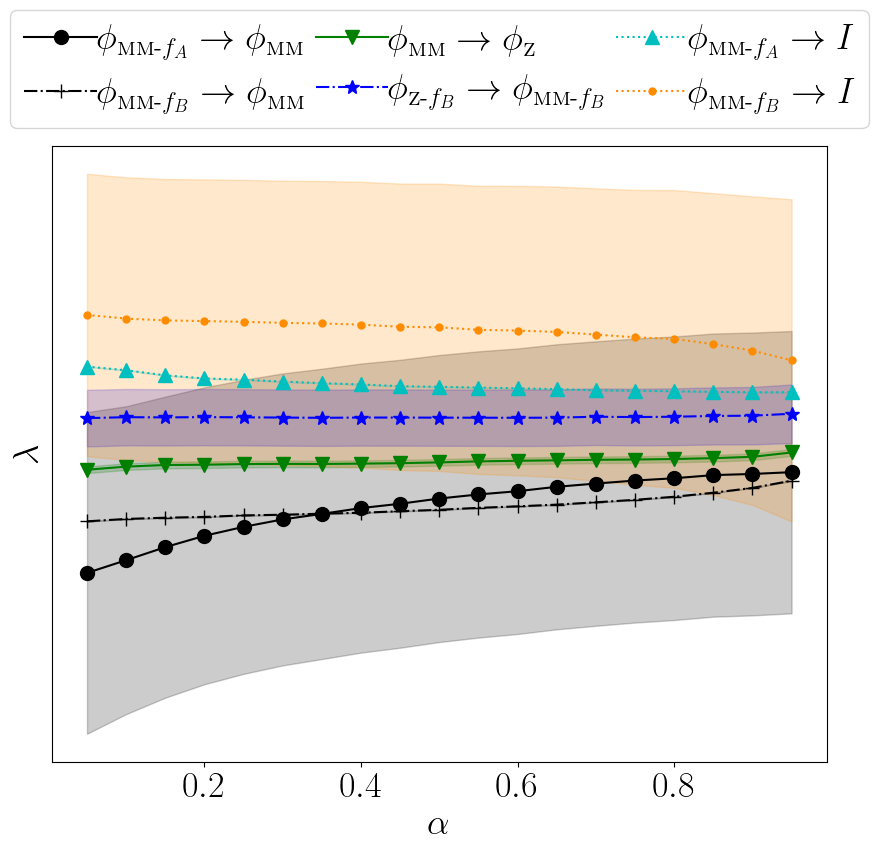}}
}
\caption{Effect of normalization on $f_\textsc{Convex} = \alpha f_A + (1 - \alpha)f_B$,
where $f_A$ is \textsc{Splade} and $f_B$ is BM25, as a function of $\alpha$ (left);
and, the relative expansion rate of \textsc{Splade} scores with respect to BM25 scores (i.e., $\lambda$ in Definition~\ref{definition:growth}), with $95\%$ confidence intervals (right). $\phi\text{-}f_o$ indicates that the normalization function $\phi$ was only applied to $f_o$, with the other function entering fusion without normalization.}
\label{appendix:splade-bm25:figure:convex:normalization}
\end{center}
\end{figure}
\FloatBarrier

\begin{table*}[t]
\caption{NDCG@1000 (except \textsc{SciFact} and \textsc{NFCorpus} where cutoff is $100$) on the test split of various datasets for individual systems and their fusion using \ac{rrf}~\cite{chen2022ecir} ($\eta=60$) and TM2C2 ($\alpha=0.8$ in $f_\textsc{Convex} = \alpha \textsc{Splade} + (1 - \alpha) \textsc{BM25}$). The symbols $\ddagger$ and $\ast$ indicate statistical significance ($p$-value $< 0.01$) with respect to TM2C2 and RRF respectively, according to a paired two-tailed $t$-test.}
\label{appendix:splade-bm25:table:rq1}
\begin{center}
\begin{sc}
\begin{tabular}{c|cccc}
& \multicolumn{4}{c}{NDCG} \\
\toprule
\small{Dataset} & \small{BM25} & \small{Splade} & \small{TM2C2} & \small{\ac{rrf}} \\
\midrule
\small{MS MARCO} & $0.309^{\ddagger\ast}$ & $\textbf{0.508}^{\ast}$ & 0.507 & $0.444^{\ddagger}$ \\
\midrule
\small{NQ} & $0.382^{\ddagger\ast}$ & $\textbf{0.591}^{\ddagger\ast}$ & 0.587 & $0.520^{\ddagger}$ \\
\small{Quora} & $0.798^{\ddagger\ast}$ & $0.853^{\ddagger}$ & \textbf{0.876} & $0.859^{\ddagger}$ \\
\small{NFCorpus} & $0.269^{\ddagger\ast}$ & $0.314^{\ast}$ & \textbf{0.317} & $0.304^{\ddagger}$ \\
\small{HotpotQA} & $0.682^{\ddagger\ast}$ & $0.727^{\ddagger\ast}$ & \textbf{0.751} & $0.737^{\ddagger}$ \\
\small{FEVER} & $0.689^{\ddagger\ast}$ & $0.806^{\ddagger\ast}$ & \textbf{0.825} & $0.786^{\ddagger}$ \\
\small{SciFact} & $0.698^{\ddagger\ast}$ & $0.723^{\ddagger\ast}$ & \textbf{0.740} & $0.732^{\ddagger}$ \\
\small{DBPedia} & $0.415^{\ddagger\ast}$ & $0.546^{\ddagger\ast}$ & \textbf{0.556} & $0.526^{\ddagger}$ \\
\small{FiQA} & $0.315^{\ddagger\ast}$ & $0.442^{\ast}$ & \textbf{0.446} & $0.406^{\ddagger}$ \\
\bottomrule
\end{tabular}
\end{sc}
\end{center}
\end{table*}

\begin{figure}[t]
\begin{center}
\centerline{
\subfloat[in-domain]{
\includegraphics[height=2in]{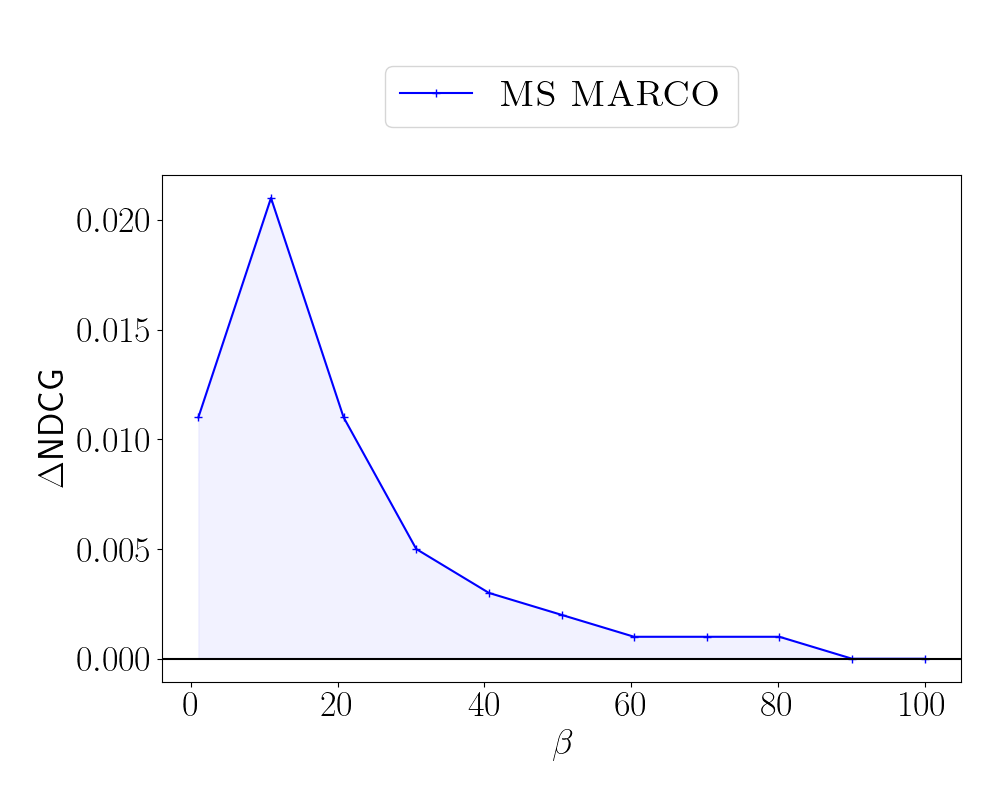}}
\subfloat[out-of-domain]{
\includegraphics[height=2in]{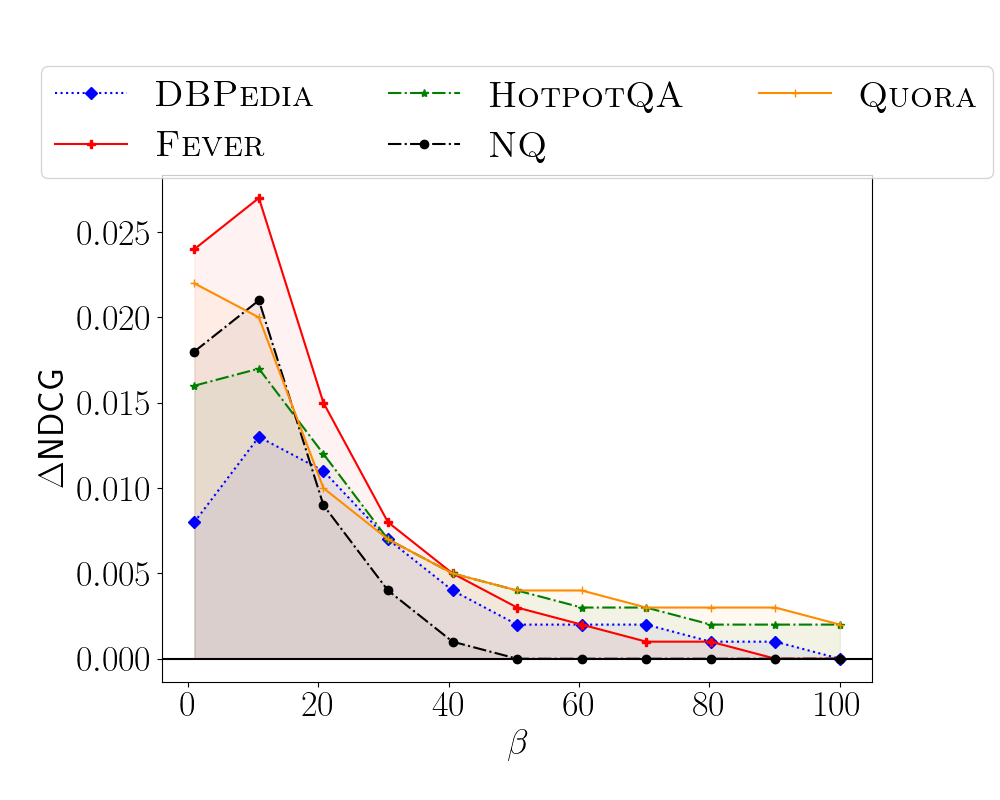}}
}
\caption{The difference in NDCG@1000 of $f_\textsc{SRRF}$ and $f_\textsc{RRF}$ with $\eta=60$ (positive indicates better ranking quality by \textsc{SRRF}) as a function of $\beta$.}
\label{appendix:splade-bm25:figure:lipschitz}
\end{center}
\end{figure}

\FloatBarrier

\section{Fusion of \textsc{Tas-B} and BM25}
\label{appendix:tasb-bm25}

\begin{figure}[h]
\begin{center}
\centerline{
\subfloat[MS MARCO]{
\includegraphics[width=0.36\linewidth,height=4cm]{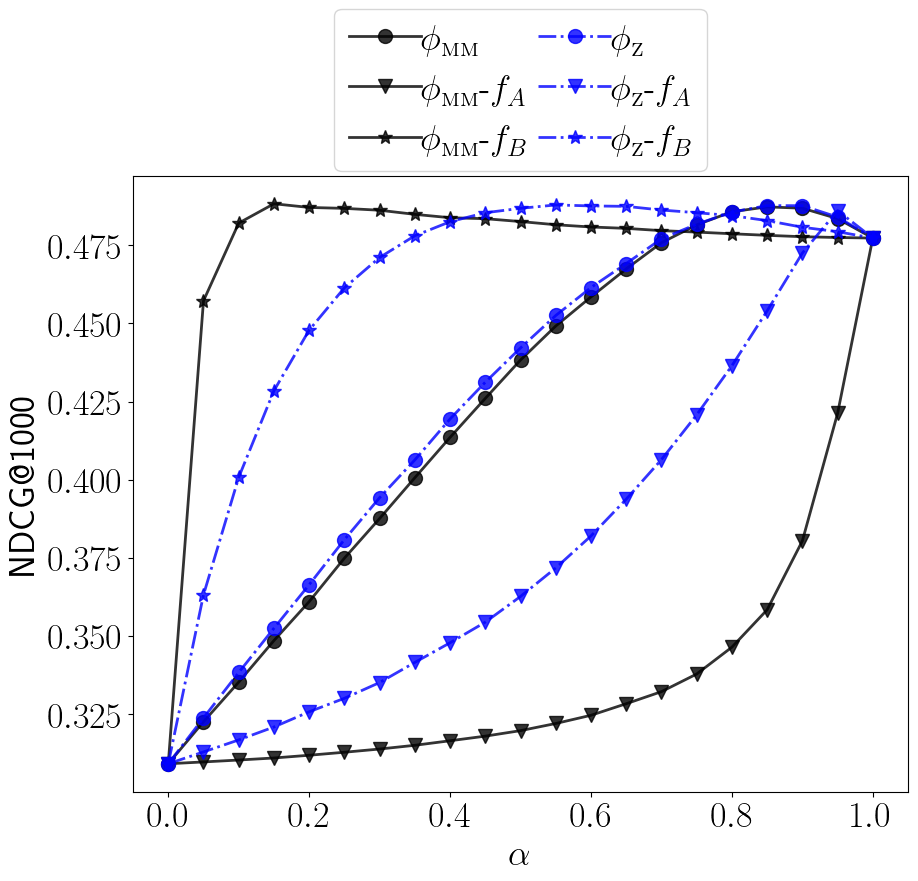}
\includegraphics[width=0.36\linewidth,height=4cm]{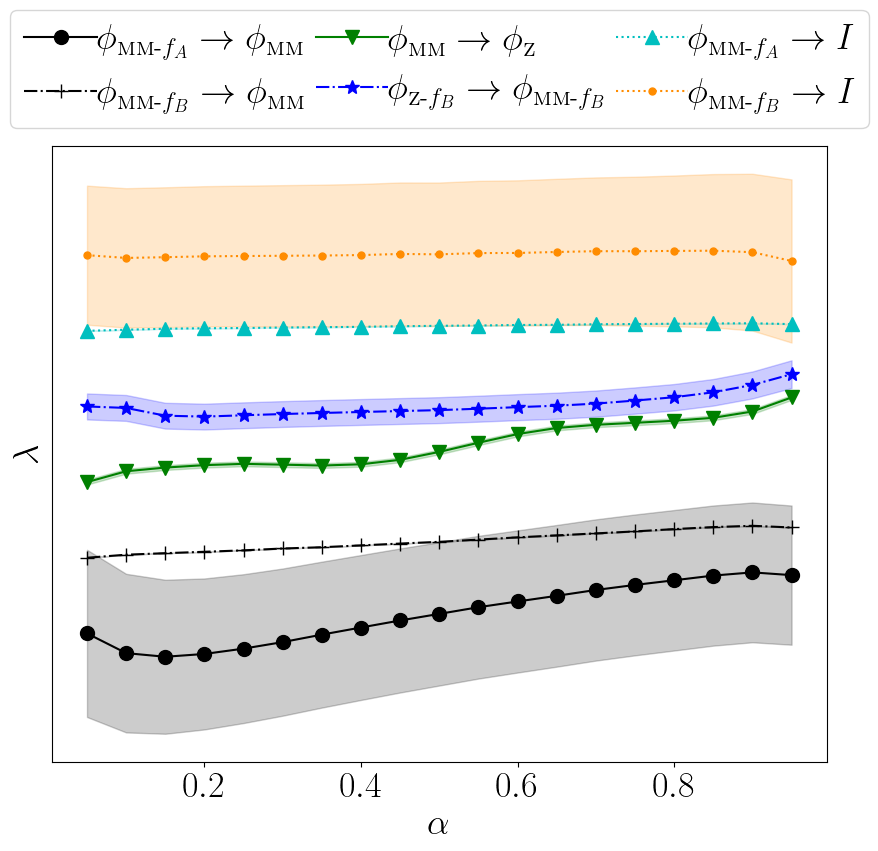}}
}
\centerline{
\subfloat[\textsc{Quora}]{
\includegraphics[trim={0 0 0 4cm},clip,width=0.36\linewidth,height=3cm]{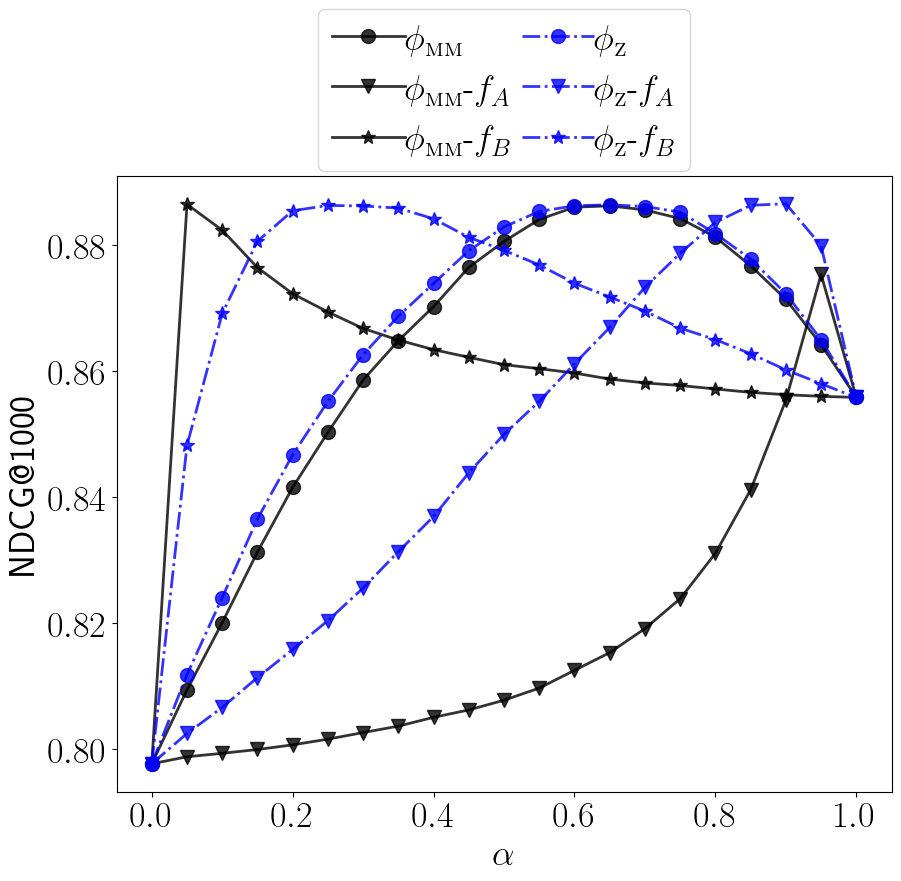}
\includegraphics[trim={0 0 0 3.6cm},clip,width=0.36\linewidth,height=3cm]{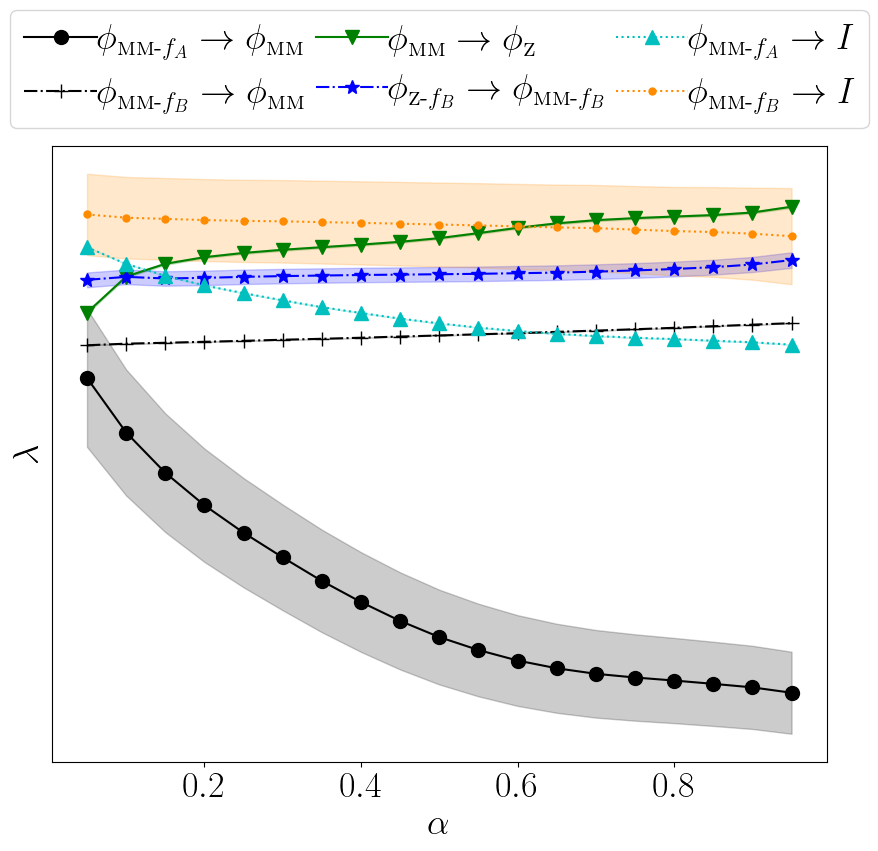}}
}
\centerline{
\subfloat[\textsc{HotpotQA}]{
\includegraphics[trim={0 0 0 4cm},clip,width=0.36\linewidth,height=3cm]{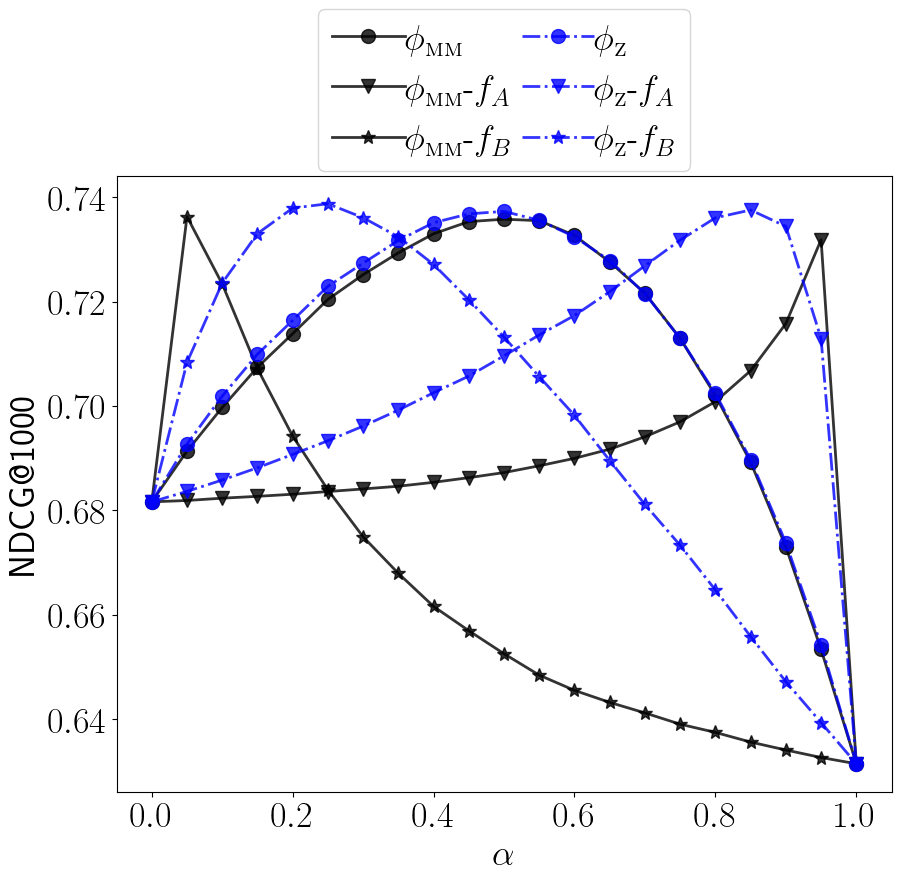}
\includegraphics[trim={0 0 0 3.6cm},clip,width=0.36\linewidth,height=3cm]{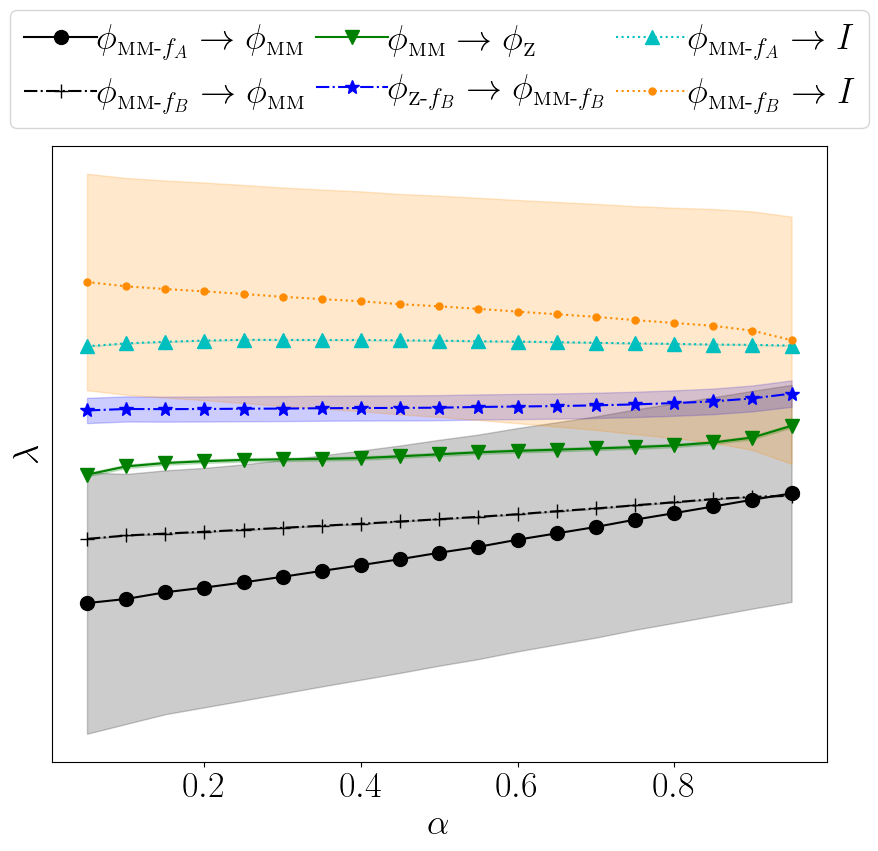}}
}
\centerline{
\subfloat[\textsc{FiQA}]{
\includegraphics[trim={0 0 0 4cm},clip,width=0.36\linewidth,height=3cm]{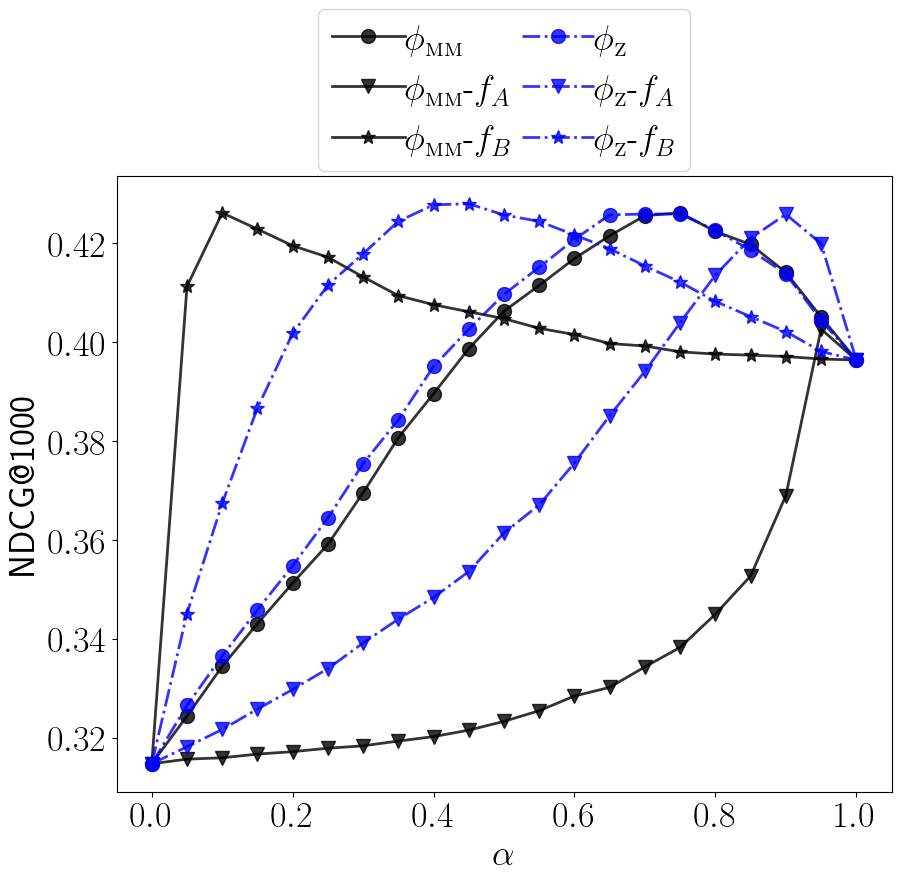}
\includegraphics[trim={0 0 0 3.6cm},clip,width=0.36\linewidth,height=3cm]{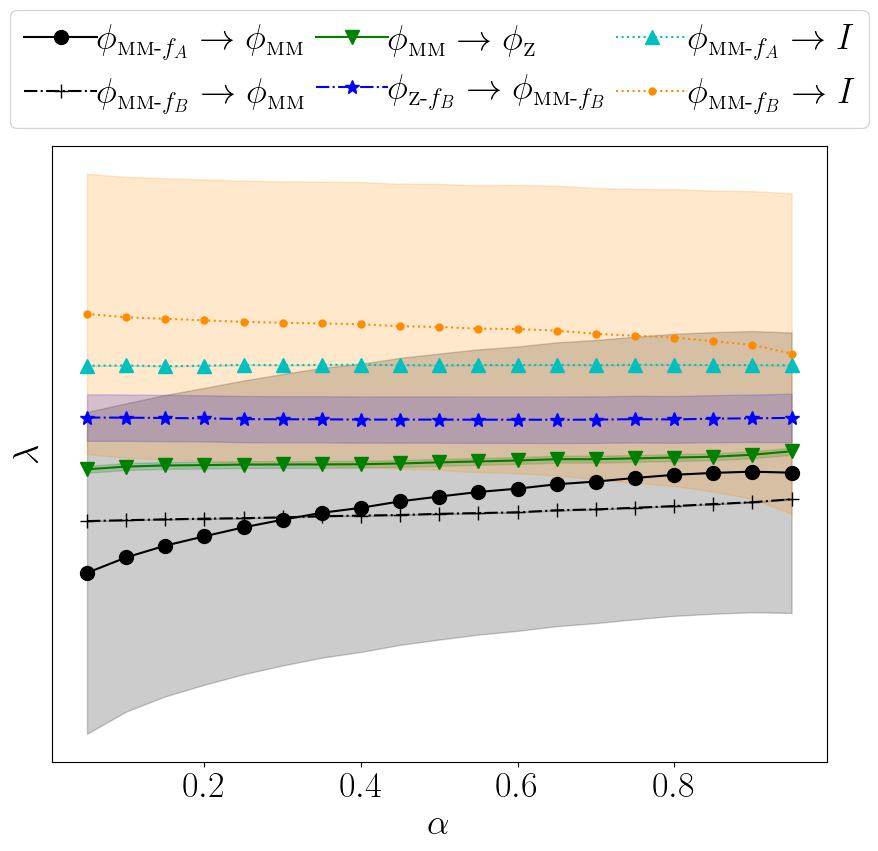}}
}
\caption{Effect of normalization on $f_\textsc{Convex} = \alpha f_A + (1 - \alpha) f_B$,
where $f_A$ is the \textsc{Tas-B} function and $f_B$ is BM25, as a function of $\alpha$ (left);
and, the relative expansion rate of \textsc{Tas-B} scores with respect to BM25 scores (i.e., $\lambda$ in Definition~\ref{definition:growth}), with $95\%$ confidence intervals (right). $\phi\text{-}f_o$ indicates that the normalization function $\phi$ was only applied to $f_o$, with the other function entering fusion without normalization.}
\label{appendix:tasb-bm25:figure:convex:normalization}
\end{center}
\end{figure}
\FloatBarrier

\begin{table*}[t]
\caption{NDCG@1000 (except \textsc{SciFact} and \textsc{NFCorpus} where cutoff is $100$) on the test split of various datasets for individual systems and their fusion using \ac{rrf}~\cite{chen2022ecir} ($\eta=60$) and M2C2 ($\alpha=0.8$ in $\alpha \textsc{Tas-B} + (1 - \alpha)\textsc{BM25}$). The symbols $\ddagger$ and $\ast$ indicate statistical significance ($p$-value $< 0.01$) with respect to M2C2 and RRF respectively, according to a paired two-tailed $t$-test.}
\label{appendix:tasb-bm25:table:rq1}
\begin{center}
\begin{sc}
\begin{tabular}{c|cccc}
& \multicolumn{4}{c}{NDCG} \\
\toprule
\small{Dataset} & \small{BM25} & \small{Tas-B} & \small{M2C2} & \small{\ac{rrf}} \\
\midrule
\small{MS MARCO} & $0.309^{\ddagger\ast}$ & $0.477^{\ddagger\ast}$ & \textbf{0.486} & $0.434^{\ddagger}$ \\
\midrule
\small{NQ} & $0.382^{\ddagger\ast}$ & $0.522^{\ddagger\ast}$ & \textbf{0.552} & $0.513^{\ddagger}$ \\
\small{Quora} & $0.798^{\ddagger\ast}$ & $0.856^{\ddagger}$ & \textbf{0.881} & $0.860^{\ddagger}$ \\
\small{NFCorpus} & $0.269^{\ddagger\ast}$ & $0.292^{\ddagger}$ & \textbf{0.307} & $0.299$ \\
\small{HotpotQA} & $0.682^{\ddagger\ast}$ & $0.631^{\ddagger\ast}$ & $0.702$ & $\textbf{0.711}^{\ddagger}$ \\
\small{FEVER} & $0.689^{\ddagger\ast}$ & $0.725^{\ddagger\ast}$ & \textbf{0.785} & $0.776^{\ddagger}$ \\
\small{SciFact} & $0.698^{\ddagger\ast}$ & $0.670^{\ddagger\ast}$ & $0.715$ & $\textbf{0.719}$ \\
\small{DBPedia} & $0.415^{\ddagger\ast}$ & $0.495^{\ddagger\ast}$ & \textbf{0.533} & $0.512^{\ddagger}$ \\
\small{FiQA} & $0.315^{\ddagger\ast}$ & $0.396^{\ddagger\ast}$ & \textbf{0.422} & $0.401^{\ddagger}$ \\
\bottomrule
\end{tabular}
\end{sc}
\end{center}
\end{table*}

\begin{figure}[t]
\begin{center}
\centerline{
\subfloat[in-domain]{
\includegraphics[height=2in]{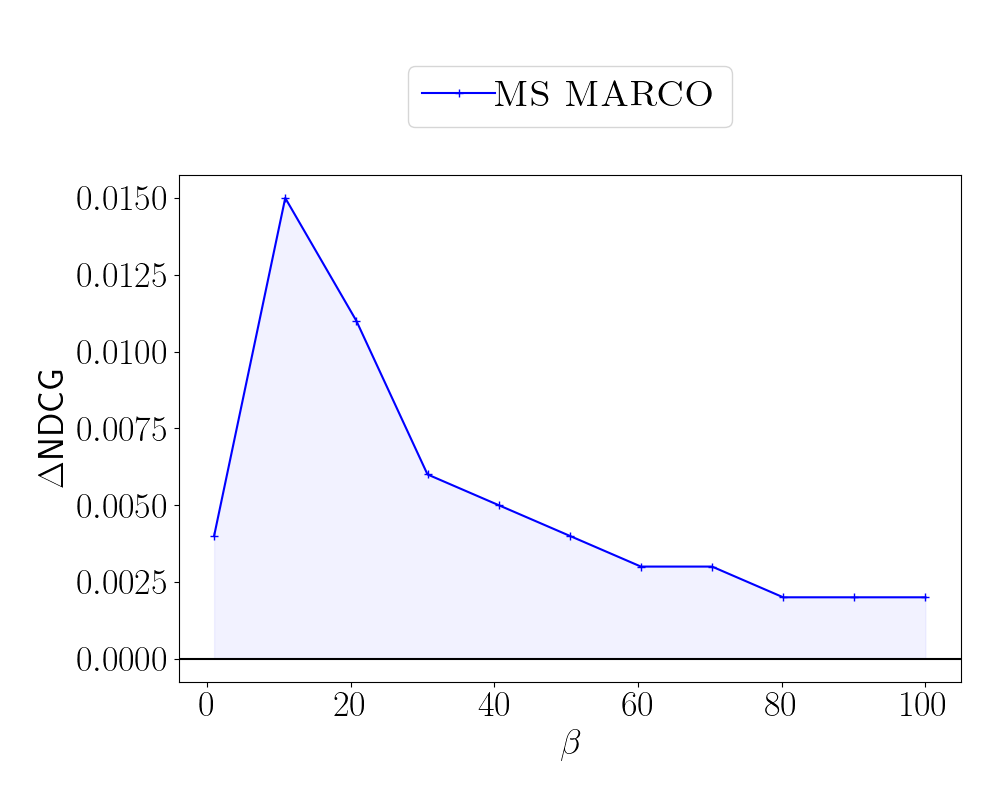}}
\subfloat[out-of-domain]{
\includegraphics[height=2in]{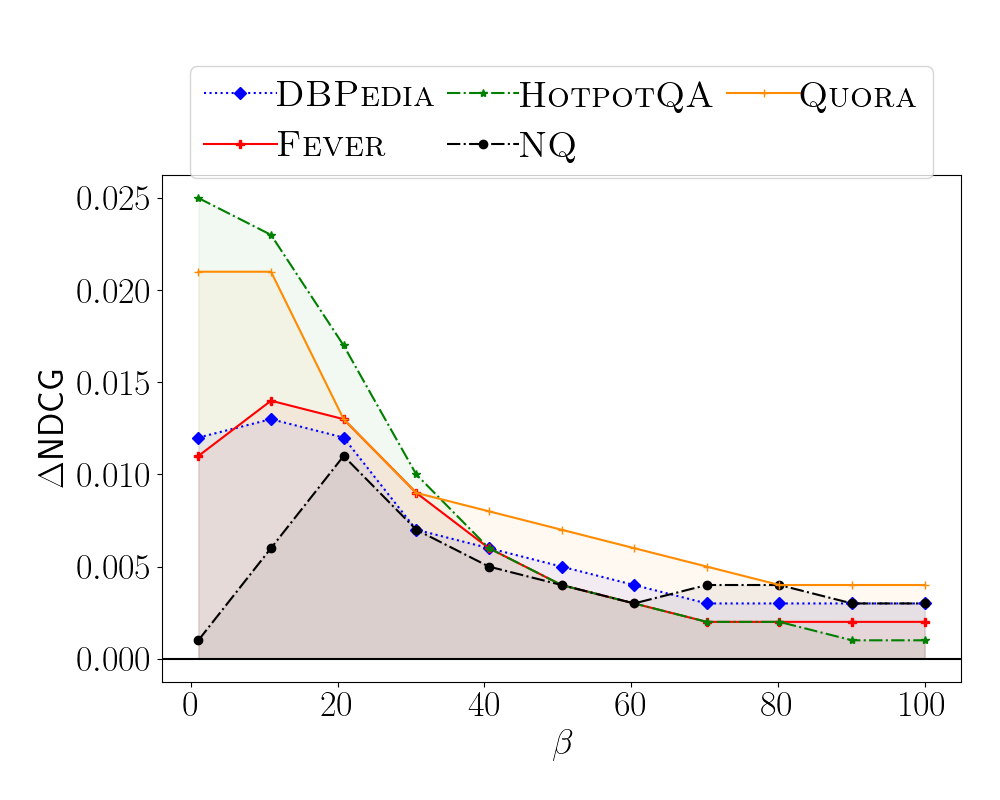}}
}
\caption{The difference in NDCG@1000 of $f_\textsc{SRRF}$ and $f_\textsc{RRF}$ with $\eta=60$ (positive indicates better ranking quality by \textsc{SRRF}) as a function of $\beta$.}
\label{appendix:tasb-bm25:figure:lipschitz}
\end{center}
\end{figure}

\FloatBarrier

\section{Fusion of \textsc{Tas-B} and \textsc{Splade}}
\label{appendix:tasb-splade}

\begin{figure}[h]
\begin{center}
\centerline{
\subfloat[MS MARCO]{
\includegraphics[width=0.36\linewidth,height=4cm]{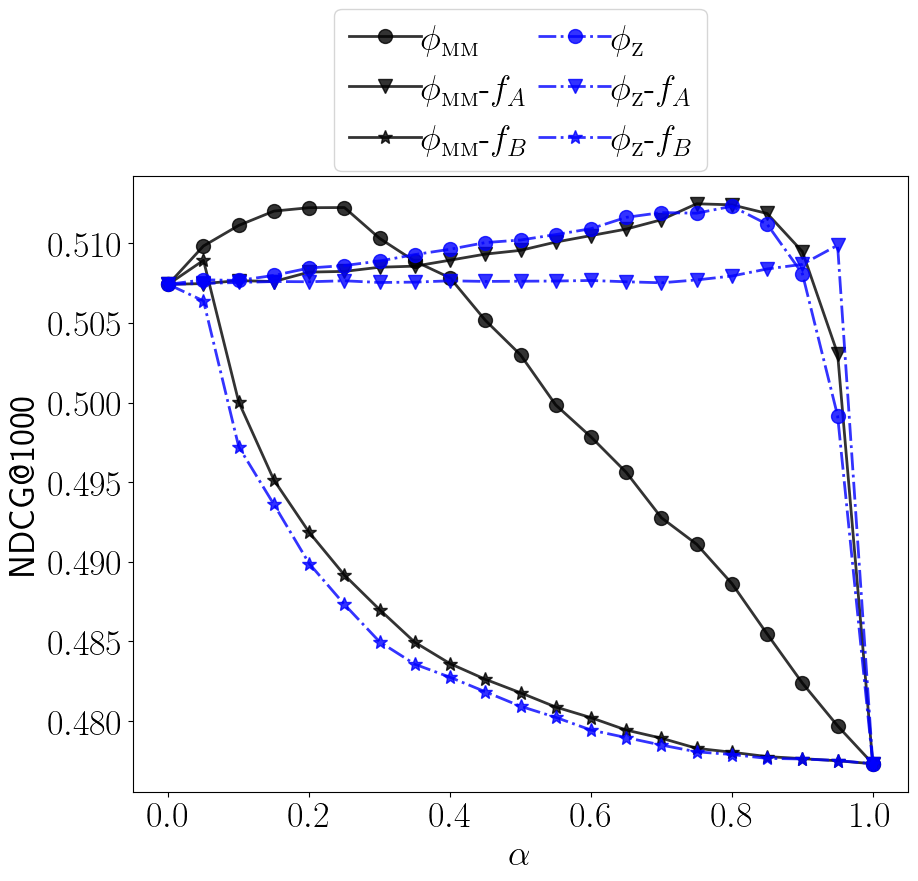}
\includegraphics[width=0.36\linewidth,height=4cm]{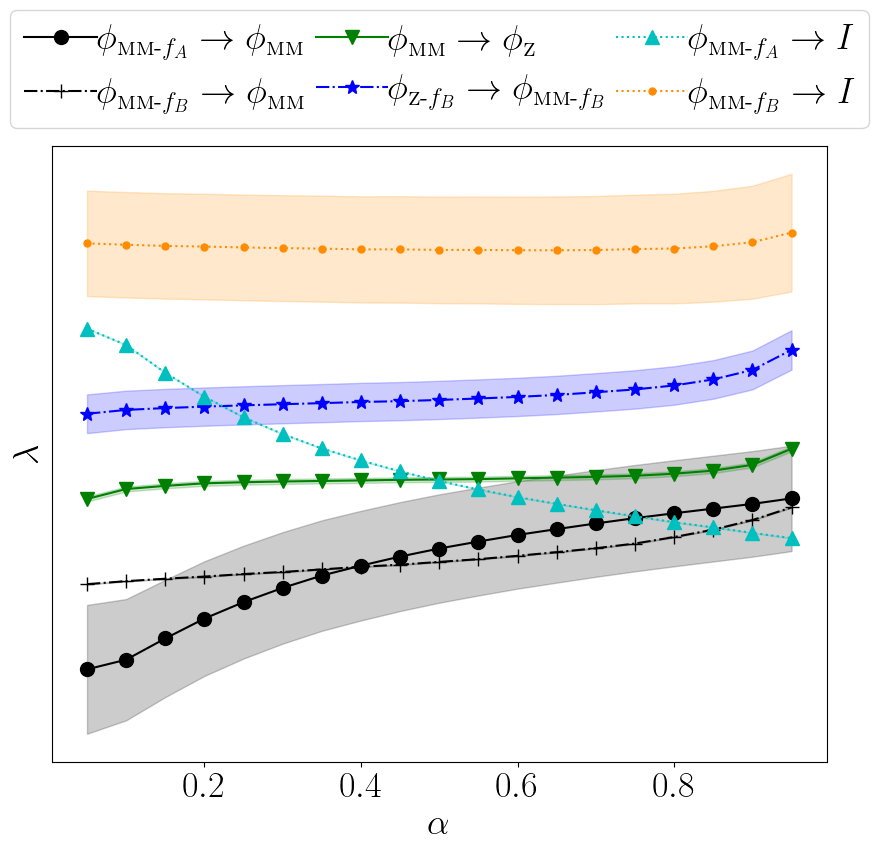}}
}
\centerline{
\subfloat[\textsc{Quora}]{
\includegraphics[trim={0 0 0 4cm},clip,width=0.36\linewidth,height=3cm]{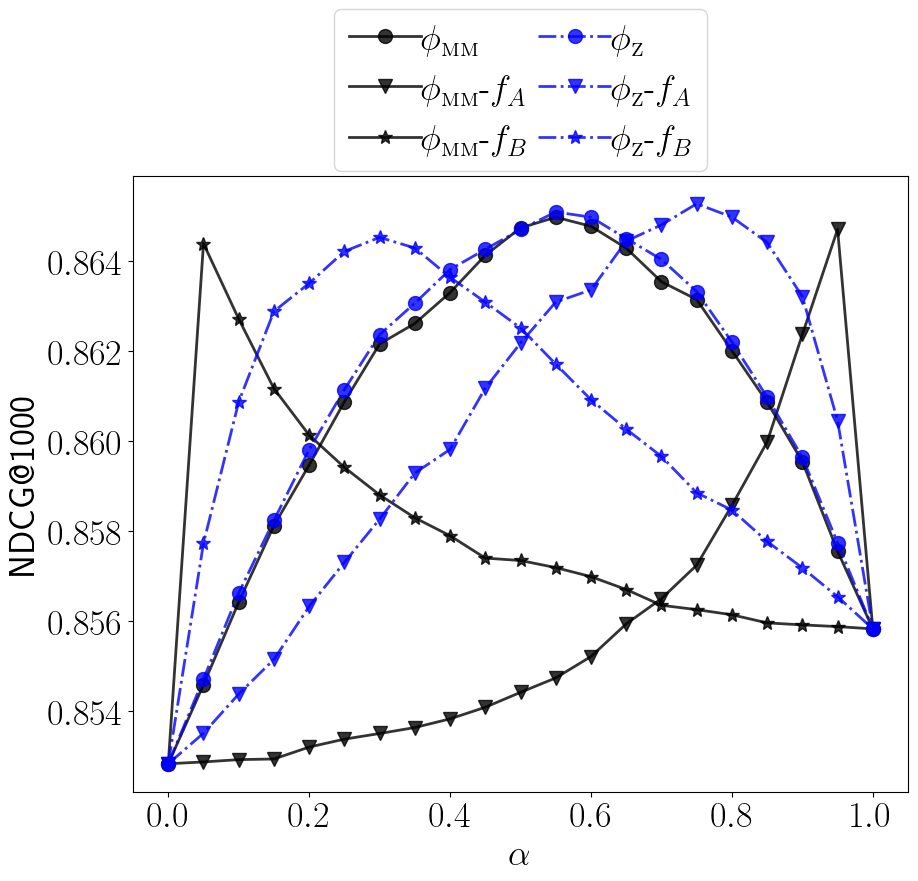}
\includegraphics[trim={0 0 0 3.6cm},clip,width=0.36\linewidth,height=3cm]{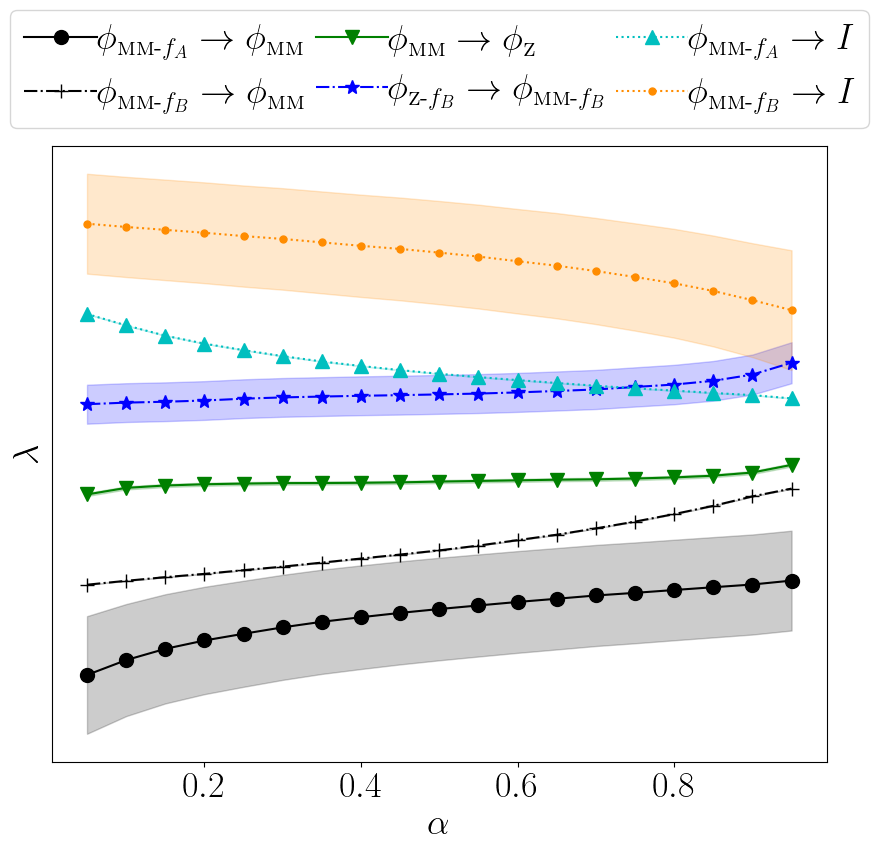}}
}
\centerline{
\subfloat[\textsc{HotpotQA}]{
\includegraphics[trim={0 0 0 4cm},clip,width=0.36\linewidth,height=3cm]{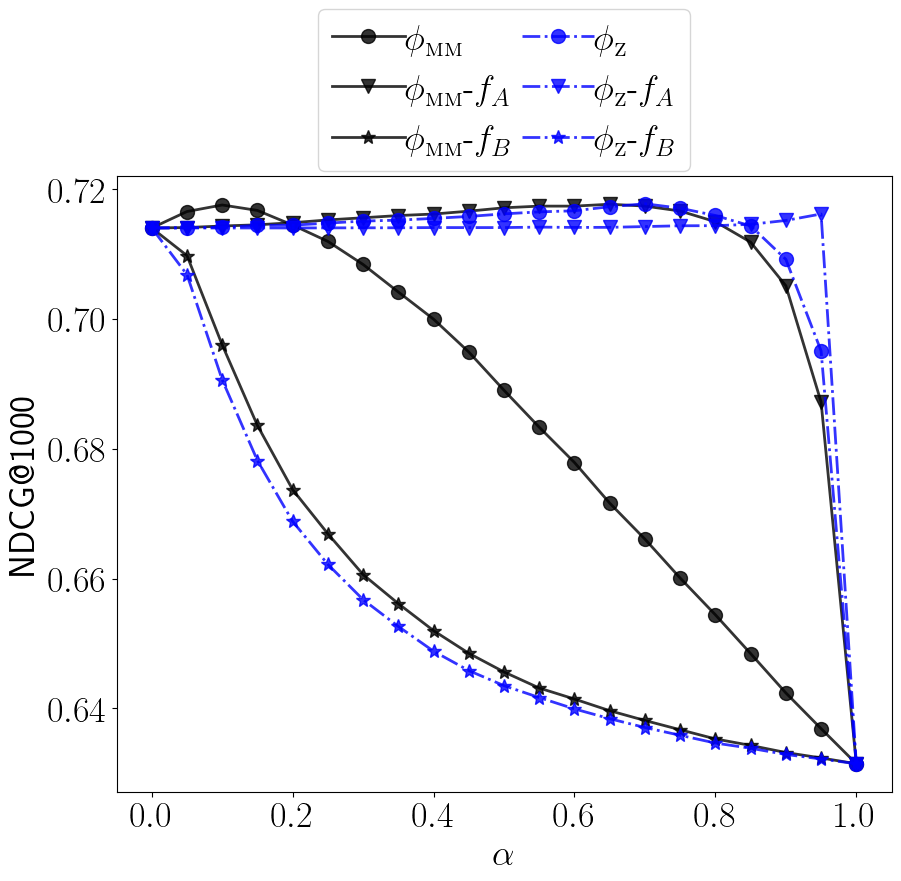}
\includegraphics[trim={0 0 0 3.6cm},clip,width=0.36\linewidth,height=3cm]{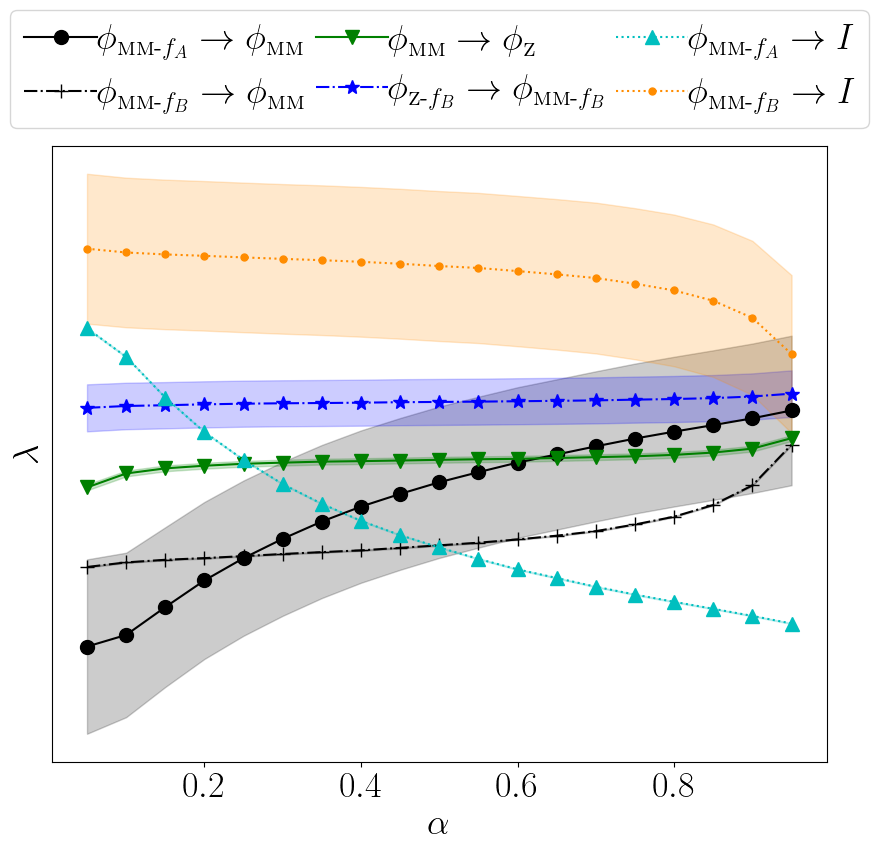}}
}
\centerline{
\subfloat[\textsc{FiQA}]{
\includegraphics[trim={0 0 0 4cm},clip,width=0.36\linewidth,height=3cm]{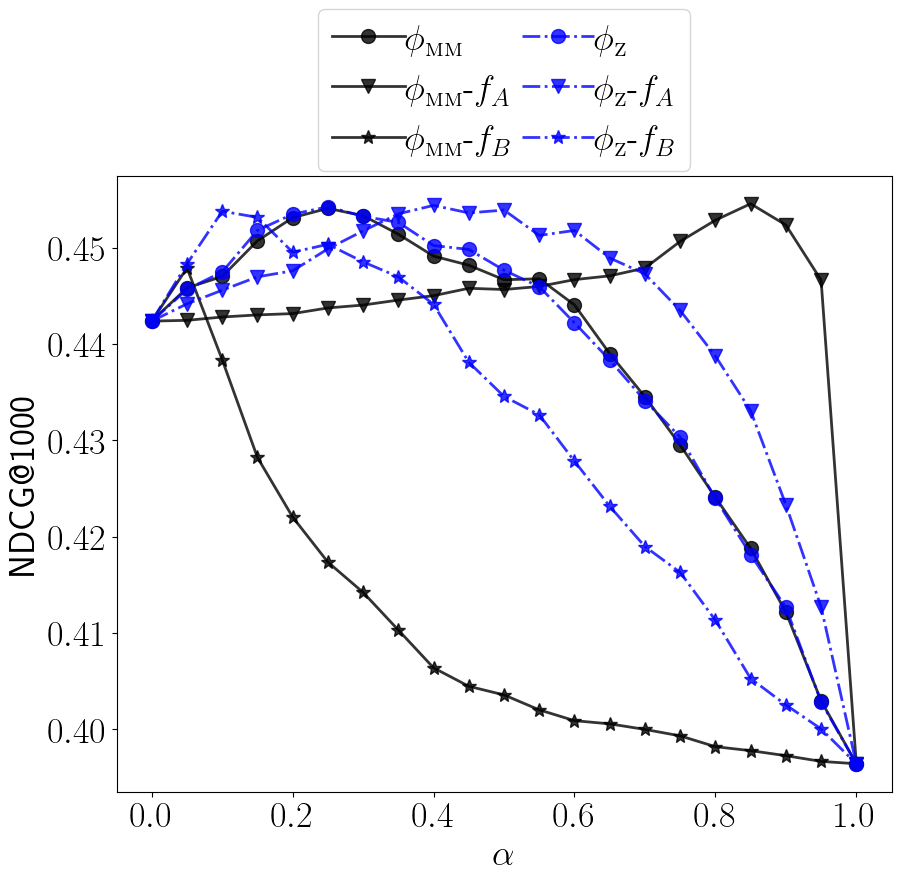}
\includegraphics[trim={0 0 0 3.6cm},clip,width=0.36\linewidth,height=3cm]{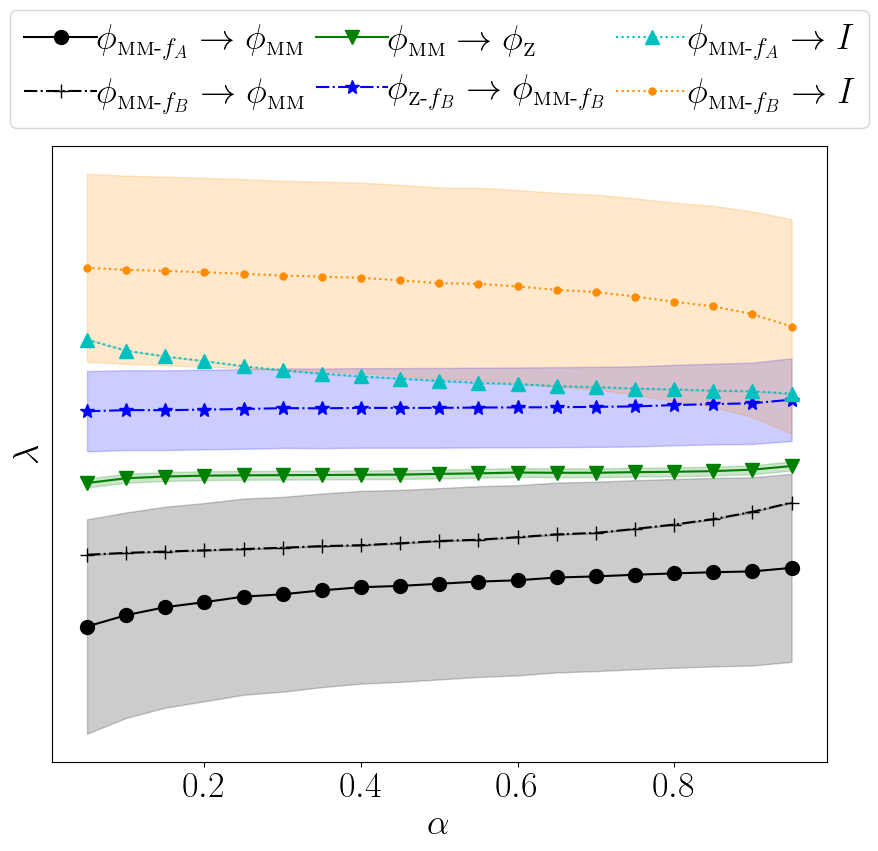}}
}
\caption{Effect of normalization on $f_\textsc{Convex} = \alpha f_A + (1 - \alpha) f_B$,
where $f_A$ is the \textsc{Tas-B} function and $f_B$ is \textsc{Splade}, as a function of $\alpha$ (left);
and, the relative expansion rate of \textsc{Tas-B} scores with respect to \textsc{Splade} scores (i.e., $\lambda$ in Definition~\ref{definition:growth}), with $95\%$ confidence intervals (right). $\phi\text{-}f_o$ indicates that the normalization function $\phi$ was only applied to $f_o$, with the other function entering fusion without normalization.}
\label{appendix:tasb-splade:figure:convex:normalization}
\end{center}
\end{figure}
\FloatBarrier

\begin{table*}[t]
\caption{NDCG@1000 (except \textsc{SciFact} and \textsc{NFCorpus} where cutoff is $100$) on the test split of various datasets for individual systems and their fusion using \ac{rrf}~\cite{chen2022ecir} ($\eta=60$) and M2C2 ($\alpha=0.2$ in $\alpha \textsc{Tas-B} + (1 - \alpha) \textsc{Splade}$). The symbols $\ddagger$ and $\ast$ indicate statistical significance ($p$-value $< 0.01$) with respect to M2C2 and RRF respectively, according to a paired two-tailed $t$-test.}
\label{appendix:tasb-splade:table:rq1}
\begin{center}
\begin{sc}
\begin{tabular}{c|cccc}
& \multicolumn{4}{c}{NDCG} \\
\toprule
\small{Dataset} & \small{Splade} & \small{Tas-B} & \small{M2C2} & \small{\ac{rrf}} \\
\midrule
\small{MS MARCO} & $0.508$ & $0.477^{\ddagger\ast}$ & \textbf{0.512} & $0.504^{\ddagger}$ \\
\midrule
\small{NQ} & $0.589$ & $0.522^{\ddagger\ast}$ & \textbf{0.593} & $0.581^{\ddagger}$ \\
\small{Quora} & $0.853$ & $0.856$ & \textbf{0.859} & $0.857$ \\
\small{NFCorpus} & $0.313$ & $0.292^{\ddagger\ast}$ & \textbf{0.318} & $0.314$ \\
\small{HotpotQA} & $0.727^{\ast}$ & $0.631^{\ddagger\ast}$ & $\textbf{0.728}$ & $0.686^{\ddagger}$ \\
\small{FEVER} & $0.806^{\ast}$ & $0.725^{\ddagger\ast}$ & $\textbf{0.811}$ & $0.795^{\ddagger}$ \\
\small{SciFact} & $0.723^{\ast}$ & $0.670^{\ddagger\ast}$ & \textbf{0.730} & $0.709^{\ddagger}$ \\
\small{DBPedia} & $0.546^{\ddagger}$ & $0.495^{\ddagger\ast}$ & $\textbf{0.555}$ & $0.545^{\ddagger}$ \\
\small{FiQA} & $0.442^{\ddagger}$ & $0.396^{\ddagger\ast}$ & $\textbf{0.453}$ & $0.440^{\ddagger}$ \\
\bottomrule
\end{tabular}
\end{sc}
\end{center}
\end{table*}

\begin{figure}[t]
\begin{center}
\centerline{
\subfloat[in-domain]{
\includegraphics[height=2in]{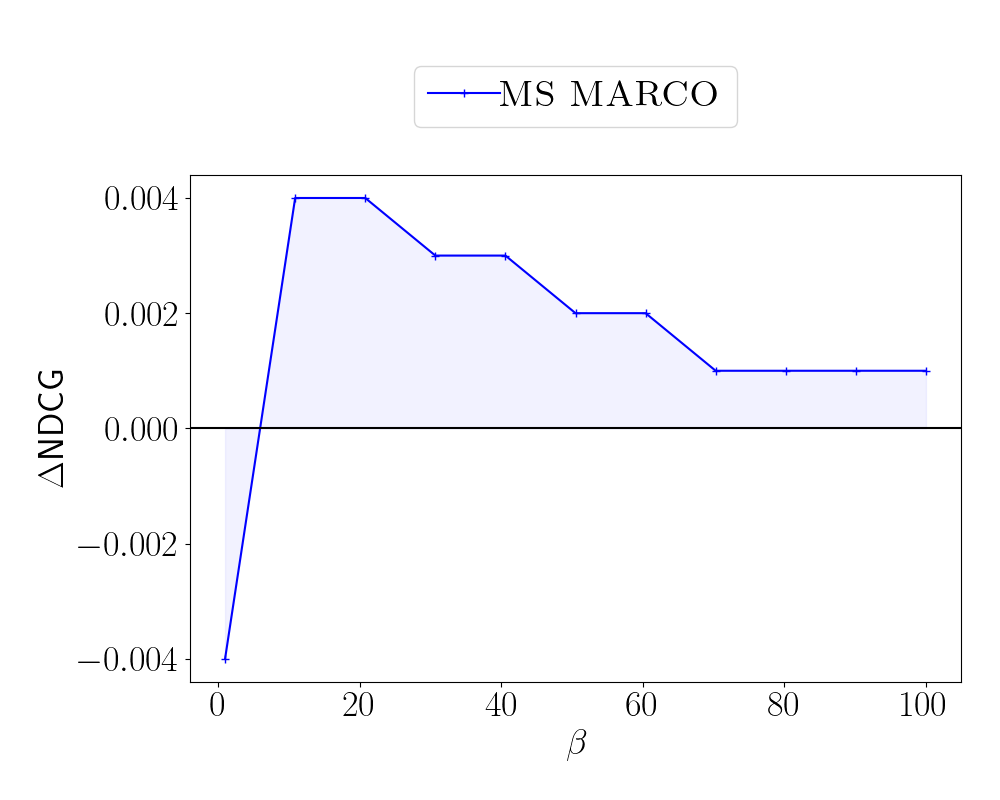}}
\subfloat[out-of-domain]{
\includegraphics[height=2in]{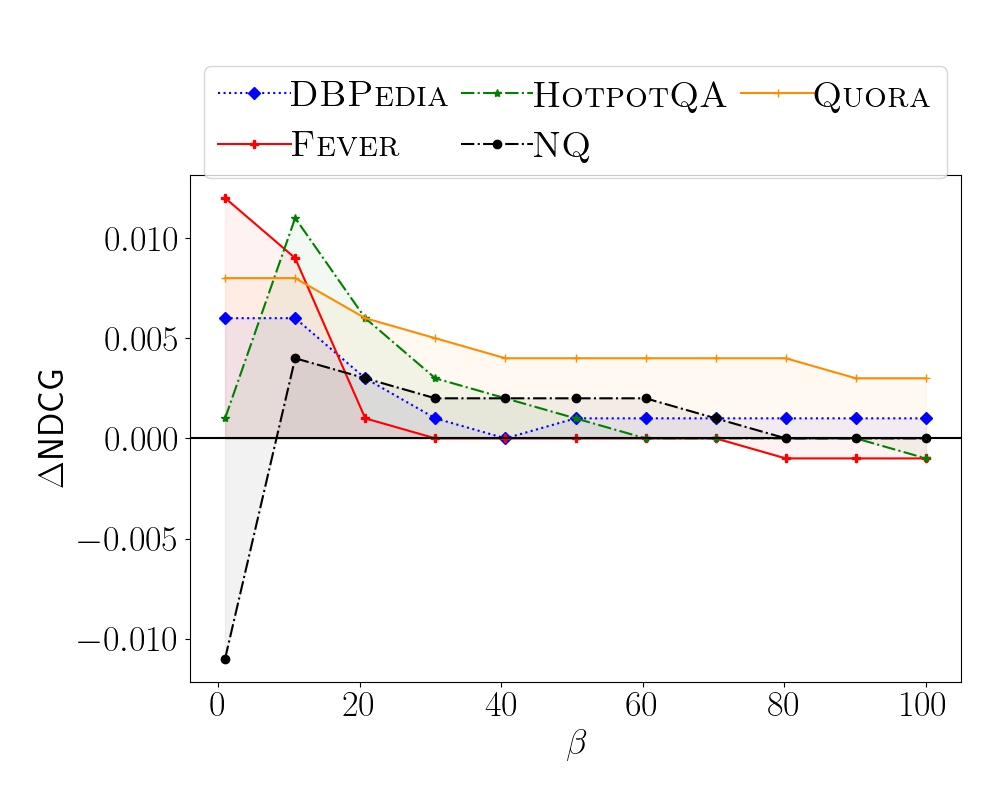}}
}
\caption{The difference in NDCG@1000 of $f_\textsc{SRRF}$ and $f_\textsc{RRF}$ with $\eta=60$ (positive indicates better ranking quality by \textsc{SRRF}) as a function of $\beta$.}
\label{appendix:tasb-splade:figure:lipschitz}
\end{center}
\end{figure}

\FloatBarrier

\section{Fusion of \textsc{Tas-B} and \textsc{All-MiniLM-L6-v2}}
\label{appendix:tasb-minilm}

\begin{figure}[h]
\begin{center}
\centerline{
\subfloat[MS MARCO]{
\includegraphics[width=0.36\linewidth,height=4cm]{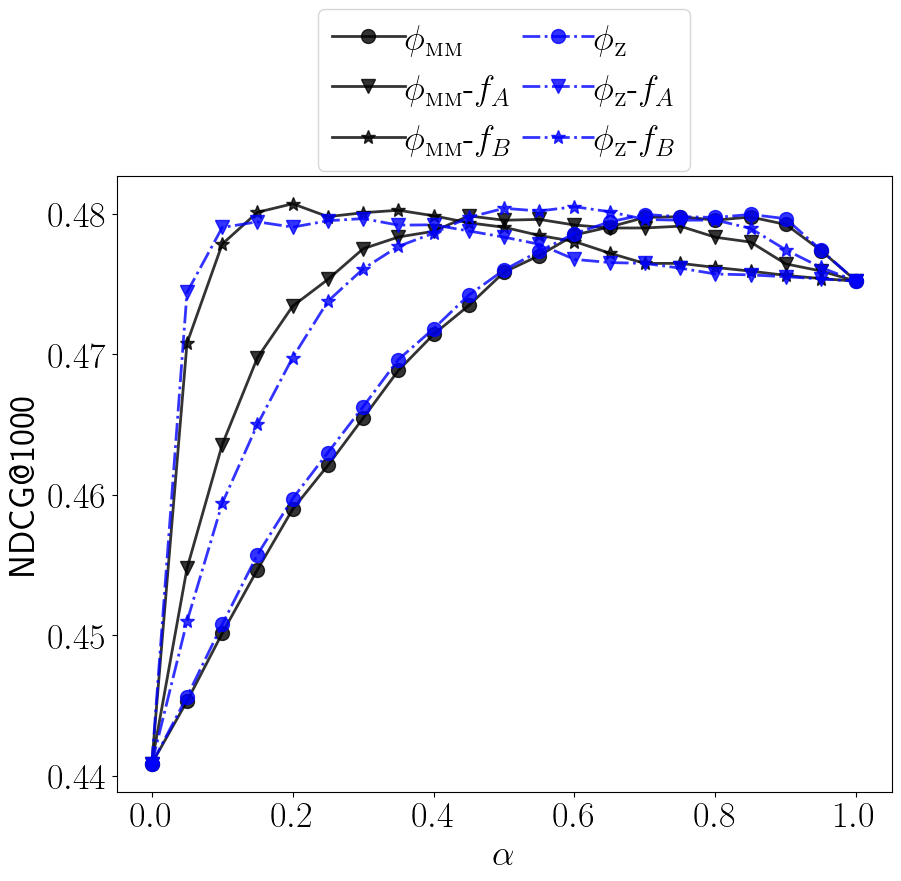}
\includegraphics[width=0.36\linewidth,height=4cm]{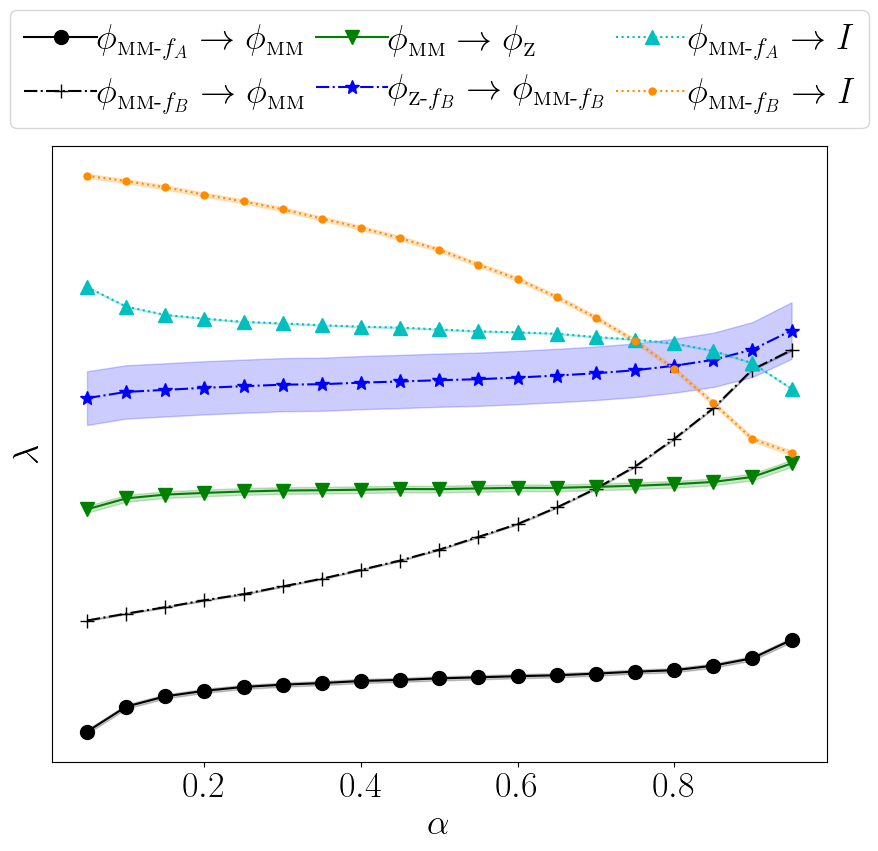}}
}
\centerline{
\subfloat[\textsc{Quora}]{
\includegraphics[trim={0 0 0 4cm},clip,width=0.36\linewidth,height=3cm]{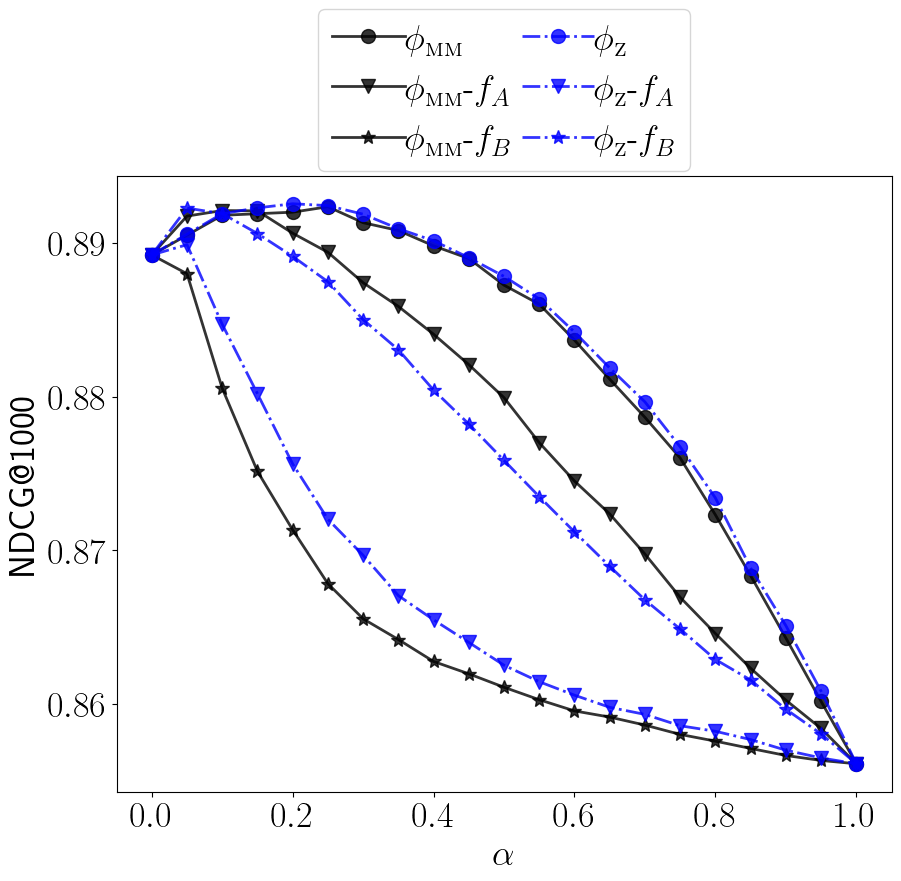}
\includegraphics[trim={0 0 0 3.6cm},clip,width=0.36\linewidth,height=3cm]{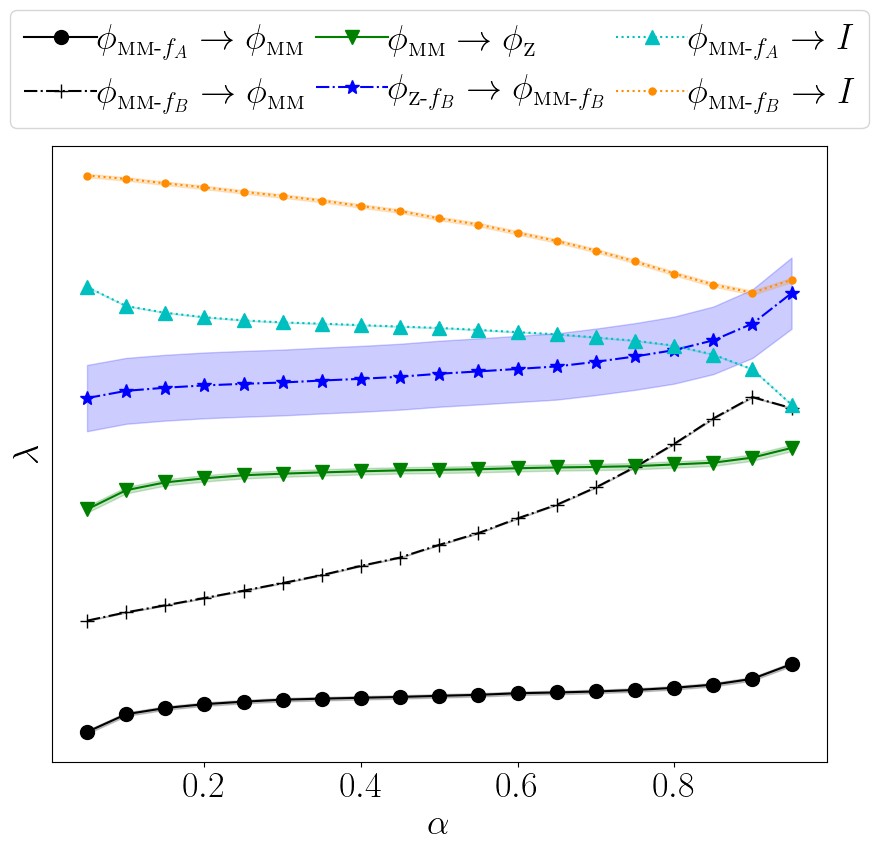}}
}
\centerline{
\subfloat[\textsc{HotpotQA}]{
\includegraphics[trim={0 0 0 4cm},clip,width=0.36\linewidth,height=3cm]{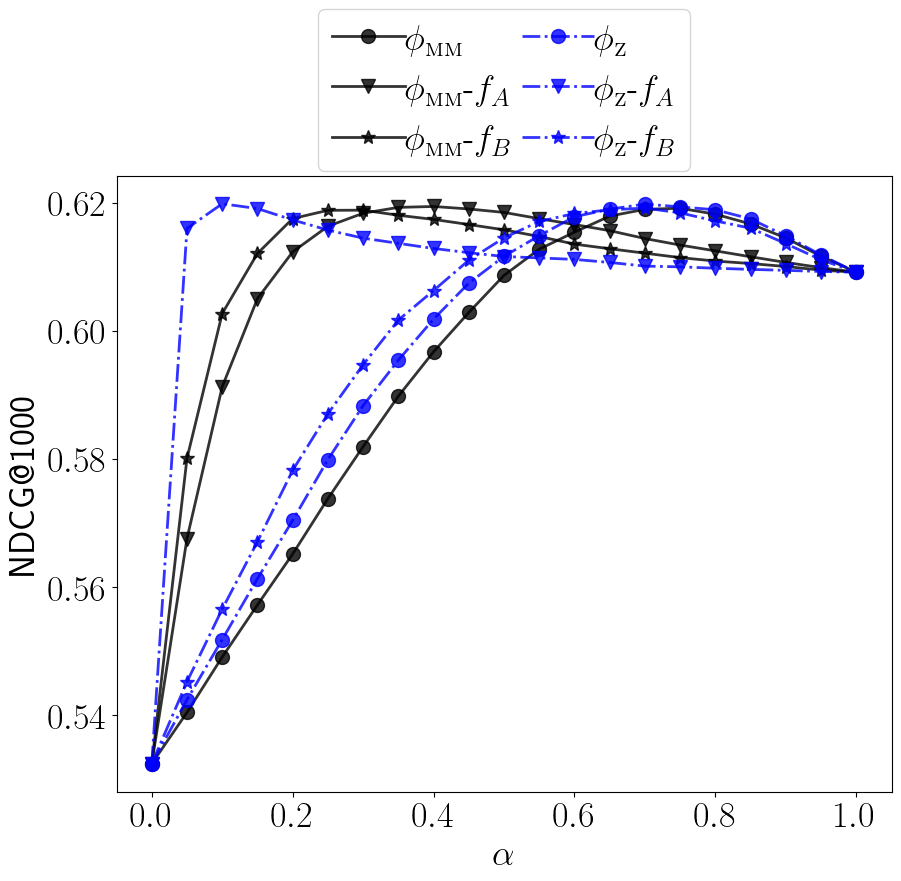}
\includegraphics[trim={0 0 0 3.6cm},clip,width=0.36\linewidth,height=3cm]{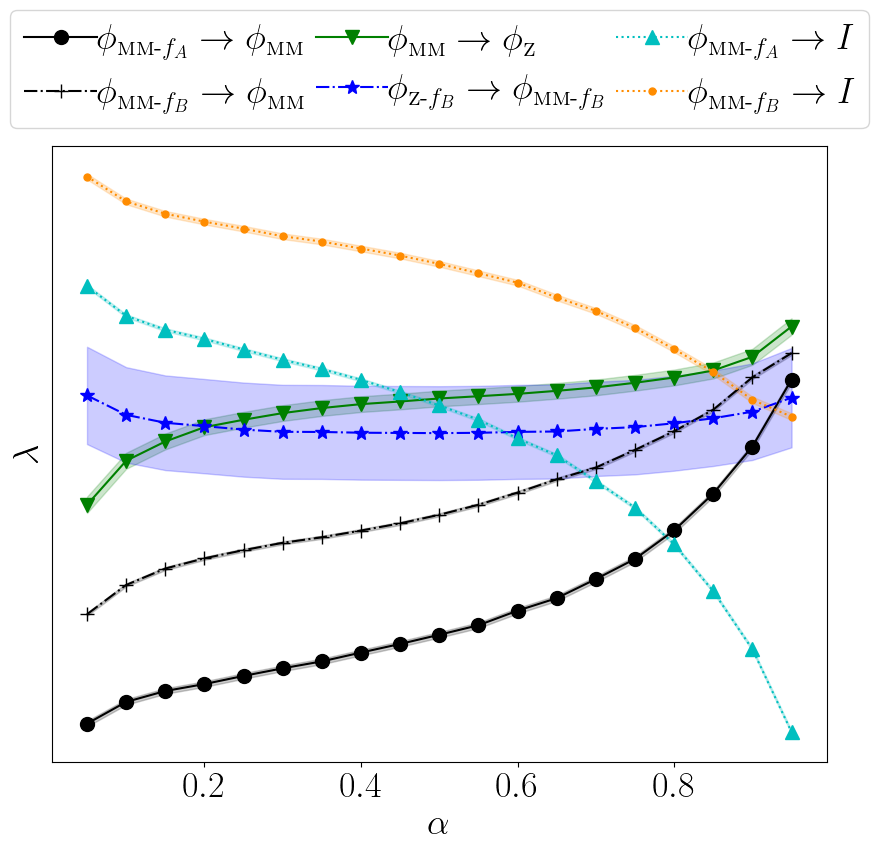}}
}
\centerline{
\subfloat[\textsc{FiQA}]{
\includegraphics[trim={0 0 0 4cm},clip,width=0.36\linewidth,height=3cm]{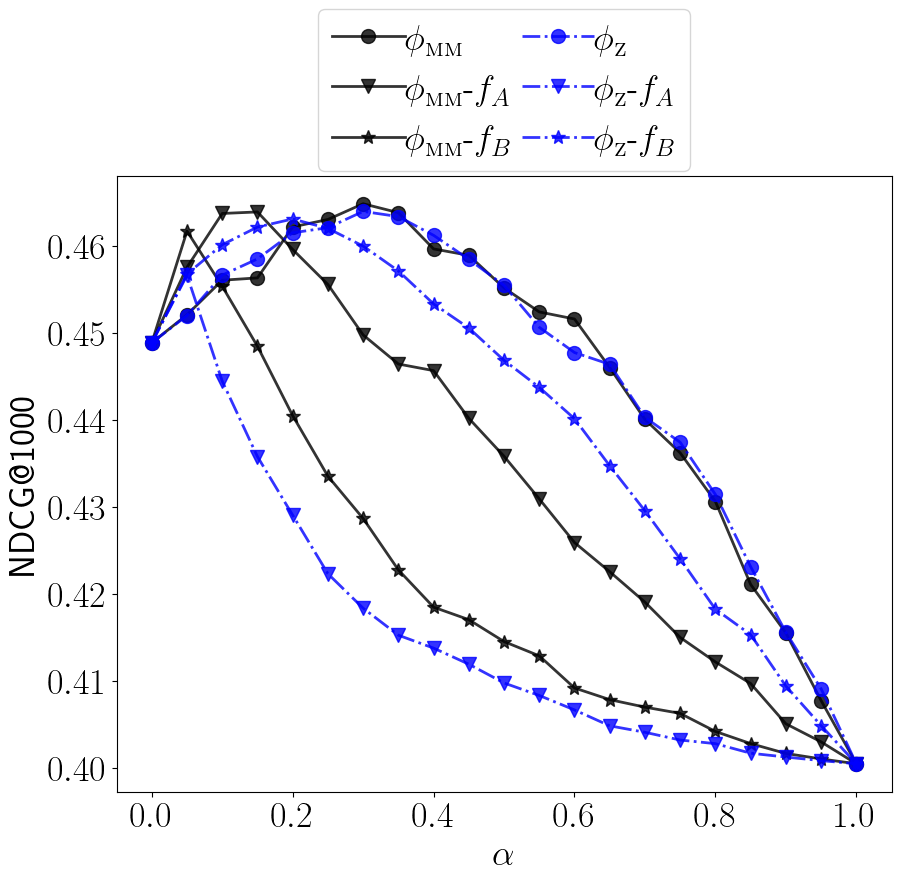}
\includegraphics[trim={0 0 0 3.6cm},clip,width=0.36\linewidth,height=3cm]{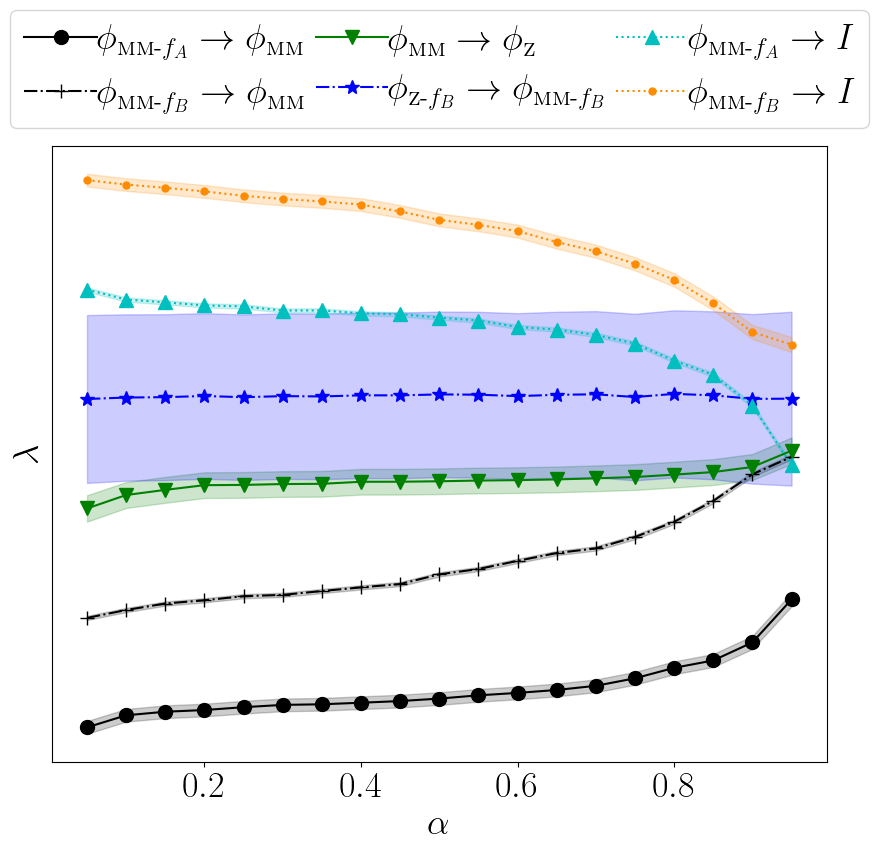}}
}
\caption{Effect of normalization on $f_\textsc{Convex} = \alpha f_A + (1 - \alpha) f_B$,
where $f_A$ is \textsc{Tas-B} and $f_B$ is \textsc{All-MiniLM-l6-v2}, as a function of $\alpha$ (left);
and, the relative expansion rate of \textsc{Tas-B} scores with respect to \textsc{All-MiniLM-l6-v2} scores (i.e., $\lambda$ in Definition~\ref{definition:growth}), with $95\%$ confidence intervals (right). $\phi\text{-}f_o$ indicates that the normalization function $\phi$ was only applied to $f_o$, with the other function entering fusion without normalization.}
\label{appendix:tasb-minilm:figure:convex:normalization}
\end{center}
\end{figure}
\FloatBarrier

\begin{table*}[t]
\caption{NDCG@1000 (except \textsc{SciFact} and \textsc{NFCorpus} where cutoff is $100$) on the test split of various datasets for individual systems and their fusion using \ac{rrf}~\cite{chen2022ecir} ($\eta=60$) and M2C2 ($\alpha=0.8$ in $\alpha \textsc{Tas-B} + (1 - \alpha) \textsc{All-MiniLM-l6-v2}$). The symbols $\ddagger$ and $\ast$ indicate statistical significance ($p$-value $< 0.01$) with respect to M2C2 and RRF respectively, according to a paired two-tailed $t$-test.}
\label{appendix:tasb-minilm:table:rq1}
\begin{center}
\begin{sc}
\begin{tabular}{c|cccc}
& \multicolumn{4}{c}{NDCG} \\
\toprule
\small{Dataset} & \small{All-MiniLM-l6-v2} & \small{Tas-B} & \small{M2C2} & \small{\ac{rrf}} \\
\midrule
\small{MS MARCO} & $0.441^{\ddagger\ast}$ & $0.477^{\ddagger}$ & \textbf{0.483} & $0.474^{\ddagger}$ \\
\midrule
\small{NQ} & $0.505^{\ddagger\ast}$ & $\textbf{0.522}^{\ddagger\ast}$ & $0.547$ & $\textbf{0.550}$ \\
\small{Quora} & $\textbf{0.889}^{\ddagger\ast}$ & $0.856^{\ddagger}$ & $0.872$ & $0.877$ \\
\small{NFCorpus} & $0.296^{\ddagger\ast}$ & $0.292^{\ddagger\ast}$ & 0.307 & $\textbf{0.316}^{\ddagger}$ \\
\small{HotpotQA} & $0.520^{\ddagger\ast}$ & $0.631^{\ddagger\ast}$ & \textbf{0.646} & $0.610^{\ddagger}$ \\
\small{FEVER} & $0.558^{\ddagger\ast}$ & $0.725^{\ast}$ & \textbf{0.730} & $0.669^{\ddagger}$ \\
\small{SciFact} & $0.681^{\ddagger\ast}$ & $0.670^{\ddagger\ast}$ & $0.704$ & $\textbf{0.712}$ \\
\small{DBPedia} & $0.425^{\ddagger\ast}$ & $0.495^{\ddagger}$ & \textbf{0.510} & $0.496^{\ddagger}$ \\
\small{FiQA} & $\textbf{0.467}^{\ddagger}$ & $0.396^{\ddagger\ast}$ & $0.438$ & $0.467^{\ddagger}$ \\
\bottomrule
\end{tabular}
\end{sc}
\end{center}
\end{table*}

\begin{figure}[t]
\begin{center}
\centerline{
\subfloat[in-domain]{
\includegraphics[height=2in]{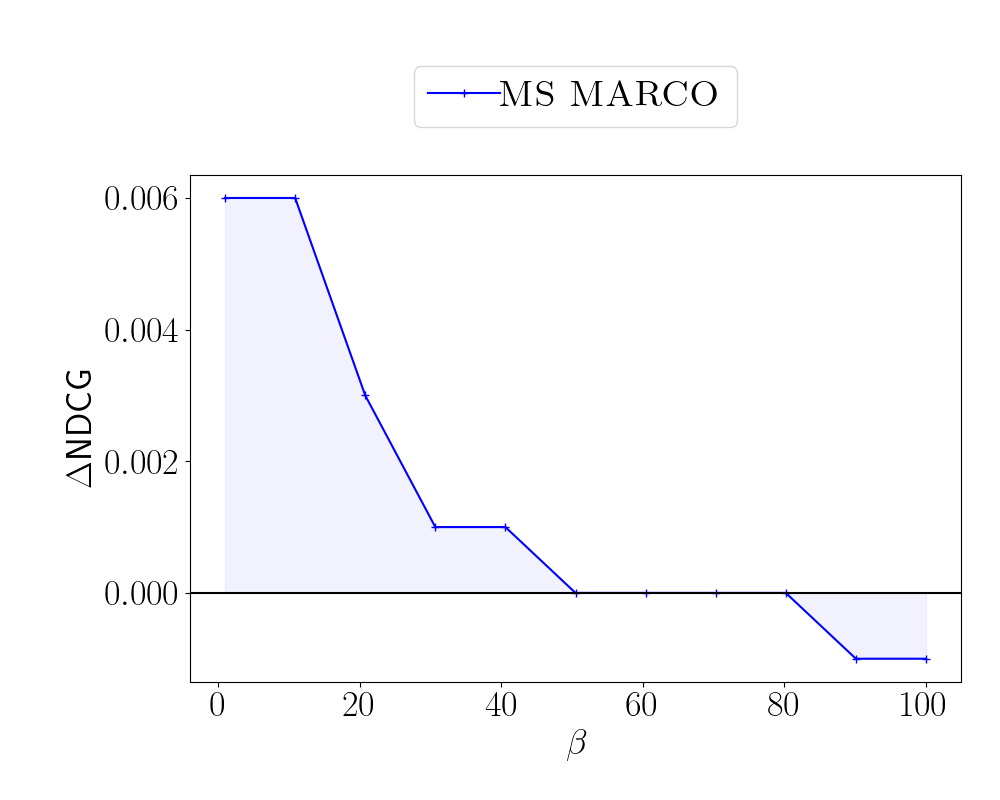}}
\subfloat[out-of-domain]{
\includegraphics[height=2in]{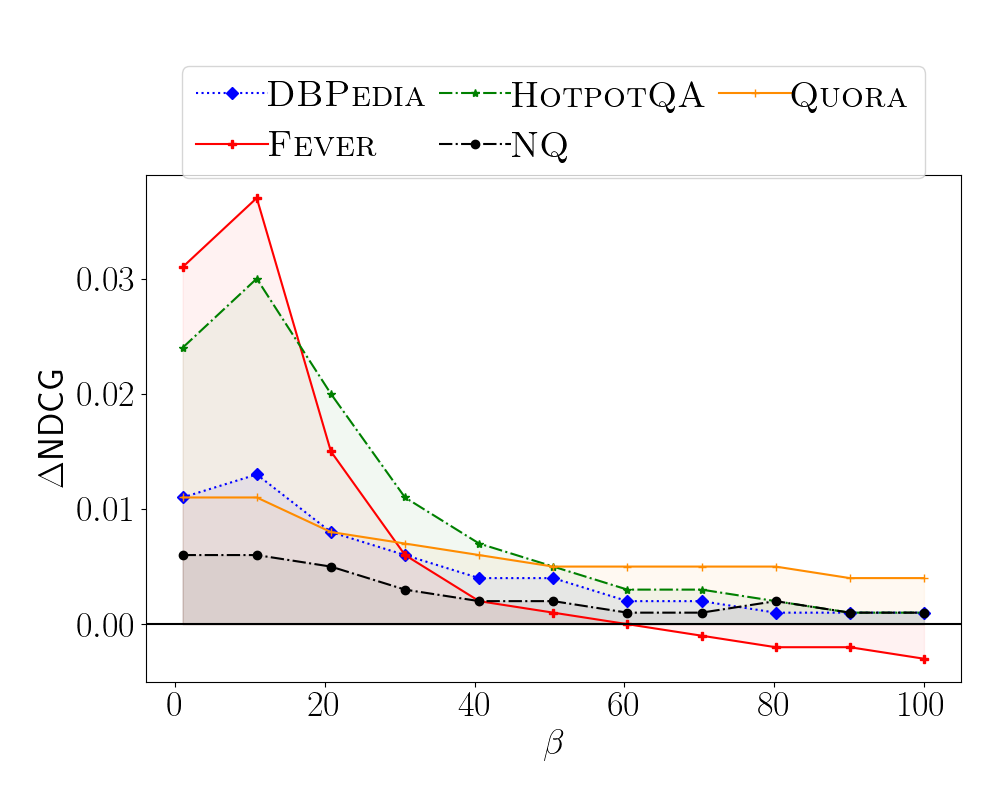}}
}
\caption{The difference in NDCG@1000 of $f_\textsc{SRRF}$ and $f_\textsc{RRF}$ with $\eta=60$ (positive indicates better ranking quality by \textsc{SRRF}) as a function of $\beta$.}
\label{appendix:tasb-minilm:figure:lipschitz}
\end{center}
\end{figure}

\FloatBarrier

\end{document}